\documentclass{acmart} 
\usepackage{xpatch}
\makeatletter
\xpatchcmd{\ps@firstpagestyle}{Manuscript submitted to ACM}{}{\typeout{success}}{\typeout{fail}}
\@ACM@manuscriptfalse
\makeatother
\usepackage{stmaryrd}  
  \SetSymbolFont{stmry}{bold}{U}{stmry}{m}{n} 
\usepackage{wasysym}   
\usepackage{xspace}
\usepackage{rotating}
\usepackage{paralist}  
\usepackage{listings}
\usepackage{newfloat}
  \DeclareFloatingEnvironment[within=none,fileext=lop]{listing}[Listing]
\usepackage{combelow}  
\usepackage{subfigure}

\usepackage{tikz}
  \usetikzlibrary{calc,shapes,arrows}

\newcommand{\citeetal}[1]{\citeauthor{#1}~\citeyearpar{#1}}

\renewcommand{\phi}{\varphi}
\renewcommand{\rho}{\varrho}

\newcommand{\set}[1]{\{#1\}}
\newcommand{\setx}[2]{\{#1\mathbin{\,|\,}#2\}}

\newcommand{\emap}{[\,]} 

\newcommand{\Nat}{\mathbb{N}}

\newcommand{\Rat}{\mathbb{Q}}
\newcommand{\Qpos}{\Rat_{\geq 0}}

\newcommand{\pto}{\nrightarrow}  
\newcommand{\pdef}{\operatorname{def}}

\newcommand{\Observations}{\mathit{Obs}}
\newcommand{\Val}{\mathit{val}}

\newcommand{\timestamp}{\operatorname{ts}}
\newcommand{\positions}{\mathit{pos}}
\newcommand{\istp}{\mathrm{tp}}
\newcommand{\tc}{\mathrm{mc}}

\newcommand{\tpp}{\mathit{tp}}
\newcommand{\tcp}{\mathit{mc}}
\newcommand{\pos}{\positions}
\newcommand{\sub}{\mathit{sub}}

\newcommand{\fPhi}[2]{\Phi_{#1}^{#2}}
\newcommand{\fPsi}[2]{\Psi_{#1}^{#2}}

\newcommand{\MTL}{\textup{MTL}}
\newcommand{\MTLdata}{{\MTL$^{\downarrow}$}\xspace}

\newcommand{\true}{\mathsf{t}}
\newcommand{\false}{\mathsf{f}}
\newcommand{\unknown}{\bot}

\newcommand{\freezequantifier}{\downarrow}
\newcommand{\freeze}[2]{\mathop{\freezequantifier^{\!#1}}{\!#2}\mathbin{\!.}}
\newcommand{\freezeshort}[1]{\mathop{\freezequantifier}{\!#1}\mathbin{\!.}}

\newcommand{\successor}{\operatorname{\fullmoon}}
\newcommand{\predecessor}{\operatorname{\newmoon}}

\newcommand{\since}{\mathbin{\mathsf{S}}}
\newcommand{\historically}{\operatorname{\text{\raisebox{-.03em}{\scalebox{1.2}{$\blacksquare$}}}}}
\newcommand{\once}{\operatorname{\text{\raisebox{-.15em}{\begin{turn}{45}\scalebox{1.1}{$\blacksquare$}\end{turn}}}}}

\newcommand{\until}{\mathbin{\mathsf{U}}}
\newcommand{\always}{\operatorname{\text{\raisebox{-.01em}{\scalebox{1.1}{$\Square$}}}}}
\newcommand{\eventually}{\operatorname{\text{\raisebox{-.15em}{\begin{turn}{45}\scalebox{1}{$\Square$}\end{turn}}}}}

\newcommand{\weakuntil}{\mathbin{\mathsf{W}}}

\newcommand{\Two}{\mathsf{2}}
\newcommand{\Three}{\mathsf{3}}

\newcommand{\omodels}{\mathrel|\joinrel\approx}

\newcommand{\sem}[2]{\llbracket{#1}\models{#2}\rrbracket}   
\newcommand{\esem}[2]{[{#1}\models{#2}]}                    
\newcommand{\osem}[2]{\llbracket{#1}\omodels{#2}\rrbracket} 
\newcommand{\eosem}[2]{[{#1}\omodels{#2}]}                  

\lstdefinelanguage{myML}{
  backgroundcolor = \color{lightgray!20!},
  firstnumber = 1,
  morekeywords={let,if,then,else,otherwise,and,or,not,return,switch,case,fun,raise,loop,for,each,foreach,do,with,while,proc,match,until,default,procedure},
  columns=fullflexible,
  sensitive=true,
  commentstyle = \itshape\color{blue},
  morecomment={[l]\#},
  morecomment={[l]//},  
  mathescape=true,
  basicstyle=\scriptsize,
  identifierstyle={\sffamily},
  escapechar = \&
}
\newcommand{\ls}{\lstinline[basicstyle=\normalsize]}

\DeclareFontFamily{U}{mathb}{\hyphenchar\font45}
\DeclareFontShape{U}{mathb}{m}{n}{
      <5> <6> <7> <8> <9> <10> gen * mathb
      <10.95> mathb10 <12> <14.4> <17.28> <20.74> <24.88> mathb12
      }{}
\DeclareSymbolFont{mathb}{U}{mathb}{m}{n}

\DeclareMathSymbol{\sqsubset}{3}{mathb}{"80}
\DeclareMathSymbol{\sqsubseteq}{3}{mathb}{"84}
\DeclareMathSymbol{\sqsubsetneq}{3}{mathb}{"88}

\newcommand{\exampleendmark}{\hfill$\lhd$}

\allowdisplaybreaks

\begin{document}

\setcopyright{acmlicensed}
\acmJournal{TOCL}
\acmVolume{21}
\acmNumber{1}
\acmArticle{5}
\acmYear{2019}
\acmMonth{9}
\acmPrice{15.00}
\acmDOI{10.1145/3355609}

\theoremstyle{acmdefinition}
\newtheorem{remark}[theorem]{Remark}

\title{Runtime Verification over Out-of-order Streams}
\titlenote{Parts of the work described in this paper have been
  previously published in the conference papers
  \cite{Basin_etal:failureaware_rv} and \cite{Basin_etal:outoforder}.
  \\
  This is the authors' version of the work. It is posted here for your
  personal use.  Not for redistribution.  The definitive version was
  published in the ACM Transactions on Computational Logic,
  https://doi.org/10.1145/3355609.}

\author{David Basin}
\affiliation{%
  \institution{ETH Z\"urich}
  \department{Department of Computer Science}
  \streetaddress{Universit\"atstrasse 6}
  \postcode{8092}
  \city{Zurich}
  \country{Switzerland}}
\email{basin@inf.ethz.ch}

\author{Felix Klaedtke}
\affiliation{%
  \institution{NEC Laboratories Europe GmbH}
  \streetaddress{Kurf\"ursten-Anlage 36}
  \postcode{69115}
  \city{Heidelberg}
  \country{Germany}}
\email{felix.klaedtke@neclab.eu}
\orcid{0000-0002-7572-1158}

\author{Eugen Z\u{a}linescu}
\affiliation{%
  \institution{Technische Universit\"at M\"unchen}
  \department{Institut f\"ur Informatik}
  \streetaddress{Boltzmanstra\ss{}e 3}
  \postcode{85748}
  \city{Garching}
  \country{Germany}}
\email{eugen.zalinescu@in.tum.de}
\orcid{0000-0002-2312-5561}

\begin{abstract}
  We present an approach for verifying systems at runtime.  Our
  approach targets distributed systems whose components communicate
  with monitors over unreliable channels, where messages can be
  delayed, reordered, or even lost.  Furthermore, our approach handles
  an expressive specification language that extends the real-time
  logic MTL with freeze quantifiers for reasoning about data values.
  The logic's main novelty is a new three-valued semantics that is
  well suited for runtime verification as it accounts for partial
  knowledge about a system's behavior. Based on this semantics, we
  present online algorithms that reason soundly and completely about
  streams where events can occur out of order.  We also evaluate our
  algorithms experimentally.  Depending on the specification, our
  prototype implementation scales to out-of-order streams with
  hundreds to thousands of events per second.
\end{abstract}



\maketitle
\renewcommand{\shortauthors}{Basin et al.}

\section{Introduction}
\label{sec:intro}

Distributed systems are omnipresent and complex, and they can
malfunction for many reasons, including software bugs and hardware or
network failures.  Monitoring is an attractive option for verifying at
runtime whether a system behavior is correct with respect to a given
specification.  But distribution opens new challenges.  The monitors
themselves become components of the (extended) system and like any
other system component they may exhibit delays, finite or even
infinite, when communicating with other components.

Various runtime-verification approaches exist for different kinds of
systems, including distributed
systems~\cite{Barringer_etal:eagle,Bauer_etal:rv_tltl,MN04-signals,Meredith_etal:mop,Basin_etal:rv_mfotl,
  Sen_etal:decentralized_distributed_monitoring,Bauer_Falcone:decentralised_monitor,Falcone_etal:decentralized_monitor_regular,Mostafa_Bonakdarbour:decentralized_rv}.
The specification languages used in these approaches are typically
based on temporal logics or finite-state machines, which describe the
correct system behavior in terms of \emph{infinite} streams of system
actions.  However, at any point in time, a monitor has only partial
knowledge about the system's behavior.  In particular, a monitor can
at best only be aware of the actions the system performed so far,
which correspond to a finite prefix of the infinite action stream.
For this reason, many of the runtime-verification approaches rely on
an extension of the standard Boolean semantics of the linear-time
temporal logic LTL with a third truth value, as proposed by
\citeetal{Bauer_etal:ltl_rv}.  Namely, an LTL formula evaluates to the
Boolean truth value~$b$ on a finite stream of actions~$\sigma$ if the
formula evaluates to $b$ on all infinite streams that extend~$\sigma$;
otherwise, the formula's truth value is unknown on $\sigma$.

This three-valued semantics, however, only accounts for settings where
monitors are always aware of all previously performed actions.  It is
insufficient to reason soundly and completely about system behavior at
runtime when, for example, unreliable channels are used to inform the
monitors about the actions performed.  In fact, the existing
runtime-verification approaches are of limited use for distributed
systems where components might crash or network failures occur, for
example, when a component is temporarily unreachable and a monitor
therefore cannot learn the component's behavior during this time
period.  Even in the absence of failures, monitors can receive
messages about the system behavior in any order due to network delays.
A naive solution for coping with out-of-order message delivery is to
have the monitor buffer messages and reorder them prior to processing
them.  However, this can delay reporting a violation when the
violation is already detectable on some of the buffered messages.
This is undesirable for applications where one cannot afford to wait
and the monitor should promptly output its verdict. Moreover, the
verdict should remain correct when some of the monitor's knowledge
gaps are subsequently closed.
Another limitation concerns the expressivity of the specification
languages used by the existing runtime-verification approaches for
distributed systems.  It is not possible to express real-time
constraints, which are common requirements for distributed
systems. Such constraints specify, for example, deadlines to be met.
Furthermore, the supported specification languages cannot handle data
values.

In this paper, we present a runtime-verification approach that
overcomes these limitations.  Our approach handles specifications that
are given as formulas in an extension of the real-time logic
MTL~\cite{Koymans:realtime_properties,AlurHenzinger:realtimelogics_survey}.
Namely, we extend MTL with a freeze
quantifier~\cite{Henzinger:half_order} to extract data values from
events and bind these values to logical variables.  We call this
extension \MTLdata (pronounced ``MTL freeze''), where $\downarrow$ is
the symbol for the freeze quantifier.  Our runtime-verification
approach accounts for out-of-order message deliveries and soundly
operates in the presence of failures, such as components crashing.  We
also provide completeness guarantees for our approach, roughly meaning
that in the absence of failures but with arbitrary finite message
delays, violations and satisfactions of specifications are eventually
reported.
We build upon a timed model for distributed
systems~\cite{Cristian_Fetzer:timed_model}.  The system components use
their local clocks to timestamp observations, which they send to the
monitors. The monitors use these timestamps to determine the elapsed
time between observations, for example, to check whether real-time
constraints are met.  Furthermore, the timestamps totally order the
observations.  This is in contrast to a time-free
model~\cite{Fischer_etal:consensus}, where the events of a distributed
system can only be partially ordered, for example, using Lamport
timestamps~\cite{Lamport:clocks}.  However, since the accuracy of
existing clocks is limited, the monitors' conclusions might only be
valid for the provided timestamps. See Section~\ref{subsec:discussion},
where we elaborate on this point.

A cornerstone of our monitoring approach is a new three-valued
semantics for \MTLdata that is well suited to reason in settings where
system components communicate with the monitors over unreliable
channels.  Specifically, we define \MTLdata's semantics over the three
truth values~$\true$,~$\false$, and~$\unknown$.  We interpret these
truth values as in Kleene logic~\cite{Kleene:1950} and conservatively
extend the logic's standard Boolean semantics, where $\true$ and
$\false$ stand for ``true'' and ``false,'' respectively, and the third
truth value~$\unknown$ stands for ``unknown'' and accounts for the
monitor's knowledge gaps.  The models of \MTLdata are finite words
where knowledge gaps are explicitly represented.  Intuitively, a
finite word corresponds to a monitor's knowledge about the system
behavior at a given time and the knowledge gaps may result from
message delays, losses, crashed components, and the like.  Critically
in our setting, reasoning is monotonic with respect to the partial
order on truth values, where $\unknown$ is less than $\true$ and
$\false$, and $\true$ and $\false$ are incomparable.  This
monotonicity property guarantees that closing knowledge gaps does not
invalidate previously obtained Boolean truth values.

We also present online algorithms for verifying systems at runtime
with respect to \MTLdata specifications.  Our algorithms' output is
sound and complete for \MTLdata's three-valued semantics and with
respect to the monitor's partial knowledge about the actions performed
at each point in time.  In a nutshell, the algorithms work as follows.
They receive as input timestamped messages from the system components,
which describe the actions these components perform.  No assumptions
are made on the order in which these messages are received.
The algorithms update their state for each received message.  This
state comprises an acyclic graph structure for reasoning about the
system behavior, that is, computing verdicts about the monitored
specification's fulfillment.  The graph's nodes store the truth values
of the subformulas for the different times that data values are frozen
to quantified variables, including the times with no or only partial
knowledge.  The graph is refined when the monitor receives knowledge
about a specific point in time, whereby the nodes representing the
knowledge gap are split and instantiated.  In each such update, the
algorithms propagate data values down to the graph's leaves and
propagate Boolean truth values for subformulas up along the graph's
edges.  When a Boolean truth value is propagated to a root node of the
graph, the algorithms output a verdict.

Overall, our main contributions are as follows.  First, we define a
new three-valued semantics for a temporal logic, which is well suited
for runtime verification, in particular, for reasoning about
incomplete traces. Second, we present online algorithms to reason
soundly and completely about incomplete traces.  Moreover, these
algorithms output verdicts promptly. Third, we experimentally evaluate
the performance of our algorithms and explore the performance impact
on handling messages that arrive out of order.  Finally, we describe
the deployment of our online algorithms for verifying distributed
systems at runtime.

The remainder of this paper is structured as follows.  In
Section~\ref{sec:prelim}, we provide preliminaries.  In
Section~\ref{sec:mtl}, we present our new three-valued semantics for
monitoring. In Sections~\ref{sec:prop} and~\ref{sec:data}, we
present our monitoring algorithms, including a proof of their
correctness.  We evaluate our algorithms in Section~\ref{sec:eval}.
In Section~\ref{sec:app}, we describe the deployment of our
runtime-verification approach for distributed systems.  Finally, in
Sections~\ref{sec:related} and~\ref{sec:concl}, we discuss related
work and draw conclusions.

\section{Preliminaries}
\label{sec:prelim}

In this section, we recall standard notation and terminology that will
be used throughout the paper.

\subsubsection*{Intervals.}

An \emph{interval} $I$ is a nonempty subset of the positive
rationals~$\Qpos$ such that if $a,b\in I$ and $a\leq c\leq b$ then
$c\in I$, for all $a,b,c\in\Qpos$. We use standard notation and
terminology for intervals. For example, $(a,b]$ denotes the interval
that is left-open with bound~$a$ and right-closed with bound~$b$.
Note that an interval $I$ with cardinality $|I|=1$ is a singleton
$I=\set{\tau}=[\tau,\tau]$, for some $\tau\in\Qpos$.  An interval~$I$
is \emph{unbounded} if its right bound is~$\infty$, and \emph{bounded}
otherwise.  With less-than, $<$, we denote the partial order on
intervals, that is $I<J$ iff $I\cap J=\emptyset$ and $I$'s right bound
is not greater than $J$'s left bound.  Let
$I-J := \setx{\tau-\tau'}{\tau\in I\text{ and }\tau'\in J}\cap\Qpos$.

\subsubsection*{Partial Functions.}

For a partial function $f:A\pto B$, let
$\pdef(f):=\setx{a\in A}{f(a)\text{ is defined}}$.  If
$\pdef(f)=\{a_1,\dots,a_n\}$, for some $n\in\Nat$, then we also write
$[a_1\mapsto f(a_1),\dots, a_n\mapsto f(a_n)]$ for $f$, when $f$'s
domain $A$ and its codomain $B$ are irrelevant or clear from the
context.  Note that $\emap$ denotes the partial function that is
undefined everywhere.  We also carry over the notation for set
comprehension, for instance,
$[a\mapsto a+1\mathbin{|} a\geq0\text{ and $a$ is even}]$ denotes the
partial function that is defined on the nonnegative even integers and
returns their successor.
Furthermore, we write $f[a\mapsto b]$ to denote the update of a
partial function $f:A\pto B$ at $a\in A$, that is, $f[a\mapsto b]$
equals $f$, except that $a$ is mapped to $b$ if $b\in B$, and
$a\not\in\pdef(f[a\mapsto b])$ if $b\notin B$.
With $f[a\mapsto\bot]$ we denote the restriction of $f$ to the domain
$\pdef(f)\setminus\set{a}$.
Finally, for partial functions $f, g:A\pto B$, we write
$f\sqsubseteq g$ if $\pdef(f)\subseteq \pdef(g)$ and $f(a)=g(a)$, for
all $a\in\pdef(f)$.

\subsubsection*{Truth Values.}

Let $\Three$ be the set~$\set{\true,\false,\bot}$, where
$\true$~(true) and $\false$~(false) denote the standard Boolean
values, and $\unknown$ denotes the truth value ``unknown.''
Table~\ref{tab:connectives} shows the truth tables of some standard
logical operators over $\Three$.  Observe that these operators
coincide with their Boolean counterparts when restricted to the
set~$\Two:=\{\true,\false\}$.
\begin{table}[t]
  \caption{Truth tables for three-valued logical operators 
    (strong Kleene logic).}
  \label{tab:connectives}
  \centering
  \renewcommand{\arraystretch}{1}
  \begin{tabular}{c@{\ } | @{\ }c }
    $\neg$   &  \\
    \hline
    $\true$  & $\false$ \\
    $\false$ & $\true$ \\
    $\unknown$   & $\unknown$ \\
  \end{tabular}
  \qquad
  \begin{tabular}{ c@{\ } | @{\ }c  c  c }
    $\vee$   & $\true$ & $\false$ & $\unknown$ \\
    \hline
    $\true$  & $\true$ & $\true$  & $\true$\\
    $\false$ & $\true$ & $\false$ & $\unknown$\\
    $\unknown$   & $\true$ & $\unknown$   & $\unknown$\\
  \end{tabular}
  \qquad
  \begin{tabular}{ c@{\ } | @{\ }c  c  c}
    $\wedge$ & $\true$ & $\false$ & $\unknown$ \\
    \hline
    $\true$  & $\true$  & $\false$ & $\unknown$\\
    $\false$ & $\false$ & $\false$ & $\false$\\
    $\unknown$   & $\unknown$   & $\false$ & $\unknown$
  \end{tabular}
  \qquad
  \begin{tabular}{ c@{\ } | @{\ }c  c  c}
    $\rightarrow$ & $\true$ & $\false$ & $\unknown$ \\
    \hline
    $\true$  & $\true$ & $\false$ & $\unknown$  \\
    $\false$ & $\true$ & $\true$  & $\true$ \\
    $\unknown$   & $\true$ & $\unknown$   & $\unknown$  
  \end{tabular}
\end{table}
We partially order the elements in $\Three$ by their knowledge:
$\unknown\prec\true$, $\unknown\prec\false$, and $\true$ and $\false$
are incomparable as they carry the same amount of knowledge.  Note
that $(\Three,\prec)$ is a lower semilattice where $\curlywedge$
denotes the meet.
We remark that the operators in Table~\ref{tab:connectives} are
monotonic, which ensures that reasoning is monotonic in
knowledge. Intuitively, when closing a knowledge gap, represented
by~$\unknown$, with $\true$ or $\false$, we never obtain a truth value
that disagrees with the previous one.

\subsubsection*{Timed Words.}

Let $\Sigma$ be an alphabet.  A \emph{timed word} over $\Sigma$ is an
infinite
word~$(\tau_0,a_0)(\tau_1,a_1)\ldots\in(\Qpos\times\Sigma)^\omega$,
where the sequence of $\tau_i$s is strictly monotonic and nonzeno,
that is, $\tau_i<\tau_{i+1}$, for every $i\in\Nat$, and for every
$t\in\Qpos$, there is some $i\in\Nat$ such that $\tau_i>t$.  Note that
we use a dense time domain and assume a nonfictitious clock semantics,
that is, there is no stuttering of equal timestamps.

\subsubsection*{Metric Temporal Logic.}

Let $P$ be a finite set of predicate symbols, where $\iota(p)$ denotes
the arity of $p\in P$. Furthermore, let $V$ be a set of variables and
$R$ a finite set of registers.
The syntax of the real-time logic \MTLdata is given by the grammar:
\begin{equation*}
  \phi
  \mathbin{\ ::=\ }
  \true \mathbin{\,\big|\,} 
  p(x_1,\dots,x_{\iota(p)}) \mathbin{\,\big|\,}
  \freeze{r}{x}\phi \mathbin{\,\big|\,} 
  \neg\phi \mathbin{\,\big|\,}
  \phi\vee\phi \mathbin{\,\big|\,}
  \predecessor_I\phi \mathbin{\,\big|\,}
  \successor_I\phi \mathbin{\,\big|\,}
  \phi\since_I\phi \mathbin{\,\big|\,}
  \phi\until_I\phi 
  \,,
\end{equation*}
where $p\in P$, $x, x_1,x_2\dots,x_{\iota(p)}\in V$, $r\in R$, and $I$
is an interval.  We remark that \MTLdata extends the standard
propositional metric temporal logic
(MTL)~\cite{Koymans:realtime_properties,AlurHenzinger:realtimelogics_survey}
with a freeze quantifier~$\downarrow$.  We call a formula an \emph{MTL
  formula} if all the predicate symbols occurring in it have arity $0$
and the freeze quantifier does not occur in it.

A formula is \emph{closed} if each variable occurrence is bound by a
freeze quantifier.  A formula is \emph{temporal} if the connective at
the root of the formula's syntax tree is $\predecessor_I$,
$\successor_I$, $\since_I$, or $\until_I$.
We denote by~$\sub(\phi)$ the set of $\phi$'s subformulas.
We employ standard syntactic sugar. For example, $\phi\wedge\psi$
abbreviates $\neg(\neg\phi\vee\neg\psi)$, $\phi\rightarrow\psi$
abbreviates $\neg\phi\vee\psi$, and $\eventually_I\phi$
(``eventually'') and $\always_I\phi$ (``always'') abbreviate
$\true\until_I\phi$ and $\neg\eventually_I\neg\phi$, respectively.
The past-time counterparts $\once_I\phi$ (``once'') and
$\historically_I\phi$ (``historically'') are defined as expected.  The
nonmetric variants of the temporal connectives are also easily
defined, for example, $\always\phi:=\always_{[0,\infty)}\phi$.  We
also use standard conventions concerning the connectives' binding
strength to omit parentheses.  For example, $\neg$ binds stronger than
$\wedge$, which binds stronger than $\vee$, and the connectives
$\neg$, $\vee$, etc. bind stronger than the temporal connectives,
which bind stronger than the freeze quantifier.  Finally, to simplify
notation, we omit the superscript~$r$ in formulas like
$\freeze{r}{x}\phi$ whenever $r\in R$ is irrelevant or clear from the
context.

\begin{example}
  \label{ex:mtl}
  Before defining \MTLdata's semantics, we provide some intuition.
  The following formula formalizes the policy that whenever a customer
  executes a transaction that exceeds some threshold (e.g.,~\$2,000),
  then this customer must not execute any other transaction for a
  fixed time period (e.g., 3~days).
  \begin{equation*}
    \always 
    \freeze{\mathit{cid}}{c}\freeze{\mathit{tid}}{t}\freeze{\mathit{sum}}{a}
    \mathit{trans}(c,t,a)\wedge a>2000 \to
    \always_{(0,3]}\freeze{\mathit{tid}}{t'}\freeze{\mathit{sum}}{a'}
    \neg\mathit{trans}(c,t',a')
  \end{equation*}
  Note that in the formula, we take the liberty to deviate slightly
  from the given grammar, which does not include constant and function
  symbols.  Such an extension would be straightforward, but we omit it
  for the sake of brevity.  In particular, the formula contains the
  constant symbol~$2000$, interpreted as expected.  Furthermore, the
  binary predicate symbol~$>$, also with its expected rigid
  interpretation, is written in infix.

  We assume that the predicate symbol $\mathit{trans}$ is interpreted
  as a singleton relation or the empty set at any point in time.  For
  instance, the interpretation $\set{(\mathit{Alice},42,99)}$ of
  $\mathit{trans}$ at time~$\tau$ describes the action of
  $\mathit{Alice}$ executing a transaction with identifier~$42$ with
  the amount \$99 at time~$\tau$.  When the interpretation is the
  empty set, no transaction is executed.
  We further assume that when the interpretation of the predicate
  symbol $\mathit{trans}$ is nonempty, the registers $\mathit{cid}$,
  $\mathit{tid}$, and $\mathit{sum}$ store (a)~the transaction's
  customer, (b)~the transaction identifier, and (c)~the transferred
  amount, respectively.  If the interpretation is the empty set, then
  the registers store a dummy value, representing undefinedness.

  The variables $c$, $t$, $a$, $t'$, and~$a'$ are frozen to the
  respective register values.  For example, $c$ is frozen to the value
  stored in the register~$\mathit{cid}$ at each point in time and is
  used to identify subsequent transactions from this customer.  Also
  note that, for instance, the variables $t$ and~$t'$ are frozen to
  values stored in the registers~$\mathit{tid}$ at different times.
  The freeze quantifier can be seen as a weak form of the standard
  first-order quantifiers~\cite{Henzinger:half_order}. Since each
  register stores exactly one value at any time, it is irrelevant
  whether we quantify existentially or universally over a register's
  value.  \exampleendmark
\end{example}

Let $D$---the \emph{data domain}---be a nonempty set of values.
Furthermore, let $\Sigma$ be the set of the pairs $(\sigma,\rho)$,
where $\sigma$ is a function over $P$ with
$\sigma(p)\subseteq D^{\iota(p)}$ for $p\in P$ and $\rho$ is a
function over $R$ with $\rho(r)\in D$ for $r\in R$.  Intuitively,
$\sigma$ interprets the predicate symbols at the given time point and
$\rho$ provides the values of the registers in $R$.
\MTLdata's Boolean semantics is defined inductively over the formula
structure.  We define a function
$\phi\mapsto\sem{w,i,\nu}{\phi}\in\Two$, for a given timed word~$w$
over $\Sigma$, $i\in\Nat$, and a valuation~$\nu:V\to D$.  Let
$w=\big(\tau_0,(\sigma_0,\rho_0)\big)\,\big(\tau_1,(\sigma_1,\rho_1)\big)\dots$.
\begin{align*}
    \sem{w,i,\nu}{\true}
    :=\ &
    \true
    \\
    \sem{w,i,\nu}{p(\overline{x})}
    :=&
    \begin{cases}
      \true & \text{if $\nu(\overline{x})\in\sigma_i(p)$}
      \\
      \false & \text{otherwise}
    \end{cases}
    \\
    \sem{w,i,\nu}{\freeze{r}{x}\phi}
    :=\ &
    \sem{w,i,\nu[x\mapsto \rho_i(r)]}{\phi}
    \\[.1cm] 
    \sem{w,i,\nu}{\neg\phi}
    :=\ &
    \neg\sem{w,i,\nu}{\phi}
    \\
    \sem{w,i,\nu}{\phi\vee\psi}
    :=\ &
    \sem{w,i,\nu}{\phi}
    \vee
    \sem{w,i,\nu}{\psi}
    \\
    \sem{w,i,\nu}{\predecessor_I\phi}
    :=\ &
    i>0\wedge\tau_i-\tau_{i-1}\in I\wedge \sem{w,i-1,\nu}{\phi}
    \\
    \sem{w,i,\nu}{\successor_I\phi}
    :=\ &
    \tau_{i+1}-\tau_i\in I\wedge \sem{w,i+1,\nu}{\phi}
    \\
    \sem{w,i,\nu}{\phi\since_I\psi}
    :=\ &
    \bigvee_{j\in\Nat, j\leq i}
      \big(
      \tau_i-\tau_j\in I \wedge
      \sem{w,j,\nu}{\psi} \wedge
      \bigwedge_{j<k\leq i}\sem{w,k,\nu}{\phi}
      \big)
    \\
    \sem{w,i,\nu}{\phi\until_I\psi}
    :=\ &
    \bigvee_{j\in\Nat, j\geq i}
      \big(
      \tau_j-\tau_i\in I \wedge
      \sem{w,j,\nu}{\psi} \wedge
      \bigwedge_{i\leq k<j}\sem{w,k,\nu}{\phi}
      \big)
\end{align*}
Note that we abuse notation here and identify the logic's constant
symbol~$\true$ with the Boolean value $\true$, and the
connectives~$\neg$ and $\vee$ with the corresponding logical
operators.  Furthermore, we use standard conventions, for example,
$p(\overline{x})$ abbreviates $p(x_1,\dots,x_{\iota(p)})$ and
$\nu(\overline{x})\in\sigma_i(p)$ abbreviates
$\big(\nu(x_1),\dots,\nu(x_{\iota(p)})\big)\in\sigma_i(p)$.  Finally,
note that the disjunction in the $\until_I$ case is infinite.

\section{Metric Temporal Logic for Monitoring}
\label{sec:mtl}

In this section, we present a three-valued semantics for \MTLdata that
conservatively approximates the logic's standard Boolean semantics.
Our new semantics is defined with monitoring in mind in that it
accounts for knowledge gaps that arise during monitoring, which may be
fully or partially filled later.  We first introduce in
Section~\ref{subsec:observations} the models of our semantics, which
support reasoning about incomplete traces.  Afterwards, in
Sections~\ref{subsec:semantics} and ~\ref{subsec:properties}, we present the semantics and establish
basic properties about it.  We
conclude by defining correctness requirements for monitoring in
Section~\ref{subsec:requirements}.

\subsection{Observations}
\label{subsec:observations}

A monitor usually has only partial knowledge about the behavior of the
system it monitors.  For instance, for nonterminating systems, a
monitor is only aware of a finite prefix of the system's behavior.
Thus, when modeling this behavior as a timed word, the monitor only
knows a finite prefix of this word.  Moreover, when communication to
the monitor is unreliable or delayed, the monitor may not even have
the entire finite prefix, but only portions thereof.  In the
following, we introduce a notion of observations that supports
reasoning based on partial information about the system behavior.

Throughout this section, we fix an alphabet~$\Sigma$.  We require that
$\Sigma$ is partially ordered and denote the partial order
by~$\sqsubset$.  Intuitively, $a\sqsubset b$ means that $a$ carries
less information than $b$.  Furthermore, we require that $\Sigma$ has
a least element~$a_0$.
\begin{definition}
  \label{def:observation}
  The set of \emph{observations} $\Observations(\Sigma)$ is inductively
  defined.
  \begin{itemize}[--]
  \item The word $\big([0,\infty),a_0\big)$ of length~$1$ is
    in $\Observations(\Sigma)$.
  \item If the word $w$ is in $\Observations(\Sigma)$, then the word
    obtained by applying one of the following transformations to $w$
    is in $\Observations(\Sigma)$.
    \begin{enumerate}[(T1)]
    \item \label{enum:observation_split} Some letter $(I,a)$ of $w$,
      with $|I|>1$, is replaced by the three-letter word
      \begin{equation*}
        \big(I\cap [0,\tau),a\big)\ 
        \big(\set{\tau},a\big)\  
        \big(I\cap(\tau,\infty),a\big)
        \,,
      \end{equation*}
      where $\tau\in I$ and $\tau>0$.  If $\tau=0$, then
      $(I,a)$ is replaced by the two-letter word
      $\big(\set{\tau},a\big)\,\big(I\cap(\tau,\infty),a\big)$.
    \item \label{enum:observation_removal} Some letter
      $(I,a)$ of $w$, with $|I|>1$ and $I$ bounded, is removed.
    \item \label{enum:observation_data} Some letter $(I,a)$ of $w$,
      with $|I|=1$, is replaced by $(I,a')$ with $a\sqsubset a'$.
    \end{enumerate}
  \end{itemize}
\end{definition}
For an observation $w$ of length $n\in\Nat$, let
$\positions(w):=\{0,\dots,n-1\}$.  We call $i\in\positions(w)$ a
\emph{time point} in $w$ if the interval~$I_i$ of the letter at
position~$i$ in $w$ is a singleton.  In this case, the element of
$I_i$ is the \emph{timestamp} of the time point~$i$, denoted by
$\timestamp_w(i)$.

Given the inductive definition of the set $\Observations(\Sigma)$, the
partial order over $\Sigma$ naturally extends to a partial order on
observations. We thereby obtain the following refinement relation on
observations.
\begin{definition}
  For $w,w'\in\Observations(\Sigma)$, let $w\sqsubset_1 w'$ iff $w'$
  is obtained from $w$ by one of the transformations
  (T\ref{enum:observation_split}), (T\ref{enum:observation_removal}),
  or (T\ref{enum:observation_data}).
  The observation $w'$ \emph{refines} the observation $w$ if
  $w\sqsubseteq w'$, where $\sqsubseteq$ is the reflexive-transitive
  closure of $\sqsubset_1$.
\end{definition}

\begin{example}
  \label{ex:trans-observation}
  Recall the set of predicates symbols $P=\{\mathit{trans}\}$ and the
  set of registers $R=\{\mathit{cid},\mathit{tid},\mathit{sum}\}$ from
  Example~\ref{ex:mtl}. For brevity, we ignore here the rigid
  interpretations of the constant symbol $2000$ and the binary
  predicate symbol $\geq$.  Furthermore, recall the data domain $D$
  that contains all customers and the positive integers.  Let $\Sigma$
  be the alphabet consisting of the pairs $(\sigma,\rho)$ with
  $\sigma:P\pto 2^{D\times D\times D}$ and $\rho:R\pto D$.  Note that
  the partial orders on the two sets of partial functions extend to a
  partial order on $\Sigma$ and that $(\emap,\emap)$ is $\Sigma$'s
  least element.

  A monitor's knowledge can be represented by observations over
  $\Sigma$.
  A monitor's initial knowledge is represented by the observation
  $w_0 = \big([0,\infty),(\emap,\emap)\big)$.
  Suppose that a transaction of $\$99$ with identifier~$42$ from
  $\mathit{Alice}$ is executed at time~$3.0$.
  The monitor's initial knowledge $w_0$ is then updated by the
  transformations~(T\ref{enum:observation_split})
  and~(T\ref{enum:observation_data}) to
  $w_1 = \big([0,3.0),(\emap,\emap)\big)\,
  \big(\{3.0\},(\sigma,\rho)\big)\,\big((3.0,\infty),(\emap,\emap)\big)$,
  where $\sigma(\mathit{trans}) = \set{(\mathit{Alice}, 42, 99)}$ and
  $\rho = [\mathit{cid}\mapsto\mathit{Alice}, \mathit{tid}\mapsto 42,
  \mathit{sum}\mapsto 99]$. Note that $w_0\sqsubseteq w_1$.
  
  If the monitor also receives the information that no action occurred
  in the interval~$[0,3.0)$, then its updated knowledge is represented
  by the
  observation~$\big(\{3.0\},(\sigma,\rho)\big)\,
  \big((3.0,\infty),(\emap,\emap)\big)$, obtained from $w_1$ by the
  transformation~(T\ref{enum:observation_removal}).
  The information that no action has occurred in an interval can be
  communicated explicitly or implicitly by the monitored system to the
  monitor, for instance, by attaching a sequence number to each
  action.  See Section~\ref{subsubsec:sequencenumber} for details.
  \exampleendmark
\end{example}
We remark that the interval associated with the last letter of an
observation is always unbounded. This reflects that a monitor is
unaware of what it will observe in the future.  More generally, a
letter $(I,a)$ of an observation with $|I|>1$ represents a knowledge
gap of the monitor. In particular, $a$ is the alphabet's least
element~$a_0$, meaning that nothing is known about the interpretation
of the predicate symbols and the register values during the time
period $I$.  Finally, note that according to
Definition~\ref{def:observation}, knowledge gaps $(I,a)$ can
completely disappear~(T\ref{enum:observation_removal}), or can be
partially resolved by adding a new time point where the interval is
split~(T\ref{enum:observation_split}), where
(T\ref{enum:observation_data}) can add additional knowledge to the new
time point by replacing $a$ with a letter that is larger with respect
to the alphabet's partial order.
For simplicity, we do not include a transformation in
Definition~\ref{def:observation} that allows one to shrink
nonsingleton intervals, that is, a transformation that replaces a
letter $(I,a)$ with $|I|>1$ by a letter $(I',a)$ with $|I'|>1$,
$I'\subsetneq I$, and $I'$ is unbounded if $I$ is unbounded.

\subsection{Three-valued Semantics}
\label{subsec:semantics}

\enlargethispage{\baselineskip}  
                                
\MTLdata's models under the three-valued semantics are observations,
which represent a monitor's partial knowledge
about the system behavior at a given point in time.  This is in
contrast to the models for the standard Boolean semantics for MTL,
which are timed words and capture the complete system
behavior in the limit.

For defining \MTLdata's three-valued semantics, we fix a \emph{data
  domain} $D$, which is a nonempty set of values with
$\unknown\not\in D$.  Furthermore, let $\Sigma$ be the alphabet
consisting of the letters $(\sigma,\rho)$, where $\sigma$ and $\rho$
are partial functions, namely,
$\sigma:P\pto \bigcup_{p\in P}2^{D^{\iota(p)}}$ and $\rho:R\pto D$.
Note that $\Sigma$ is partially ordered and its least element is
$(\emap,\emap)$.
Analogous to the definition of \MTLdata's Boolean semantics in
Section~\ref{sec:prelim}, we define the logic's three-valued semantics
by a function $\phi\mapsto \osem{w,i,\nu}{\phi}\in\Three$, for a given
observation~$w\in\Observations(\Sigma)$, $i\in\positions(w)$, and a
partial valuation~$\nu:V\pto D$.  We define this function inductively
over the formula structure.  In the following, we assume that
$w=\big(I_0,(\sigma_0,\rho_0)\big)\dots\big(I_{n-1},(\sigma_{n-1},\rho_{n-1})\big)$
and abuse notation by identifying the logic's constant symbol~$\true$
with the Boolean value $\true\in\Three$, and the connectives~$\neg$
and $\vee$ with the corresponding three-valued logical operators in
Table~\ref{tab:connectives}.
The nontemporal cases are as expected.
\begin{align*}
  \osem{w,i,\nu}{\true}
  :=\ &
  \true
  \\
  \osem{w,i,\nu}{p(\overline{x})}
  :=\ &
  \begin{cases}
    \true & \text{if $\overline{x}\in\pdef(\nu)$,
      $p\in\pdef(\sigma_i)$, and $\nu(\overline{x})\in\sigma_i(p)$}
    \\      
    \false & 
    \text{if $\overline{x}\in\pdef(\nu)$,
      $p\in\pdef(\sigma_i)$, and $\nu(\overline{x})\not\in\sigma_i(p)$}
    \\
    \unknown & \text{otherwise}
  \end{cases}
  \\
  \osem{w,i,\nu}{\freeze{r}{x}\phi}
  :=\ &
  \begin{cases}
    \osem{w,i,\nu[x\mapsto \rho_i(r)]}{\phi} &
    \text{if $x\in\pdef(\rho_i)$}
    \\
    \osem{w,i,\nu[x\mapsto \unknown]}{\phi} &\text{otherwise}
  \end{cases}
  \\
  \osem{w,i,\nu}{\neg\phi}
  :=\ &
  \neg\osem{w,i,\nu}{\phi}
  \\
  \osem{w,i,\nu}{\phi\vee\psi}
  :=\ &
  \osem{w,i,\nu}{\phi}
  \vee
  \osem{w,i,\nu}{\psi}
\end{align*}

The temporal cases are less straightforward.  In particular, the
definition must account for  letters in $w$ where a
nonsingleton interval represents knowledge gaps that may either
disappear or may be replaced by multiple letters in a refinement.  We
make use of the auxiliary functions
$\istp_w:\positions(w)\rightarrow\Three$ and
$\tc_{w,I}:\positions(w)\times\positions(w)\rightarrow\Three$, which
are as follows for the observation $w$ and an interval $I$.
\begin{equation*}
  \istp_w(i):=\begin{cases}
    \true & \text{if $|I_i|=1$}
    \\
    \unknown & \text{otherwise}
  \end{cases}
  \qquad\text{and}\qquad
  \tc_{w,I}(i,j):=\begin{cases}
    \true & 
    \text{if $I_i-I_j\not=\emptyset$ and $I_i-I_j\subseteq I$}
    \\
    \false & 
    \text{if $(I_i-I_j)\cap I=\emptyset$}
    \\
    \unknown & 
    \text{otherwise}
  \end{cases}
\end{equation*}
We use $\istp_w$ to check whether a position is a \emph{time point}
(hence, the name ``$\istp$''), and we use $\tc_{w,I}$ to check whether the
\emph{metric constraint}~$I$ of a temporal connective is valid or
unsatisfiable between two positions in $w$ (hence, the name ``$\tc$'').
Note that if $\tc_{w,I}(i,j)=\unknown$, then the metric constraint
between the positions $i$ and $j$ could either be satisfied or
violated, depending on some timestamps $\tau\in I_i$ and
$\tau'\in I_j$.

The semantics of the temporal connectives $\since_I$ and $\until_I$
is defined as follows.
\begin{align*}
  \osem{w,i,\nu}{\phi\since_I\psi}
  :=\ &
  \bigvee_{j\in\positions(w), j\leq i}
  \Big(\istp_w(j) \wedge \tc_{w,I}(i,j) \wedge 
  \osem{w,j,\nu}{\psi} 
  \wedge
  \bigwedge_{j<k\leq i}\big(\istp_w(k)\rightarrow
  \osem{w,k,\nu}{\phi}\big)\Big)
  \\
  \osem{w,i,\nu}{\phi\until_I\psi}
  :=\ &
  \bigvee_{j\in\positions(w), j\geq i}
  \Big(\istp_w(j) \wedge \tc_{w,I}(j,i) \wedge 
  \osem{w,j,\nu}{\psi} 
  \wedge
  \bigwedge_{i\leq k<j} \big(\istp_w(k)\rightarrow
  \osem{w,k,\nu}{\phi}\big)\Big)
\end{align*}
We comment on the definitions for $\phi\since_I\psi$ and
$\phi\until_I\psi$. First, note that $j$ ranges over so-called
``anchor'' positions and $k$ ranges over so-called ``continuation''
positions. For a position $j$ to be a ``valid'' anchor position, $j$
must be a time point, which is the case when
$\istp_w(j)=\true$. Otherwise, $\istp_w(j)=\unknown$.  Using the truth
value~$\false$ instead of $\unknown$ would be incorrect since it is
not yet known whether a refinement of $w$ will contain a time point
with a timestamp in $I_j$.
Furthermore, note that the function $\tc_{w,I}$ returns $\unknown$ if
it is unknown in $w$ whether the formula's metric constraint is always
satisfied or never satisfied for the positions $i$ and $j$.
Finally, suppose that a position $k$ between $j$ and $i$ is an
``invalid'' continuation position, that is, $\phi$'s truth value at
$k$ is $\false$.  If the interval $I_k$ is not a singleton,
then $\istp_w(k)$ ``downgrades'' this truth value to $\unknown$, since it
will be irrelevant in refinements of $w$ that do not contain any time
points with a timestamp in $I_k$.

Finally, we define the semantics of the temporal connectives
$\predecessor_I$ and $\successor_I$ as
\begin{equation*}
  \osem{w,i,\nu}{\predecessor_I\phi}:=
  c_0\vee c_{-1}\vee c_{-2}
  \qquad\text{and}\qquad
  \osem{w,i,\nu}{\successor_I\phi}:=
  c_0\vee c_1\vee c_2
\end{equation*}
with
\begin{equation*}
  c_k := \begin{cases}
    \tc_{w,I}(i,i)\wedge\osem{w,i,\nu}{\phi}\wedge\neg\istp_w(i)
    &\text{if $k=0$ and $I\not=\{0\}$,}
    \\
    \tc_{w,I}(i,i-1)\wedge\osem{w,i-1,\nu}{\phi}\wedge\istp_w(i-1)\wedge\istp_w(i)
    &\text{if $k=-1$ and $i\geq1$,}
    \\
    \tc_{w,I}(i+1,i)\wedge\osem{w,i+1,\nu}{\phi}\wedge\istp_w(i+1)\wedge\istp_w(i)
    &\text{if $k=1$ and $i<n-1$,}
    \\
    \tc_{w,I}(i,i-2)\wedge\osem{w,i-2,\nu}{\phi}\wedge\neg\istp_w(i-1)
    &\text{if $k=-2$ and $i\geq2$,}
    \\
    \tc_{w,I}(i+2,i)\wedge\osem{w,i+2,\nu}{\phi}\wedge\neg\istp_w(i+1)
    &\text{if $k=2$ and $i<n-2$,}
    \\
    \false 
    &\text{otherwise.}
  \end{cases}
\end{equation*}
We comment on the definition for $\successor_I\phi$ with $i<n-2$;
the other cases are analogous or restricted cases of this one.  One
might expect that the conjunct~$c_1$ is already sufficient. However,
having only $c_1$ could result in the wrong truth value $\false$ for
$\successor_I\phi$ at $i$ when $\osem{w,i+1,\nu}{\phi}=\false$.  If,
for example, $|I_i|>1$ then it is still possible to satisfy
$\successor_I\phi$ when refining the observation $w$ at $I_i$.  A
refinement of $w$ may consist of two time points with the timestamps
$\tau$ and $\tau'$ in $I_i$, where $\tau'-\tau\in I$ and $\phi$ is
true at the time point with timestamp~$\tau'$. The conjunct $c_0$
takes care of such a refinement at $i$.  The conjunct $c_2$ is
necessary when $|I_{i+1}|>1$.  In this case $i$ and $i+2$ are time
points in $w$.  The observation~$w$ may be refined by removing the
letter at position~$i+1$, resulting in an observation where $w$'s
letter at position~$i+2$ is the successor of $w$'s letter at position
$i$.
Note that $c_0$ and $c_2$ can be $\false$ or $\unknown$ but never
$\true$ because of the negative $\istp_w$ literals occurring in $c_0$
and $c_2$. Furthermore, again because of the $\istp_w$ literals, we
have that $c_0=c_2=\false$ whenever $c_1=\true$.  Finally, observe
that $\tc_{w,\{0\}}(i,i)\not=\false$.  However, the metric constraint
$\{0\}$ is only satisfiable for time points that have equal timestamps
and we require that timestamps are strictly increasing. Hence, the
additional constraint $I\not=\{0\}$ is needed when $k=0$.

Observe that it may be the case that $\osem{w,i,\nu}{\phi}\in\Two$
when $i$ is not a time point in $w$.  A trivial example is when
$\phi=\true$.  In a refinement of $w$, it might turn out that there
are no time points with timestamps in $I_i$, and hence a monitor
should not output a verdict for the specification~$\phi$ at
position~$i$ in $w$.  We address this artifact by downgrading (with
respect to the partial order~$\prec$) a Boolean truth value
$\osem{w,i,\nu}{\phi}$ to $\unknown$ when $i$ is not a time point. To
this end, we introduce the following variant of the semantics.
\begin{definition}
  \label{def:topsem}
  For a formula $\phi$, an observation $w$, $\tau\in\Qpos$, and a
  partial valuation $\nu$, we define
  \begin{equation*}
    \eosem{w, \tau,\nu}{\phi}:=
    \begin{cases}
      \osem{w,i,\nu}{\phi}
      & \text{if $\tau$ is the timestamp of some time point
        $i\in\positions(w)$, and}
      \\
      \unknown & \text{otherwise.}
    \end{cases}
  \end{equation*}
\end{definition}

\subsection{Properties}
\label{subsec:properties}

Our first theorem shows that \MTLdata's three-valued semantics
conservatively approximates its standard Boolean
semantics. Intuitively speaking, if a formula~$\phi$ evaluates to a
Boolean value for an observation at time $\tau\in\Qpos$, then $\phi$
has the same Boolean value at time~$\tau$ for any timed word that
refines the observation.
To state the theorem, we need the following definitions.  A timed
word~$w'$ \emph{refines} an observation~$w$, written $w\sqsubseteq w'$
for short, if for every $j\in\Nat$, there is some $i\in\positions(w)$,
such that $\tau_j\in I_i$, $\sigma_i\sqsubseteq \sigma'_j$, and
$\rho_i\sqsubseteq\rho'_j$, where $(I_\ell,(\sigma_\ell,\rho_\ell))$
and $(\tau_k,(\sigma'_k,\rho'_k))$, for $\ell\in\positions(w)$ and
$k\in\Nat$, are the letters of $w$ and $w'$, respectively.
Furthermore, similar to Definition~\ref{def:topsem}, we define for
$\tau\in\Qpos$, a timed word~$w$, a valuation $\nu$, and a formula
$\phi$,
\begin{equation*}
  \esem{w,\tau,\nu}{\phi}:=
  \begin{cases}
    \sem{w,j,\nu}{\phi} & 
    \text{if the $j$th letter of $w$ is $(\tau,(\sigma,\rho))$, and}
    \\
    \unknown & \text{otherwise.}
  \end{cases}
\end{equation*}
\begin{theorem}
  \label{thm:conservative}
  Let $\phi$ be a formula, $\mu$ a partial valuation, $\nu$ a total
  valuation, $u$ an observation, $v$ a timed word, and
  $\tau\in\Qpos$. If $u\sqsubseteq v$ and $\mu\sqsubseteq \nu$, then
  $\eosem{u,\tau,\mu}{\phi} \preceq \esem{v,\tau,\nu}{\phi}$.
\end{theorem}
\begin{proof}
  Let $(I_i,(\sigma_i,\rho_i))$ and $(\tau_j,(\sigma'_j,\rho'_j))$,
  for $i\in\positions(u)$ and $j\in\Nat$, be the letters of $u$
  and~$v$, respectively.
  Since $u\sqsubseteq v$, there is a function
  $\pi:\Nat\to\positions(u)$ such that (R1)~$\tau_j\in I_{\pi(j)}$,
  (R2)~$\sigma_{\pi(j)}\sqsubseteq \sigma'_j$, and
  (R3)~$\rho_{\pi(j)}\sqsubseteq\rho'_j$, for every $j\in\Nat$.
  It is easy to see that $\pi$ is monotonic.

  We prove by structural induction on $\phi$ that for every
  $i'\in\Nat$ and partial valuations~$\mu$ and~$\nu$ with
  $\pdef(\nu)=V$ and $\mu\sqsubseteq\nu$, it holds that
  $\osem{u,\pi(i'),\mu}{\phi} \preceq \sem{v,i',\nu}{\phi}$.  The
  theorem easily follows from this statement.
  Let $i'\in\Nat$, and let $\mu$ and $\nu$ be partial valuations with
  $\pdef(\nu)=V$ and $\mu\sqsubseteq\nu$.
  Furthermore, let $i=\pi(i')$.
  Note that the statement clearly holds for
  $\osem{u,i,\mu}{\phi}=\bot$.  Hence, it suffices to show that
  $\osem{u,i,\mu}{\phi} = \sem{v,i',\nu}{\phi}$, provided that
  $\osem{u,i,\mu}{\phi}\in\Two$.

  \emph{Base cases.} The case $\phi=\true$ is trivial. Consider the
  case $\phi=p(\overline{x})$, for some $p\in P$. As
  $\osem{u,i,\nu}{p(\overline{x})}\in\Two$, it holds that
  $\overline{x}\in\pdef(\nu)$ and $p\in\pdef(\sigma_i)$. It follows
  from the theorem's premise that
  $\mu(\overline{x})=\nu(\overline{x})$ and from~(R2) that
  $\sigma_{i}(p)=\sigma'_{i'}(p)$. Thus
  $\osem{u,i,\nu}{p(\overline{x})} = \sem{v,i',\nu}{p(\overline{x})}$.

  \emph{Inductive cases.} The cases where $\phi$ is of the form
  $\neg\alpha$ or $\alpha\vee\beta$ are straightforward and are omitted.
  We also omit the cases for $\predecessor_I\alpha$,
  $\successor_I\alpha$, and $\alpha\since_I\beta$, since they are
  similar to the case $\alpha\until_I\beta$.

  First, assume that $\phi$ is of the form $\freeze{r}{x}\psi$.
  Let $\eta=\mu[x\mapsto \rho_i(r)]$ if $r\in\pdef(\rho_i)$, and
  $\eta=\mu[x\mapsto\unknown]$ otherwise. Similarly, let
  $\eta'=\nu[x\mapsto \rho'_{i'}(r)]$ if $r\in\pdef(\rho'_{i'})$, and
  $\eta'=\mu[x\mapsto\unknown]$ otherwise.
  By~(R3), we have that if $r\in\pdef(\rho_i)$, then
  $r\in\pdef(\rho'_{i'})$ and $\rho_i(r)=\rho'_{i'}(r)$, and thus
  $\eta(x)=\eta'(x)$. Furthermore, if $r\not\in\pdef(\rho_i)$, then
  $x\not\in\pdef(\eta)$. Hence, $\eta\sqsubseteq\eta'$.
  It follows from the induction hypothesis that
  $\osem{u,i,\eta}{\psi} \preceq \sem{v,i',\eta'}{\psi}$ and therefore
  $\osem{u,i,\mu}{\phi}=\sem{v,i',\nu}{\phi}$.

  Assume that $\phi$ is of the form $\alpha\until_I\beta$. We consider
  first the case $\osem{u,i,\mu}{\alpha\until_I\beta}=\true$. By
  definition, there is some $j\in\positions(u)$ with $j\geq i$ such that
  $\istp_u(j)=\true$, $\tc_{u,I}(i,j)=\true$,
  $\osem{u,j,\mu}{\beta}=\true$, and
  $\istp_u(k)\to\osem{u,k,\mu}{\alpha}=\true$, for all $k$ with
  $i \leq k < j$.
  As $j$ is a time point in $u$, $\pi(j')=j$, for some $j'\in\Nat$.
  From~(R1), we have that $\tau_{j'}=\timestamp_{u}(j)$.  As
  $\tc_{u,I}(i,j)=\true$ and $I_j=\set{\tau_{j'}}$, we have that
  $\tau_{j'}-\tau\in I$, for all $\tau\in I_i$. From (R1), we have
  that $\tau_{i'}\in I_i$. Thus, $\tau_{j'}-\tau_{i'}\in I$~(I1).
  From the induction hypothesis, 
  $\osem{u,j,\mu}{\beta} \preceq \sem{v,j',\nu}{\beta}$. Hence,
  $\sem{v,j',\nu}{\beta}=\true$~(I2).
  We also have that 
  $\osem{u,\pi(k'),\mu}{\alpha}\preceq\sem{v,k',\nu}{\alpha}$, for
  any $k'\in\Nat$.
  Let $k'\in\Nat$ such that $i'\leq k'< j'$, and let $k=\pi(k')$. By
  the monotonicity of $\pi$, we have that $i\leq k\leq j$. Since $j$ is
  a time point in $u$ we also have that $k<j$.
  As $\istp_u(k)\to\osem{u,k,\mu}{\alpha}=\true$ and $\istp_u$ is
  never $\false$ by definition, we have that
  $\osem{u,k,\mu}{\alpha}=\true$. Then
  $\sem{v,k',\mu}{\alpha}=\true$~(I3).
  Summing up, from (I1), (I2), (I3), and as $k'$ was chosen
  arbitrarily, we obtain that
  $\sem{v,i',\nu}{\alpha\until_I\beta}=\true$.
  
  The case $\osem{u,i,\mu}{\alpha\until_I\beta} = \false$ is as
  follows.  Note that each disjunct in the definition of
  $\osem{u,i,\mu}{\alpha\until_I\beta}$ is $\false$.
  We fix an arbitrary $j'\geq i'$ and let $j=\pi(j')$. It holds that
  $\istp_u(j) \wedge \tc_{u,I}(j,i) \wedge \osem{u,j,\mu}{\beta}
  \wedge \bigwedge_{i\leq k<j} (\istp_u(k)\rightarrow
  \osem{u,k,\mu}{\alpha}) = \false$.
  Since $\istp_u(j)\not=\false$, one of the remaining conjuncts must
  be $\false$.
  \begin{enumerate}[(1)]
  \item If $\tc_{u,I}(i,j)=\false$, then $\tau'-\tau''\not\in I$, for
    all $\tau''\in I_i$ and $\tau'\in I_j$. From~(R1),
    $\tau_{i'}\in I_i$ and $\tau_{j'}\in I_j$, it follows that
    $\tau_{j'}-\tau_{i'}\not\in I$.
  \item If $\osem{u,j,\mu}{\beta}=\false$, then
    $\sem{v,j',\nu}{\beta}=\false$, by the induction hypothesis.
  \item If $\istp_u(k)\to\osem{u,k,\mu}{\alpha}=\false$, for some $k$
    with $i\leq k< j$, then $\istp_u(k)=\true$ and
    $\osem{u,k,\mu}{\alpha}=\false$. It follows as before that there
    is a $k'$ with $i'\leq k'<j'$ such that
    $\sem{v,k',\nu}{\alpha}=\false$.
  \end{enumerate}
  We have thus obtained that either $\tau_{i'}-\tau_{j'}\not\in I$ or
  one of the conjuncts of
  $\sem{v,j',\nu}{\beta} \wedge \bigwedge_{i'\leq k'< j'}
  \sem{v,k',\nu}{\alpha}$ is $\false$. In other words, if $j'$ is
  such that $\tau_{j'}-\tau_{i'}\in I$, then
  $\sem{v,j',\nu}{\beta} \wedge \bigwedge_{i'\leq k'<j'}
  \sem{v,k',\nu}{\alpha} = \false$.
  As $j'$ was chosen arbitrarily, we conclude that
  $\sem{v,i',\nu}{\alpha\until_I\beta}=\false$. 
\end{proof}

The next theorem states that \MTLdata's three-valued semantics is
monotonic in $\sqsubseteq$ (on observations and partial valuations)
and $\preceq$ (on truth values). This property is crucial for
monitoring, since it guarantees that a verdict output for an
observation stays valid for refined observations.
\begin{theorem}
  \label{thm:monotonicity}
  Let $\phi$ be a formula, $\mu$ and $\nu$ partial valuations, $u$
  and $v$ observations, and $\tau\in\Qpos$. If $u\sqsubseteq v$ and
  $\mu\sqsubseteq\nu$, then
  $\eosem{u,\tau,\mu}{\phi} \preceq \eosem{v,\tau,\nu}{\phi}$.
\end{theorem}
\begin{proof}
  The proof is similar to that of Theorem~\ref{thm:conservative} and
  details are thus omitted. We just note that we make use of the 
  following properties~(R1$'$), (R2$'$), and (R3$'$), which correspond to
  the ones used in the proof of Theorem~\ref{thm:conservative}.

  Let $w$ and $w'$ be observations with letters
  $\big(I_i,(\sigma_i,\rho_i)\big)$ and, respectively,
  $\big(I'_j,(\sigma'_j,\rho'_j)\big)$, for $i\in\positions(w)$ and
  $j\in\positions(w')$.  We claim that if $w\sqsubseteq w'$, then there
  is a monotonic function $\pi:\positions(w')\rightarrow\positions(w)$
  with the following properties.
  \begin{enumerate}[(R1$'$)]
  \item $I_j'\subseteq I_{\pi(j)}$, for all $j\in \positions(w')$.
  \item $\sigma_{\pi(j)}\sqsubseteq \sigma_j'$, for all
    $j\in \positions(w')$. 
  \item $\rho_{\pi(j)}\sqsubseteq \rho'_{j}$, for all
    $j\in \positions(w')$.
  \end{enumerate}
  If $w=w'$ then take $\pi$ to be the identity. If $w'$ is obtained
  from $w$ using one of the transformations, that is, if
  $w\sqsubset_1 w'$, then, for each transformation it is easy to
  construct a function~$\pi$ that has the stated properties. If
  $w\sqsubsetneq w'$, then there is a sequence
  $(w_i)_{0\leq i \leq n}$ of observations, with $n>1$, such that
  $w=w_0\sqsubset_1 w_1 \sqsubset_1 \dots \sqsubset_1 w_n=w'$. From
  the previous observation, there is a sequence of functions
  $\pi_i:\positions(w_{i})\to\positions(w_{i-1})$, with
  $1\leq i\leq n$, each having the stated properties. The functions'
  composition $\pi=\pi_1\circ\dots\circ\pi_{n}$ also has these
  properties.
\end{proof}

We next investigate the decision problem that underlies monitoring.
Note that we do not require that the interpretations of the predicate
symbols are finite relations.  However, for monitoring, the relations
must be decidable, and a monitor needs an algorithm for
performing membership checks.  For the following theorem, we assume
that the membership of a tuple in a predicate symbol's interpretation
at a time point can be checked in PSPACE.
\begin{theorem}
  \label{thm:decidable}
  For a formula~$\phi$, an observation~$w$, a partial
  valuation~$\nu$, $\tau\in\Qpos$, and a truth value $b\in\Two$,
  the problem of whether $\eosem{w,\tau,\nu}{\phi}$ equals~$b$ is
  $\mathrm{PSPACE}$-complete.
\end{theorem}
\begin{proof}
  We first show that the problem is PSPACE-hard by reducing the
  satisfiability problem for quantified Boolean logic (QBL) to it.
  Let $\alpha$ be a closed QBL formula over the propositions
  $p_1,\dots,p_n$.  We define the set $P$ of predicate symbols as
  $\{P_1,\dots,P_n\}$, where each predicate symbol has arity $1$.
  Moreover, let $R:=\{r\}$ and $D:=\{0,1\}$, and let $w$ be the
  observation
  \begin{equation*}
    \big(\{0\},\sigma,\rho_0\big)\ 
    \big(\{1\},(\sigma,\rho_1)\big)\ 
    \big(\{3\},(\emap,\emap)\big)\ 
    \big((3,\infty),(\emap,\emap)\big)\,,
  \end{equation*}
  with $\sigma(P_i)=\{1\}$, for each $i\in\{1,\dots,n\}$, and
  $\rho_i(r)=i$, for $i\in\{0,1\}$.  Finally, we translate the QBL
  formula $\alpha$ to an \MTLdata formula $\alpha^*$ as follows.
  \begin{equation*}
      p_i^* := P_i(x_i) 
      \quad
      (\neg \beta)^* := \neg \beta^*
      \quad
      (\beta\vee\gamma)^* := \beta^*\vee\gamma^*
      \quad
      (\exists p_i.\,\beta)^* :=
      \once\eventually_{[0,1]}\freeze{}{x_i}{\beta^*}
      \quad
      (\forall p_i.\,\beta)^* := \historically\,\always_{[0,1]}\freeze{}{x_i}{\beta^*}
  \end{equation*}
  It is easy to see that $\alpha$ is satisfiable iff
  $\eosem{w,0,\emap}{\alpha^*}=\true$.

  We only sketch the problem's membership in PSPACE.  Note that $w$ is
  finite. If there is no time point in $w$ with timestamp $\tau$, then
  $\eosem{w,\tau,\nu}{\phi}=\unknown$. Suppose that $i\in\positions(w)$
  is a time point in $w$ with timestamp~$\tau$.
  A computation of $\phi$'s truth value at position $i$ can be easily
  obtained from the inductive definition of the satisfaction
  relation~$\omodels$.  Note, however, that the space consumed by
  naively unfolding the semantic definitions would in general not be
  polynomially bounded.  One reason is that subformulas may occur
  multiple times in the unfolding for different time points and
  valuations.  Instead, we must carry out this computation by a
  depth-first traversal when unfolding the semantic definitions to
  stay in PSPACE.  Furthermore, note that our additional assumption on
  the membership checks allows us to determine in PSPACE the truth
  value of an atomic formula at a time point.
\end{proof}

In a propositional setting, the corresponding decision problem can be
solved in polynomial time using dynamic programming, where the truth
values at the positions of an observation are propagated up the
formula structure.  Note that the truth value of a proposition at a
position is given by the observation's letter at that position.
This is in contrast to \MTLdata, where atomic formulas can have free
variables and their truth values at the positions in an
observation~$w$ may depend on the data values stored in the registers
and frozen to these variables at different time points of $w$.  Before
truth values are propagated up, the bindings of variables to data
values must be propagated down.

\subsection{Monitoring Correctness Requirements}
\label{subsec:requirements}

A monitor for a specification iteratively receives information about
the system behavior.  Abstractly speaking, the monitor's input is an
infinite sequence $(\mathit{in}_i)_{i\in\Nat}$, where $\mathit{in}_i$
describes a part of the system behavior and is received by the monitor
in its $i$th iteration.  The monitor's output is an infinite sequence
$(\mathit{out}_i)_{i\in\Nat}$, where $\mathit{out}_i$ is the output in
iteration~$i$ describing when the monitor's specification is satisfied
or violated.
In the following, we concretize a monitor's input and output for our
setting and define correctness requirements for monitoring.
Note that we assume that a monitor never terminates and that it
infinitely often receives information about the system behavior.  This
assumption is invalid if, for instance, the system observed by the
monitor ever terminates.  Nevertheless, we make this assumption to
simplify matters and it is easy to adapt our definitions and results
to the general case.

We first turn to a monitor's input, which is a sequence of
observations $(w_i)_{i\in\Nat}$.  That is, we view the
observation~$w_i$ as the input to the monitor at iteration $i\in\Nat$.
In practice, a monitor would receive at iteration $i>0$ a message that
describes just the differences between $w_{i-1}$ and $w_i$.
Furthermore, note that the $w_i$s can be understood as abstract
descriptions of the monitor's state over time, representing the
monitor's knowledge about the system behavior, where $w_0$ represents
the monitor's initial knowledge.  Also note that if the timed word $v$
is the system behavior in the limit, then $w_i\sqsubseteq v$, for all
$i\in\Nat$, assuming that components do not send bogus
messages. However, for every $i\in\Nat$, there are infinitely many
timed words~$u$ with $w_i\sqsubseteq u$.  Since messages sent to the
monitor can be lost, it can even be the case that there is a timed
word $u$ with $u\not=v$ and $w_i\sqsubseteq u$, for all $i\in\Nat$.
\begin{definition}
  The infinite sequence $\bar{w} = (w_i)_{i\in\Nat}$ of observations
  is \emph{valid} if $w_0 = \big([0,\infty), (\emap,\emap)\big)$ and
  $w_i\sqsubsetneq w_{i+1}$, for all $i\in\Nat$.
\end{definition}

We turn to a monitor's output.  Based on the input $(w_i)_{i\in\Nat}$,
the monitor outputs in each iteration $i\in\Nat$ a set $V_i$ of
\emph{verdicts}, which is a finite set of pairs $(\tau,b)$ with
$\tau\in\Qpos$ and $b\in\Two$.  Intuitively, $\tau$ is the time at
which the specification has the Boolean value $b$.
\begin{definition}
  \label{def:soundness_completeness-observation}
  Let $\phi$ be a closed formula, $\bar{w}=(w_i)_{i\in\Nat}$ a valid
  observation sequence, and $\bar{V}=(V_i)_{i\in\Nat}$ a sequence of
  verdict sets.
  \begin{enumerate}[(i)]
  \item $\bar{V}$ is \emph{observationally sound} for $\bar{w}$ and
    $\phi$ if for all partial valuations~$\nu$ and $i\in\Nat$,
    whenever $(\tau,b)\in V_i$ then $\eosem{w_i,\tau,\nu}{\phi}=b$.
  \item $\bar{V}$ is \emph{observationally complete} for $\bar{w}$ and
    $\phi$ if for all partial valuations~$\nu$, $i\in\Nat$, and
    $\tau\in\Qpos$, if $\eosem{w_i,\tau,\nu}{\phi}\in\Two$ then
    $(\tau,b)\in\bigcup_{j\leq i}V_j$, for some $b\in\Two$.
  \end{enumerate}
\end{definition}
We say that a monitor is \emph{observationally sound} if for all valid
observation sequences~$\bar{w}$ and closed formulas $\phi$, its
sequence of verdict sets is observationally sound for $\bar{w}$ and
$\phi$.  The definition of a monitor being \emph{observationally
  complete} is analogous.

It follows from Theorem~\ref{thm:decidable} that monitors for \MTLdata
exist that are both observationally sound and complete.  In
Sections~\ref{sec:prop} and~\ref{sec:data}, we present such monitoring
algorithms in detail.  In the remainder of this section, we relate the
correctness requirements from
Definition~\ref{def:soundness_completeness-observation} to
requirements that demand that a monitor outputs a verdict as soon as
the specification has the same Boolean value on every extension of the
monitor's current knowledge.  Such requirements are stronger and
achieving them can be hard or even impossible for nontrivial
specification languages.  In particular, we show that monitors
satisfying such a requirement do not exist for \MTLdata.  We start
with an example that illustrates the differences on the verdicts for
monitoring.
\begin{example}
  Consider the formula~$\phi=\always(p \wedge \eventually \neg p)$.
  Note that under the classical Boolean semantics, $\phi$ is logically
  equivalent to $\false$, however not under the three-valued
  semantics.  For example, $\osem{w,0,\nu}{\phi}=\bot$, for
  $w=\big([0,\infty),(\emap,\emap)\big)$ and any valuation~$\nu$.
  Given a valid observation sequence~$\bar{w}$, an observationally
  sound and complete monitor for $\bar{w}$ and~$\phi$ first outputs
  the verdict $(0,\false)$ for the minimal $i$ such that $w_i$
  contains a letter that assigns $p$ to false.
  In contrast, a sound and complete monitor for the classical Boolean
  semantics (see Definition~\ref{def:soundness_completeness} below)
  must immediately output the verdict $\false$.
  \exampleendmark
\end{example}

For an observation~$w$, we define
$U_w:=\setx{v}{v\text{ a timed word with }w\sqsubseteq v}$.
Intuitively, $U_w$ contains the timed words that are compatible with
the reported system behavior that a monitor received so far,
represented by $w$.
\begin{definition}
  \label{def:soundness_completeness}
  Let $\phi$ be a closed formula, $\bar{w}$ a valid observation
  sequence, and $\bar{V}$ a sequence of verdict sets.
  \begin{enumerate}[(i)]
  \item $\bar{V}$ is \emph{sound} for $\bar{w}$ and $\phi$ if for all
    valuations~$\nu$ and $i\in\Nat$, whenever $(\tau,b)\in V_i$, then
    $\bigcurlywedge_{v\in U_{w_i}}\esem{v,\tau,\nu}{\phi}=b$, that is,
    the meet $\curlywedge$ of the truth values in the lower
    semilattice $(\Three,\prec)$ is $b$.
  \item $\bar{V}$ is \emph{complete} for $\bar{w}$ and $\phi$ if for
    all valuations~$\nu$, $i\in\Nat$, and $\tau\in\Qpos$, whenever
    $\bigcurlywedge_{v\in U_{w_i}}\esem{v,\tau,\nu}{\phi}\in\Two$, then
    $(\tau,b)\in \bigcup_{j\leq i}V_j$, for some $b\in\Two$.
  \end{enumerate}
\end{definition}
We say that a monitor is \emph{sound} if for all valid observation
sequences $\bar{w}$ and closed formulas $\phi$, its sequence of
verdict sets is sound for $\bar{w}$ and $\phi$.  The definition of a
monitor being \emph{complete} is analogous.

With the help of Theorem~\ref{thm:conservative}, we prove that the
completeness requirement from
Definition~\ref{def:soundness_completeness-observation} is indeed a
weaker notion than the completeness requirement from
Definition~\ref{def:soundness_completeness}, while the soundness
requirement from Definition~\ref{def:soundness_completeness} offers
the same correctness guarantees as the one from
Definition~\ref{def:soundness_completeness-observation}.
\begin{theorem}
  \label{thm:obs_requirements}
  Let $M$ be a monitor. 
  \begin{enumerate}[(a)]
  \item If $M$ is observationally sound, then $M$ is sound.
  \item If $M$ is complete, then $M$ is observationally complete.
  \end{enumerate}
\end{theorem}
\begin{proof}
  Let $\bar{V}$ be the sequence of verdict sets that $M$ iteratively
  outputs for $\phi$ and $\bar{w}$.

  We first prove~(a).  Assume that $M$ is observationally sound.  Let
  $\nu$ be a total valuation, $i\in\Nat$, $\tau\in\Qpos$, and
  $b\in\Two$ such that $(\tau,b)\in V_i$. Then, by definition,
  $\eosem{w_i,\tau,\nu}{\phi}=b$. For $v\in U_{w_i}$, we have that
  $w_i\sqsubseteq v$.  By Theorem~\ref{thm:conservative}, we obtain
  that $\esem{v,\tau,\nu}{\phi}=b$.  It follows that
  $\bigcurlywedge_{v\in U_{w_i}}\esem{v,\tau,\nu}{\phi}=b$.
  We conclude that $M$ is sound.

  It remains to prove~(b).  Assume that $M$ is complete.  Let $\nu$ be
  a partial valuation, $i\in\Nat$, and $\tau\in\Qpos$ such that
  $\esem{w_i,\tau,\nu}{\phi}=b'$, for some $b'\in\Two$.
  Let $\nu'$ be a total valuation with $\nu\sqsubseteq \nu'$ and
  $v\in U_{w_i}$. As $w_i\sqsubseteq v$, we obtain from
  Theorem~\ref{thm:conservative} that
  $\esem{v,\tau,\nu'}{\phi}=b'$. As $v$ was chosen arbitrarily, we get
  $\bigcurlywedge_{v\in U_{w_i}}\esem{v,\tau,\nu'}{\phi}=b'$. From
  $M$'s completeness, it follows that there are $b\in\Two$ and
  $j\in\Nat$ with $j\leq i$ such that $(\tau,b)\in V_j$. We conclude
  that $M$ is observationally complete.
\end{proof}

The correctness requirements in
Definition~\ref{def:soundness_completeness} are related to the use of
a three-valued ``runtime-verification'' semantics for a specification
language as introduced
by~\citeauthor{Bauer_etal:rv_tltl}~\citeyearpar{Bauer_etal:rv_tltl}
for LTL and adopted by other runtime-verification approaches, for
example, the one
by~\citeauthor{Bauer-FMSD15}~\citeyearpar{Bauer-FMSD15}.  Both a sound
and complete monitor, and a monitor implementing the three-valued
``runtime-verification'' semantics output a verdict as soon as the
specification has the same Boolean value on every extension of the
monitor's current knowledge.  However, as we explain next,
efficient monitors can be hard to
achieve or may not even exist for nontrivial specification languages.
\begin{remark}
  \label{rem:correctness}
  Having a sound and complete monitor $M$ for a specification language
  is at least as hard as checking satisfiability for this language.
  For instance, we can use a sound and complete monitor $M$ to check
  satisfiability for \MTLdata as follows. We run $M$ for the closed
  formula~$\phi$ whose satisfiability we want to check.  We refine
  $M$'s initial knowledge by the
  transformations~(T\ref{enum:observation_split})
  and~(T\ref{enum:observation_removal}) and add the first time point
  with the timestamp~$0.0$.  The formula $\phi$ is unsatisfiable under
  the standard Boolean semantics iff $M$'s verdict set~$V_1$ contains
  $(0.0,\false)$. Already MTL with the standard Boolean semantics is
  undecidable~\cite{OuaknineW06} and many of its nontrivial decidable
  fragments have a high complexity.  Recall that the satisfiability
  problem for LTL is
  PSPACE-complete~\cite{Sistla_Clarke:complexity_ltl}.

  Some monitoring approaches try to compensate for this complexity
  burden with a preprocessing step. For
  instance,~\citeauthor{Bauer_etal:rv_tltl}~\citeyearpar{Bauer_etal:rv_tltl}
  translates an LTL formula into an automaton prior to monitoring.
  The resulting automaton can be directly used for sound and complete
  monitoring in environments where messages are neither delayed nor
  lost.  However, there are no obvious extensions that handle
  out-of-order message delivery.  Furthermore, not every specification
  language has such a corresponding automaton model and, for those
  where translations are known, the automaton construction can be very
  costly.  For LTL, the size of the automaton is already in the worst
  case doubly exponential in the size of the
  formula~\cite{Bauer_etal:rv_tltl}.  \exampleendmark
\end{remark}

\section{Monitoring in the Propositional Setting}
\label{sec:prop}

In this section, we present an observationally sound and complete
online algorithm for MTL. We extend the algorithm in the next section
to \MTLdata, where we also provide the proof details.  To support
scalable monitoring, the verdict computation is incremental in that
the results from previous computations are reused whenever
observations are refined by the
transformations~(T\ref{enum:observation_split}),
(T\ref{enum:observation_removal}), and (T\ref{enum:observation_data})
from Definition~\ref{def:observation}.
We start with the algorithm's main
procedure~(Section~\ref{subsec:main}). Afterwards, we describe the
state the algorithm maintains~(Section~\ref{subsec:state}) and further
algorithmic details~(Section~\ref{subsec:algo}).

\subsection{Main Procedure}
\label{subsec:main}

\begin{listing}
  \begin{minipage}[t]{0.45\linewidth}
\begin{lstlisting}[caption={The monitor's main loop for MTL.},label={lst:monitor},captionpos=b]
procedure MonitorMTL($\phi$)
  Init($\phi$)
  loop
    $m$ := ReceiveMessage()
    $\mathit{ts}$ := UpdateKnowledge($m$)
    foreach $t$ with $t$ in $\mathit{ts}$ do
      case (T1): $J$, $\tau$ := DeltaT1($t$)
                $\!$AddTimePoint($\phi$, $J$, $\tau$)
      case (T2): $K$ := DeltaT2($t$)
                $\!$RemoveInterval($K$)  
      case (T3): $\tau$, $\sigma$ := DeltaT3($t$)
                $\!$foreach $p$ with $p\in\pdef(\sigma)$ do
                  PropagateTruthValue($p$, $\{\tau\}$, $\sigma(p)$)
\end{lstlisting}
  \end{minipage}
\end{listing}
The pseudocode of the monitor's top-level procedure is shown in
Listing~\ref{lst:monitor}.  In a nutshell, after the monitor
initializes its state, it enters a nonterminating loop. In
each loop iteration, the monitor receives a message, updates its state
according to the information extracted from the message, and outputs
the computed verdicts.
Recall from Section~\ref{subsec:requirements} that each message received
describes the ``delta'' between two subsequent observations in a valid
observation sequence~$(w_i)_{i\in\Nat}$. The message format and
therefore how the monitor obtains the necessary information from a
message and its current state are system-dependent. A possible
realization is given in Section~\ref{sec:app}.

We provide a brief description of the procedures used by the monitor's
top-level procedure. The procedure \ls{Init} initializes the monitor's
state; see Section~\ref{subsec:state} for details.
The procedure \ls{ReceiveMessage} receives a
message, for instance, over a channel or from a log file.
The procedure
\ls{UpdateKnowledge} updates the monitor's knowledge
about the system behavior.  This procedure also returns a
list of the transformations that transform the observation~$w_{i-1}$
into the observation~$w_i$ in the $i$th iteration.  The
monitor uses the 
procedures \ls{DeltaT1}, \ls{DeltaT2}, and \ls{DeltaT3}
to learn how the observation is updated. 
Concretely, \ls{DeltaT1} returns the timestamp~$\tau$ of a new time
point and the interval $J$ that is split at~$\tau$.
\ls{DeltaT2} returns the interval~$K$ of the letter that is removed
from the observation.
\ls{DeltaT3} returns the Boolean values $\sigma(p)$ of the newly
assigned propositions $p\in\pdef(\sigma)$ at the time point with the
timestamp~$\tau$.
The procedures \ls{AddTimePoint}, \ls{RemoveInterval}, and
\ls{PropagateTruthValue} are central to the monitor.  They update the
monitor's state.  For instance, \ls{PropagateTruthValue} propagates
the Boolean values of newly assigned propositions.
Section~\ref{subsec:algo} provides algorithmic details for these three
procedures.

Before we proceed, we introduce the following conventions that we use
in the remainder of this section.  Let $\phi$ be the MTL formula
that is monitored with propositions in $P$.  The letter~$I$ ranges over the
metric constraints of the temporal connectives that occur in $\phi$.
The letters~$\alpha$, $\beta$, and $\gamma$ range over elements in
$\sub(\phi)$.  Furthermore, let $w$ be an observation. It ranges over
the elements in the valid observation sequence $(w_i)_{i\in\Nat}$.
The letters~$J$, $K$, and $H$ range over the intervals that occur in
letters of $w$.  The lower case letters $j$, $k$, and $h$ are the
indexes of the letters in $w$ with the intervals $J$, $K$, and $H$,
respectively.  We also simplify notation.  We omit the partial
valuation $\nu$ in $\osem{w,i,\nu}{\gamma}$, that is, we only write
$\osem{w,i}{\gamma}$.  Note that $\nu$ is irrelevant for MTL.  We also
assume that $\phi$ is not an atomic formula and subformulas of $\phi$
are pairwise distinct.  Both assumptions are without loss of
generality. For example, the second one is met when representing
formulas as abstract syntax trees.

\subsection{Monitor State}
\label{subsec:state}

\subsubsection{Reduction to Propositional Logic}
\label{subsubsec:reduct}

At the core of the monitor is a mapping of MTL's three-valued
semantics into propositional logic with the standard two-valued
semantics.  From a high-level perspective, the monitor's state
comprises a representation of propositional formulas, which the
monitor refines and simplifies whenever it receives information about
the system behavior.  For readability, we start with a variant of
these propositional formulas that is close to the definition of MTL's
three-valued semantics.

The propositional formula $\fPhi{w}{\gamma,J}$ over propositions of
the form $\alpha^{K}$, $\tpp^{K}$, and $\tcp^{H,K}_I$ is defined as
follows. Its inductive definition follows the definition of MTL's
three-valued semantics in Section~\ref{subsec:semantics}, where the
propositions $\tpp^{K}$ and $\tcp^{H,K}_I$ take the role of the
corresponding functions.
\begin{equation*}
  \fPhi{w}{\gamma,J} := 
    \begin{cases}
      \true & \text{if $\gamma=\true$}
      \\
      p^{J} & \text{if $\gamma=p$ with $p\in P$}
      \\
      \neg\alpha^{J} & \text{if $\gamma=\neg\alpha$}
      \\
      \alpha^{J}\vee\beta^{J} & \text{if $\gamma=\alpha\vee\beta$}
      \\
      \bigvee_{K\leq J}\big(\tpp^{K} \wedge \tcp^{J,K}_I \land \beta^{K} \wedge
      \bigwedge_{K<H\leq J}(\tpp^{H} \to \alpha^{H})
      \big)
      & 
      \text{if $\gamma=\alpha\since_I\beta$}
      \\
      \bigvee_{K\geq J}\big(\tpp^{K} \wedge \tcp^{K,J}_I \wedge \beta^{K} \wedge
      \bigwedge_{J\leq H<K}(\tpp^{H} \to \alpha^{H})
      \big)
      & 
      \text{if $\gamma=\alpha\until_I\beta$}
      \\
      C_w^{J,0}\vee C_w^{J,-1}\vee C_w^{J,-2}
      &
      \text{if $\gamma=\predecessor_I\alpha$}
      \\
      C_w^{J,0}\vee C_w^{J,+1}\vee C_w^{J,+2}
      &
      \text{if $\gamma=\successor_I\alpha$}
    \end{cases}
\end{equation*}
with
\begin{equation*}
  C_w^{J,\pm\ell}:=
  \begin{cases} 
    \tcp^{J,J}_I\wedge\alpha^J\wedge\neg\tpp^J
    & \text{if $\ell=0$ and $I\not=\{0\}$}
    \\
    \tcp^{\max\{J,J\pm1\},\min\{J,J\pm1\}}_I\wedge\alpha^{J\pm1}\wedge\tpp^J\wedge\tpp^{J\pm1}
    & \text{if $\ell=\pm1$ and $j\pm1\in\pos(w)$}
    \\
    \tcp^{\max\{J,J\pm2\},\min\{J,J\pm2\}}_I\wedge\alpha^{J\pm2}\wedge\neg\tpp^{J\pm1}
    & \text{if $\ell=\pm2$ and $j\pm2\in\pos(w)$}
    \\
    \false & \text{otherwise}
  \end{cases}
\end{equation*} 
where $J\pm\ell$ denotes the interval of $w$'s letter at the position
$j\pm\ell$, provided that $j\pm\ell\in\positions(w)$.
We also define the substitution~$\theta_w$ over the propositions of
$\fPhi{w}{\gamma,J}$ as follows.
\begin{equation*}
  \alpha^K \mapsto \begin{cases}
    \true & \text{if $\osem{w,k}{\alpha}=\true$}
    \\
    \false & \text{if $\osem{w,k}{\alpha}=\false$}
  \end{cases}
  \qquad\quad  
  \tpp^K \mapsto \true \quad\text{if $|K|=1$}
  \qquad\quad
  \tcp_I^{H,K} \mapsto \begin{cases}
    \true & \text{if $H-K\not=\emptyset$ and $H-K\subseteq I$}
    \\
    \false & \text{if $(H-K)\cap I=\emptyset$}
  \end{cases}
\end{equation*}
For the propositions not listed, $\theta_w$ is undefined.  In general,
a substitution~$\theta$ is a partial function from propositions to
propositional formulas. Its homomorphic extension to propositional
formulas is as expected, in particular, $\theta(\Psi)$ is the
propositional formula in which the occurrences of
propositions~$p\in\pdef(\theta)$ within the propositional formula~$\Psi$
are replaced by $\theta(p)$, and the occurrences of propositions not
in $\pdef(\theta)$ are unaltered.

Let $\equiv$ denote semantic equivalence between propositional
formulas.  The following lemma connects $\gamma$'s truth value under
MTL's three-valued semantics with the propositional formula
$\theta_w(\fPhi{w}{\gamma,J})$.  Its proof is straightforward and
omitted.
\begin{lemma}
  \label{lem:propformula}
  The following two statements hold.
  \begin{enumerate}[(i)]
  \item If $\osem{w,j}{\gamma}\in\Two$ then 
    $\theta_w(\fPhi{w}{\gamma,J})\equiv \osem{w,j}{\gamma}$.
  \item If $\osem{w,j}{\gamma}=\unknown$ then
    $\theta_w(\fPhi{w}{\gamma,J})\not\equiv\true$ and 
    $\theta_w(\fPhi{w}{\gamma,J})\not\equiv\false$.
  \end{enumerate}
\end{lemma}
Note that the propositional formula $\theta_w(\fPhi{w}{\gamma,J})$
tells us more than the truth value $\osem{w,j}{\gamma}$. When
$\theta_w(\fPhi{w}{\gamma,J})\not\equiv b$, for $b\in\Two$, we also
know, in addition to $\osem{w,j}{\gamma}=\bot$, what causes the
uncertainty, namely, the corresponding counterparts of the
propositions that are not replaced by Boolean constants.

Next, we provide a tailored version of $\fPhi{w}{\gamma,J}$ that is
better suited for monitoring.  Note that for the cases
$\gamma=\alpha\since_I\beta$ and $\gamma=\alpha\until_I\beta$, at an
anchor position $K$, the truth value $\osem{w,k}{\alpha}$ is
irrelevant for $\osem{w,j}{\gamma}$.  However, when $|K|>1$ and when
refining $w$ by splitting $K$ at some $\kappa\in K$, we obtain new
anchor and continuation positions for which the truth value
$\osem{w,k}{\alpha}$ becomes relevant.  With the tailored
version~$\fPsi{w}{\gamma,J}$ of $\fPhi{w}{\gamma,J}$ we keep track of
$\alpha$'s truth value at anchor positions $K$
(cf. Example~\ref{ex:alphaK}).  The definition of $\fPsi{w}{\gamma,J}$
is as follows.
\begin{equation*}
  \fPsi{w}{\gamma,J}:=\begin{cases}
    \bigvee_{K\leq J}\big(
    \tpp^K\wedge
    \tcp^{J,K}_I \wedge \beta^K \wedge
    (\overline{\tpp}^K\rightarrow \alpha^K) \wedge
    \bigwedge_{K<H\leq J}(\tpp^H\rightarrow\alpha^H)
    \big)
    &\text{if $\gamma=\alpha\since_I\beta$}
    \\
    \bigvee_{K\geq J}\big(
    \tpp^K\wedge
    \tcp^{K,J}_I \wedge \beta^K \wedge
    (\overline{\tpp}^K\rightarrow \alpha^K) \wedge
    \bigwedge_{J\leq H<K}(\tpp^H\rightarrow\alpha^H)
    \big)
    &\text{if $\gamma=\alpha\until_I\beta$}
    \\
    \fPhi{w}{\gamma,J} & \text{otherwise}
  \end{cases}
\end{equation*}
For the new propositions $\overline{\tpp}^K$, we extend the
substitution $\theta_w$ by $\overline{\tpp}^K\mapsto\false$ if
$|K|=1$.

\begin{example}
  \label{ex:alphaK}
  We illustrate the definitions of the propositional formulas
  $\fPhi{w}{\gamma,J}$ and $\fPsi{w}{\gamma,J}$, with
  $\gamma=\alpha\until\beta$, and the reason for using
  $\theta_w(\fPsi{w}{\gamma,J})$ for monitoring.
  Let $w$ be an observation with the intervals $J_0=[0,\tau)$,
  $J_1=\set{\tau}$, and $J_2=(\tau,\infty)$, and where $\alpha$'s and
  $\beta$'s truth values are everywhere $\unknown$, except for
  position $1$, where $\osem{w,1}{\alpha}=\osem{w,1}{\beta}=\false$.
  By definition,
  \begin{align*}
    \fPhi{w}{\gamma,J_0} =\  
    &
    \big(\tpp^{J_0}\wedge\tcp_I^{J_0,J_0}\wedge\beta^{J_0}\big)\vee
    \big(\tpp^{J_1}\wedge\tcp_I^{J_1,J_0}\wedge\beta^{J_1}\wedge(\tpp^{J_0}\rightarrow\alpha^{J_0})\big)\,\vee
    \\
    &
    \big(\tpp^{J_2}\wedge\tcp_I^{J_2,J_0}\wedge\beta^{J_2}\wedge(\tpp^{J_0}\rightarrow\alpha^{J_0})\wedge(\tpp^{J_1}\rightarrow\alpha^{J_1})\big)
    \,,
    \\
    \fPsi{w}{\gamma,J_0} =\  
    &
    \big(\tpp^{J_0}\wedge\tcp_I^{J_0,J_0}\wedge\beta^{J_0}\wedge(\overline{\tpp}^{J_0}\rightarrow\alpha^{J_0})\big)\vee 
    \big(\tpp^{J_1}\wedge\tcp_I^{J_1,J_0}\wedge\beta^{J_1}\wedge(\overline{\tpp}^{J_1}\rightarrow\alpha^{J_1})\wedge(\tpp^{J_0}\rightarrow\alpha^{J_0})\big)\,\vee
    \\
    & 
    \big(\tpp^{J_2}\wedge\tcp_I^{J_2,J_0}\wedge\beta^{J_2}\wedge(\overline{\tpp}^{J_2}\rightarrow\alpha^{J_2})\wedge(\tpp^{J_1}\rightarrow\alpha^{J_1})\wedge(\tpp^{J_0}\rightarrow\alpha^{J_0})\big)\,,
     \\
  \intertext{and}
    \theta_w=\ 
    &[\alpha^{J_1}\mapsto\false, \beta^{J_1}\mapsto\false,
      \tpp^{J_1}\mapsto\true, \overline{\tpp}^{J_1}\mapsto\false,
      \tcp_I^{J_0,J_0}\mapsto\true,
      \tcp_I^{J_1,J_0}\mapsto\true,
      \tcp_I^{J_2,J_0}\mapsto\true]
    \,.
  \end{align*}
  Furthermore, let $w'$ be the observation that is obtained from $w$
  by the transformation~(T\ref{enum:observation_split}), where the
  interval $J_0$ is split at $\kappa\in J_0$. That is, the intervals
  of $w'$ are $K_0=[0,\kappa)$, $K_1=\set{\kappa}$,
  $K_2=(\kappa,\tau)$, $J_1$, and $J_2$.  Note that
  $\osem{w',3}{\alpha}=\osem{w',3}{\beta}=\false$ and $\unknown$
  anywhere else.

  We have the following semantic equivalences.
  \begin{equation*}
    \begin{array}{r@{\;}l@{\qquad\qquad}r@{\;}l}
      \theta_w(\fPhi{w}{\gamma,J_0}) & 
      \equiv 
      \tpp^{J_0}\wedge \beta^{J_0}
      & 
      \theta_{w'}(\fPhi{{w'}}{\gamma,K_1}) & 
      \equiv 
      \beta^{K_1} \vee 
      \big(\tpp^{K_2}\wedge\beta^{K_2}\wedge\alpha^{K_1}\big)
      \\
      \theta_w(\fPsi{w}{\gamma,J_0}) & 
      \equiv 
      \tpp^{J_0}\wedge\beta^{J_0}\wedge(\overline{\tpp}^{J_0}\rightarrow\alpha^{J_0})
      &
      \theta_{w'}(\fPsi{w'}{\gamma,K_1}) & 
      \equiv 
      \beta^{K_1} \vee 
      \big(
      \tpp^{K_2}\wedge\beta^{K_2}\wedge(\overline{\tpp}^{K_2}\rightarrow\alpha^{K_2}) \wedge\alpha^{K_1}
      \big)
    \end{array}
  \end{equation*}
  Observe that the proposition~$\alpha^{J_0}$ does not occur
  in~$\theta_w(\fPhi{w}{\gamma,J_0})$. In contrast, $\alpha^{J_0}$
  occurs in~$\theta_w(\fPsi{w}{\gamma,J_0})$.
  With the subformula~$\overline{\tpp}^{J_0}\rightarrow\alpha^{J_0}$,
  we store information about $\alpha$'s truth value in $J_0$. In this
  example, since
  $\theta_w(\overline{\tpp}^{J_0}\rightarrow
  \alpha^{J_0})\equiv\overline{\tpp}^{J_0}\rightarrow\alpha^{J_0}$
  we know that $\alpha$'s truth value in $J_0$ is~$\unknown$.  If
  $\theta_w(\overline{\tpp}^{J_0}\rightarrow\alpha^{J_0})\equiv\true$,
  then we infer that $\alpha$'s truth value in $J_0$ is~$\true$, and if
  $\theta_w(\overline{\tpp}^{J_0}\rightarrow\alpha^{J_0})\equiv\neg\overline{\tpp}^{J_0}$,
  $\alpha$'s truth value in $J_0$ is~$\false$. This information
  is relevant when splitting $J_0$.  In particular, it allows us
  to obtain $\theta_{w'}(\fPsi{w'}{\gamma,K_1})$ from
  $\theta_w(\fPsi{w}{\gamma,J_0})$ because all propositions that occur
  in $\theta_{w'}(\fPsi{w'}{\gamma,K_1})$ originate from propositions
  that already occur in $\theta_w(\fPsi{w}{\gamma,J_0})$.
  \exampleendmark
\end{example}

The following lemma shows that Lemma~\ref{lem:propformula} carries
over to $\theta_w(\fPsi{w}{\gamma,J})$. We omit its straightforward proof.
\begin{lemma} 
  For $b\in\Two$, 
  \begin{equation*}
    \theta_w(\fPhi{w}{\gamma,J})\equiv b
    \qquad\text{iff}\qquad
    \theta_w(\fPsi{w}{\gamma,J})\equiv b
    \,.
    \end{equation*}
\end{lemma}

\subsubsection{State Variables}
\label{subsubsec:state}

The monitor's state consists of the global variable~\ls{observation}
and the global variables~\ls{gate$^{\gamma,J}$}, where $\gamma$ is a
nonatomic subformula of the monitored formula $\phi$ and $J$ is an interval.
The monitor stores in the state variable \ls{observation} its
knowledge about the system behavior.  This variable is updated in each
iteration according to the message received by the procedure
\ls{UpdateKnowledge} (cf. Section~\ref{subsec:main}). More concretely,
in the monitor's $i$th iteration, the state
variable~\ls{observation} equals the observation~$w_i$ of the valid
observation sequence~$(w_i)_{i\in\Nat}$.
The state variables~\ls{gate$^{\gamma,J}$} are used for the verdict
computation.  In particular, the monitor maintains the invariant
$\mathsf{gate}^{\gamma,J}\equiv \theta_{w_i}(\fPsi{w_i}{\gamma,J})$.
Because of the assumption that $\phi$ is not an atomic formula, the
monitor only needs to maintain state variables \ls{gate}$^{\gamma,K}$,
where $\gamma$ is not atomic. Atomic formulas $p$ only occur as
propositions $p^J$ in the propositional formulas.
We remark that we chose the variable name \ls{gate} since the
propositional formulas can be seen as logic gates in a combinational
circuit. Each such gate computes a Boolean operation, where the input
signals are the formula's propositions.

For the sake of simplicity, we do not explicitly remove irrelevant
state variables. Instead, we assume that they are automatically
``garbage collected.''  For instance, a state
variable~\ls{gate$^{\gamma,J}$} with $|J|=1$ becomes irrelevant when
it is semantically equivalent to a Boolean constant and its truth
value has been propagated and, when $\gamma$ is the monitored
formula~$\phi$, the verdict for $J$ has been output.  For simplicity,
we also do not discard any knowledge about the system behavior.
In practice, one would remove irrelevant information from the state
variable \ls{observation}, for example, isolated time points for which
the monitor has already output a verdict.

\label{def:API}
Instead of fixing a concrete representation of the propositional
formulas that are stored in the state variables
\ls{gate$^{\gamma,J}$}, we provide an abstract interface for accessing
and updating \ls{gate$^{\gamma,J}$}.  In Section~\ref{subsec:algo}, we
use this interface to describe the monitor's central algorithmic
details, which are independent from an actual representation of the
propositional formulas.  In Section~\ref{subsec:datastruct}, we
describe a graph-based data structure for implementing the interface
for the generalized setting with the freeze quantifier.  Note that
this presentation-independent description also allows us to separate
concerns in the monitor's correctness proof
(cf. Section~\ref{subsec:correct}).
The interface comprises the following procedures.
\begin{itemize}[--]
\item \ls{Clone(gate)}: returns a copy of \ls{gate}.
\item \ls{IsBool(gate)}: returns true iff \ls{gate} is semantically
  equivalent to a Boolean constant.
\item \ls{ToBool(gate)}: returns the Boolean value $b\in\Two$,
  provided that \ls{gate} is semantically equivalent to the
  corresponding Boolean constant.
\item \ls{Contains(gate, $p$)}: returns true iff \ls{gate} depends on
  the proposition $p$, that is,
  $[p\mapsto\true](\mathsf{gate})\not\equiv
  [p\mapsto\false](\mathsf{gate})$.  We shall abuse terminology in the
  following by also saying that $p$ occurs in \ls{gate}, although the
  occurrence of a proposition in a formula can be representation
  dependent.
\item \ls{Eval(gate, $\theta$)}: applies the substitution $\theta$ to
  \ls{gate}, where $\theta$ only replaces propositions $p$ with one of
  the Boolean constants~$\true$ or~$\false$.
\item \ls{Instantiate(gate)}: substitutes Boolean constants for the
  propositions of the form $\tpp^L$, $\overline{\tpp}^L$,
  $\tcp^{L,L'}_I$, and $(\true)^{L}$ in \ls{gate}, wherever
  possible. Note that the Boolean constants for these propositions can
  be determined by their name.  For instance, $\tcp^{L,L'}_I$ is
  replaced by $\true$ iff $L-L'\not=\emptyset$ and $L-L'\subseteq I$,
  and by $\false$ iff $I\cap(L-L')=\emptyset$. Furthermore, note that
  \ls{Instantiate} is a special case of \ls{Eval}.
\item \ls{Rename(gate, $\theta$)}: applies the substitution $\theta$
  to \ls{gate}, where $\theta$ only renames propositions $p$ with some
  proposition $p'$.
\end{itemize}
For the following last two interface procedures \ls{Add} and
\ls{Remove}, we first introduce the following additional notion.  The
propositional formulas $\theta_w(\fPsi{w}{\gamma,J})$ and hence also
\ls{gate$^{\gamma,J}$} can be grouped into subformulas with respect to
an interval and a direct subformula of $\gamma$.  For instance, note
that \ls{gate$^{\alpha\until_I\beta,J}$} is semantically equivalent to
$\bigvee_{K\geq J}(\mathit{anchor}_K\land\bigwedge_{H\leq
  L<K}\mathit{continuation}_L)$, with
$\mathit{anchor}_K = \tpp^K \land \tcp^{K,J}_I \land \beta^K \land
(\overline{\tpp}^K\to\alpha^K)$ and
$\mathit{continuation}_L = \tpp^L\to\alpha^L$, possibly with some of
their literals replaced by Boolean constants, and where the
intervals~$K$ and~$L$ range over the intervals of the letters
in~\ls{observation}.
The \emph{$(\beta,K)$-relevant} part of
\ls{gate$^{\alpha\until_I\beta,J}$} is the propositional formula
$\mathit{anchor}_K$.
The \emph{$(\alpha,L)$-relevant} part of
\ls{gate$^{\alpha\until_I\beta,J}$} is the propositional formula
$\mathit{continuation}_L$ if $L\not=J$, and, if $L=J$, the
subformula~$\overline{\tpp}^J\rightarrow\alpha^J$ of
$\mathit{anchor}_J$, possibly with some propositions replaced by
Boolean constants.  Note that the relevant part can be~$\true$
or~$\false$.
For example, for the formula $\theta_{w'}(\fPsi{w'}{\gamma,K_1})$ in
Example~\ref{ex:alphaK}, the $(\beta,K_1)$-relevant part is
$\beta^{K_1}$, the $(\beta,K_2)$-relevant part is
$\tpp^{K_2}\wedge\beta^{K_2}\wedge(\overline{\tpp}^{K_2}\rightarrow\alpha^{K_2})$,
and the $(\beta,H)$-relevant part is $\false$, for $H$ being $K_0$,
$J_1$, or $J_2$.  Furthermore, the $(\alpha,K_2)$-relevant part is
$\overline{\tpp}^{K_2}\rightarrow\alpha^{K_2}$ and the
$(\alpha,H)$-relevant part is $\true$, for $H$ being $K_0$, $K_1$,
$J_1$, or $J_2$.
The definition of the relevant parts for other formulas $\gamma$ is as
expected and omitted. For instance, the $(\alpha,J)$-relevant part of
$\mathsf{gate}^{\alpha\vee\beta,J} = \alpha^J$ is $\alpha^J$ and its
$(\beta,J)$-relevant part is $\false$.
Note that for temporal formulas the relevant
parts are always defined for their direct subformulas. For nontemporal
formulas, the relevant parts are only defined for their direct
subformulas and when the intervals match. 
\begin{itemize}[--]
\item \ls{Add(gate, $\alpha$, $K$, $\Theta$)}: replaces the
  $(\alpha,K)$-relevant part of \ls{gate} with the propositional
  formula $\Theta$.  We require that $\Theta$ is of the form of the
  relevant parts of \ls{gate}.
  
\item \ls{Remove(gate, $\alpha$, $K$)}: returns the
  $(\alpha,K)$-relevant part of \ls{gate} and ``removes'' it from
  \ls{gate}. For anchors, the removal corresponds to a replacement
  with the Boolean constant~$\false$. For continuations, the relevant
  part is replaced by the Boolean constant~$\true$.
\end{itemize}

\subsubsection{Initialization}
\label{subsubsec:init}

\begin{listing}  
  \begin{minipage}[t]{0.45\linewidth}
    \begin{lstlisting}[caption={Initialization procedure.},label={lst:init},captionpos=b]
procedure Init($\phi$)
  observation := $w_0$
  foreach $\gamma$ with $\gamma\in\sub(\phi)$ and not IsAtom($\gamma)$ do 
    # Iterate top down with respect to $\color{blue}\phi$'s formula structure.
    $\textsf{gate}^{\gamma,[0,\infty)}$ := $\fPsi{w_0}{\gamma,[0,\infty)}$
    Instantiate(gate$^{\gamma,[0,\infty)}$)
    if IsBool(gate$^{\gamma,[0,\infty)}$) then 
      PropagateTruthValue($\gamma$, $[0,\infty)$, ToBool(gate$^{\gamma,[0,\infty)}$))
\end{lstlisting}
  \end{minipage}
\end{listing}
The procedure \ls{Init}, shown in Listing~\ref{lst:init}, initializes
the state variables.  Initially, \ls{observation} is the word
$w_0$. Recall that $w_0=\big([0,\infty),(\emap,\emap)\big)$.
Furthermore, for $\gamma\in\sub(\phi)$, \ls{Init} initializes
\ls{gate$^{\gamma,[0,\infty)}$} with the propositional formula
$\theta_{w_0}(\fPsi{w_0}{\gamma,[0,\infty)})$.
For this, the \ls{Init} procedure uses the interface procedure
\ls{Instantiate} and the procedure \ls{PropagateTruthValue}, which we
present in Section~\ref{subsubsec:prop_bool}, for propagating Boolean
values up the formula structure.  Since the formula is traversed
top-down, Boolean truth values are always propagated to already
initialized state variables.

\subsection{Algorithmic Details}
\label{subsec:algo}

In the following, we provide algorithmic details for the monitor's
central procedures \ls{AddTimePoint}, \ls{RemoveInterval}, and
\ls{PropagateTruthValue}.  Recall from Section~\ref{subsec:main} that
each of these procedures updates the monitor's state, in particular,
the propositional formulas stored in the \ls{gate} variables according
to one of the
transformations~(T\ref{enum:observation_split}),~(T\ref{enum:observation_removal}),
and~(T\ref{enum:observation_data}).

\begin{listing}
  \begin{minipage}[t]{0.48\linewidth}
    \begin{lstlisting}[caption={Procedure for transformation~(T\ref{enum:observation_split}).},label={lst:tp},captionpos=b]
procedure AddTimePoint($\phi$, $J$, $\tau$)
  foreach gate$^{\gamma,J}$ and $K\in\{J\cap[0,\tau),\{\tau\},J\cap(\tau,\infty)\}$ do
    $\textsf{gate}^{\gamma,K}$ := Clone($\textsf{gate}^{\gamma,J}$)
  Delete($J$) 
  if IsBool($\mathsf{gate}^{\phi,\{\tau\}}$) then
    OutputVerdict($\{\tau\}$, ToBool($\mathsf{gate}^{\phi,\{\tau\}}$))
  foreach gate$^{\gamma,H}$ with Contains(gate$^{\gamma,H}$, $p$), 
                            $\textsf{for}$ some proposition $p$ $\textsf{with}$ the interval $J$ do
    # Iterate top down with respect to $\color{blue}\phi$'s formula structure.
    case $\gamma = \neg\alpha$: $\hspace{.395cm}$Rename($\mathsf{gate}^{\gamma,H}$, $[\alpha^J\mapsto \alpha^H]$)
    case $\gamma = \alpha\lor\beta$: $\hspace{.13cm}$Rename($\mathsf{gate}^{\gamma,H}$, $[\alpha^J\mapsto \alpha^H, \beta^J\mapsto \beta^H]$)
    case $\gamma = \predecessor_I\alpha$: $\hspace{.21cm}$ $\dots$ # Omitted; analogous to the next case.
    case $\gamma = \successor_I\alpha$: $\hspace{.21cm}$RefineNext($\gamma$, $H$, $J$, $\tau$)
                      Instantiate($\mathsf{gate}^{\gamma,H}$) 
    case $\gamma = \alpha\since_I\beta$: $\hspace{.05cm}$ $\dots$ # Omitted; analogous to the next case.
    case $\gamma = \alpha\until_I\beta$: $\hspace{.05cm}$RefineUntil($\gamma$, $H$, $J$, $\tau$)
                      Instantiate($\mathsf{gate}^{\gamma,H}$)
    if IsBool(gate$^{\gamma,H}$) then
      PropagateTruthValue($\gamma$, $H$, ToBool(gate$^{\gamma,H}$))
    \end{lstlisting}
  \end{minipage}
  \hfill
  \begin{minipage}[t]{0.48\linewidth}
    \begin{lstlisting}[caption={Procedure for transformation~(T\ref{enum:observation_removal}).},label={lst:remove},captionpos=b]
procedure RemoveInterval($K$)
  foreach gate$^{\gamma,J}$ with $J\not=K$ and Contains(gate$^{\gamma,J}$, $\tpp^K$) do
    Eval(gate$^{\gamma,J}$, $[\tpp^K\mapsto\false]$)
    if IsBool(gate$^{\gamma,J}$) then
      PropagateTruthValue($\gamma$, $J$, ToBool(gate$^{\gamma,J}$))
  Delete($K$) 
    \end{lstlisting}

    \begin{lstlisting}[caption={Procedure for transformation~(T\ref{enum:observation_data}).},label={lst:propagateup},captionpos=b]
procedure PropagateTruthValue($\alpha$, $J$, $b$)
  if IsRoot($\alpha$) and $|J|=1$ then 
    OutputVerdict($J$, $b$)
  else if not IsRoot($\alpha$) then
    foreach gate$^{\gamma,K}$ with Contains(gate$^{\gamma,K}$, $\alpha^{J}$) do
      Eval(gate$^{\gamma,K}$, $[\alpha^J\mapsto b]$)
      if IsBool(gate$^{\gamma,J}$) then
        PropagateTruthValue($\gamma$, $J$, ToBool(gate$^{\gamma,J}$))
    \end{lstlisting}
  \end{minipage}
\end{listing}

\subsubsection{Adding a Time Point.}
\label{subsubsec:add}

The pseudocode of the procedure \ls{AddTimePoint} is given in
Listing~\ref{lst:tp}.  Suppose that the respective transformation
splits the interval $J$ at $\tau$.  For $\tau>0$, we obtain the
new intervals $L:=J\cap[0,\tau)$, $T:=\{\tau\}$, and
$R:=J\cap(\tau,\infty)$.  For brevity, we do not present the details
for the corner case $\tau=0$, where we only obtain two new
intervals, $\{0\}$ and $J\setminus\{0\}$.
We first create for each $\gamma\in\sub(\phi)$ three copies of
$\textsf{gate}^{\gamma,J}$. Namely, we create the propositional
formulas~$\textsf{gate}^{\gamma,L}$, $\textsf{gate}^{\gamma,T}$, and
$\textsf{gate}^{\gamma,R}$.  Afterwards, we remove all the
propositional formulas for the interval~$J$ from the monitor's state,
that is, we delete \ls{gate$^{\gamma,J}$}, for all
$\gamma\in\sub(\phi)$.
We then handle the special case where $\textsf{gate}^{\phi,\{\tau\}}$ is
semantically equivalent to a Boolean constant, which results in
outputting a verdict for the new time point.

\begin{listing}  

  \begin{minipage}[t]{0.48\linewidth}
    \begin{lstlisting}[caption={Auxiliary procedure for (T1) and the connective~$\successor_I$.},label={lst:auxNEXT},captionpos=b]
procedure RefineNext($\successor_I\alpha$, $H$, $J$, $\tau$)
  $\gamma$, $L$, $T$, $R$ := $\successor_I\alpha$, $J\cap[0,\tau)$, $\set{\tau}$, $J\cap(\tau,\infty)$
  $C_0$, $C_1$:= Remove(gate$^{\gamma,H}$, $\alpha$, $H$), Remove(gate$^{\gamma,H}$, $\alpha$, $H+1$)
  Remove(gate$^{\gamma,H}$, $\alpha$, $H+2$)
  if $H\not\subseteq J$ then  # Note that $\color{blue}J=H+1$, $\color{blue}|H|=1$, and $\color{blue}C_0\equiv\false$.
    Add(gate$^{\gamma,H}$, $\alpha$, $L$, $[\tcp^{J,H}_I\mapsto\tcp^{L,H}_I,\tpp^J\mapsto\tpp^{L},\alpha^J\mapsto\alpha^{L}](C_1)$)
    Add(gate$^{\gamma,H}$, $\alpha$, $T$, $[\tcp^{J,H}_I\mapsto\tcp^{T,H}_I,\tpp^J\mapsto\neg\tpp^{L},\alpha^J\mapsto\alpha^{T}](C_1)$)
  else if $H=L$ then
    Add(gate$^{\gamma,H}$, $\alpha$, $L$, $[\tcp^{J,J}_I\mapsto\tcp^{L,L}_I,\tpp^J\mapsto\tpp^{L},\alpha^J\mapsto\alpha^{L}](C_0)$)
    Add(gate$^{\gamma,H}$, $\alpha$, $T$, $[\tcp^{J,J}_I\mapsto\tcp^{T,L}_I,\tpp^J\mapsto\neg\tpp^{L},\alpha^J\mapsto\alpha^{T}](C_0)$)
  else if $H=T$ then
    Add(gate$^{\gamma,H}$, $\alpha$, $R$, $[\tcp^{J,J}_I\mapsto\tcp^{R,T}_I,\tpp^J\mapsto\neg\tpp^{R},\alpha^J\mapsto\alpha^{R}](C_0)$)
    Add(gate$^{\gamma,H}$, $\alpha$, $H+2$, $[\tcp^{H+2,J}_I\mapsto\tcp^{H+2,T}_I,\tpp^J\mapsto\neg\tpp^{R}](C_1)$)
  else if $H=R$ then
    Add(gate$^{\gamma,H}$, $\alpha$, $R$, $[\tcp^{J,J}_I\mapsto\tcp^{R,R}_I,\tpp^J\mapsto\tpp^{R},\alpha^J\mapsto\alpha^{R}](C_0)$)
    Add(gate$^{\gamma,H}$, $\alpha$, $H+2$, $[\tcp^{J+1,J}_I\mapsto\tcp^{H+2,R}_I,\tpp^J\mapsto\tpp^{R}](C_1)$)
    \end{lstlisting}
  \end{minipage}
  \hfill
  \begin{minipage}[t]{0.48\linewidth}
    \begin{lstlisting}[caption={Auxiliary procedure for (T1) and the connective~$\until_I$.},label={lst:auxUNTIL},captionpos=b]
procedure RefineUntil($\alpha\until_I\beta$, $H$, $J$, $\tau$)
  $\gamma$, $L$, $T$, $R$ := $\alpha\until_I\beta$, $J\cap[0,\tau)$, $\set{\tau}$, $J\cap(\tau,\infty)$
  anchor, continuation := Remove($\mathsf{gate}^{\gamma,H}$, $\beta$, $J$), Remove($\mathsf{gate}^{\gamma,H}$, $\alpha$, $J$)
  foreach $K$ with $K\in\{L,T,R\}$ and $H\leq K$ do
    $\theta$ := $[\alpha^J\mapsto\alpha^K,\beta^J\mapsto\beta^K,\tpp^J\mapsto\tpp^K]\,\cup$
                      $\![\tcp^{J,L}_I\mapsto\tcp^{K,L}_I\mathbin{|}\text{for some interval }L]$
    Add($\mathsf{gate}^{\gamma,H}$, $\beta$, $K$, $\theta[\overline{\tpp}^J\mapsto\overline{\tpp}^K](\mathsf{anchor})$)
    Add($\mathsf{gate}^{\gamma,H}$, $\alpha$, $K$, $\theta[\overline{\tpp}^J\mapsto\tpp^K](\mathsf{continuation})$)
  if $H\subseteq J$ then
    Rename($\mathsf{gate}^{\gamma,H}$, $[\tcp^{L,J}_I\mapsto\tcp^{L,H}_I\mathbin{|}\text{for some interval }L]$)
    \end{lstlisting}
  \end{minipage}
\end{listing}

Finally, we update the propositional formulas
$\textsf{gate}^{\gamma,H}$ in which a proposition with the
interval~$J$ occurs.  Note that $H$ can be different from $L$, $T$,
and $R$ when $\gamma$ is a temporal formula.
We make a case split on $\gamma$'s form.  The cases $p$, $\neg\alpha$,
and $\alpha\vee\beta$ are obvious.
We replace any proposition with the interval~$J$ by the corresponding
proposition with the interval~$H$.
Let us turn to the case where $\gamma$ is of the form
$\alpha\until_I\beta$. We omit the dual case $\alpha\since_I\beta$.
The procedure \ls{RefineUntil}, shown in
Listing~\ref{lst:auxUNTIL}, updates the anchor and continuation subformulas of
the propositional formula $\textsf{gate}^{\gamma,H}$, that is, the
$(\beta,J)$-relevant and $(\alpha,J)$-relevant parts of
$\textsf{gate}^{\gamma,H}$.
\begin{itemize}[--]
\item If $H>J$, then we replace the $(\beta,J)$-relevant part with the
  three relevant parts for the intervals $L$, $T$, and $R$. They
  originate from the $(\beta,J)$-relevant part.  Similarly, we replace
  the $(\alpha,J)$-relevant parts.
\item If $H$ is one of the intervals $L$, $T$, or $R$, then we replace
  the relevant parts with the interval $J$ up to the interval $H$.
  Furthermore, we need to adjust the $J$ interval in the $\tcp$
  propositions in $\textsf{gate}^{\gamma,H}$.
\end{itemize}
After \ls{RefineUntil}, \ls{AddTimePoint} calls \ls{Instantiate} to
replace propositions with Boolean constants where possible.  Note that
after the instantiation, $\textsf{gate}^{\gamma,H}$ can be
semantically equivalent to a Boolean constant.  The \ls{if} statement
at the end of the second \ls{foreach} loop performs the corresponding
check and triggers the propagation.

Finally, let us consider the case where $\gamma$ is of the form
$\successor_I\alpha$.  We omit the dual case $\predecessor_I\alpha$.
The procedure \ls{RefineNext}, shown in Listing~\ref{lst:auxNEXT},
updates a propositional formula $\textsf{gate}^{\gamma,H}$ as follows.
It first removes all its relevant parts. Note that these have one of
the three intervals $H$, $H+1$, and $H+2$. Then, depending on whether
$H\not\subseteq J$ or $H$ is one of the intervals $L$, $T$, or $R$,
\ls{RefineNext} adds the new relevant parts to
$\textsf{gate}^{\gamma,H}$.  These parts originate from the old
relevant parts with the intervals $H$ and $H+1$.

\subsubsection{Removing an Interval.}
\label{subsubsec:remove}

The pseudocode of the procedure \ls{RemoveInterval} is given in
Listing~\ref{lst:remove}. Since $K$ does not contain any time points,
we replace any occurrence of the proposition $\tpp^K$ by $\false$.  It
suffices to only update propositional formulas \ls{gate$^{\gamma,J}$},
where $\gamma$ is a temporal formula and $J\not=K$. Note that this
replacement could trigger the propagation of Boolean values. For
instance, we propagate $\false$ from
\ls{gate$^{\alpha\until_I\beta,J}$}, if $K$ is the only anchor in
\ls{gate$^{\alpha\until_I\beta,J}$}, that is, if for all $H\neq K$,
the $(\beta,H)$-relevant part of \ls{gate$^{\alpha\until_I\beta,J}$}
is~$\false$.  Afterwards, we delete all the propositional formulas
\ls{gate$^{\gamma,K}$} with $\gamma\in\sub(\phi)$ from the monitor's
state.

\subsubsection{Propagating a Boolean Value.}
\label{subsubsec:prop_bool}

The pseudocode of the procedure \ls{PropagateTruthValue} is given in
Listing~\ref{lst:propagateup}.  The procedure is called whenever a
propositional formula $\textsf{gate}^{\alpha,J}$ simplifies to a
Boolean constant~$b$.  If $|J|=1$ and $\alpha=\phi$, then we output a
verdict.  Otherwise, for $\alpha\not=\phi$, we substitute the
proposition $\alpha^J$ with its Boolean value~$b$ in all the
propositional formulas \ls{gate$^{\gamma,K}$}.  Note that $\gamma$
must be the parent formula of $\alpha$. However, for temporal
formulas, $K$ can be different from $J$. We continue the propagation
whenever the updated \ls{gate$^{\gamma,K}$} propositional formula is
semantically equivalent to a Boolean constant.

\section{Monitoring with Data Values}
\label{sec:data}

In this section, we extend the online algorithm from
Section~\ref{sec:prop} for MTL to \MTLdata.  The handling of the
freeze quantifier is orthogonal to the core ideas already used for
monitoring MTL specifications.  In Section~\ref{subsec:extalgo}, we
present the extension.  In Section~\ref{subsec:datastruct}, we
describe the graph-based data structure for representing and
manipulating the propositional formulas of the monitor's state.
Finally, in Section~\ref{subsec:correct}, we establish the algorithm's
correctness.  Note that the data structure and the correctness proof
also apply to the restricted setting of Section~\ref{sec:prop}, that
is, the online algorithm for MTL.

\subsection{Algorithmic Details}
\label{subsec:extalgo}

Throughout this section, we reuse the conventions that we introduced
in Section~\ref{sec:prop} for MTL.  Analogous to
Section~\ref{sec:prop}, we also assume that $\phi$ is not an atomic
formula and subformulas of $\phi$ are pairwise distinct.  Furthermore,
we require that the monitored formula $\phi$ is closed.  Finally, we
assume that variables are frozen at most once in $\phi$.  This
assumption is also without loss of generality and it allows us to
identify a frozen variable with the respective subformula.

\subsubsection{Reduction to Propositional Logic.}

Similar to MTL, at the core of the monitor for \MTLdata is a mapping
of \MTLdata's three-valued semantics into propositional logic.  To
this end, we first extend the propositional formulas
$\fPsi{w}{\gamma,J}$ from Section~\ref{subsec:state} to
$\fPsi{w}{\gamma,J,\nu}$ to capture \MTLdata's freeze quantifier.  We
remark that the only change in the definition below is that each
proposition for a subformula $\alpha\in\sub(\phi)$ now also carries a
partial valuation~$\nu$ in addition to an interval~$K$.
\begin{equation*}
  \fPsi{w}{\gamma,J,\nu} := 
    \begin{cases}
      \true 
      & \text{if $\gamma=\true$}
      \\
      p(\overline{x})^{J,\nu}
      & \text{if $\gamma=p(\overline{x})$ with $p\in P$}
      \\
      \alpha^{J,\nu[x\mapsto\unknown]}
      & \text{if $\gamma=\freeze{r}{x}\alpha$ and $r\not\in\pdef(\rho_j)$}
      \\
      \alpha^{J,\nu[x\mapsto\rho_j(r)]}
      & \text{if $\gamma=\freeze{r}{x}\alpha$ and $r\in\pdef(\rho_j)$}
      \\
      \neg\alpha^{J,\nu} 
      & \text{if $\gamma=\neg\alpha$}
      \\
      \alpha^{J,\nu}\vee\beta^{J,\nu} 
      & \text{if $\gamma=\alpha\vee\beta$}
      \\
      \bigvee_{K\leq J}\big(\tpp^{K} \wedge \tcp^{J,K}_I \wedge 
      \beta^{K,\nu} \wedge (\overline{\tpp}^K\rightarrow\alpha^{K,\nu}) \wedge
      \bigwedge_{K<H\leq J}(\tpp^{H} \to \alpha^{H,\nu})
      \big)
      & \text{if $\gamma=\alpha\since_I\beta$}
      \\
      \bigvee_{K\geq J}\big(\tpp^{K} \wedge \tcp^{K,J}_I \wedge
      \beta^{K,\nu} \wedge (\overline{\tpp}^K\rightarrow\alpha^{K,\nu}) \wedge
      \bigwedge_{J\leq H<K}(\tpp^{H} \to \alpha^{H,\nu})
      \big)
      & \text{if $\gamma=\alpha\until_I\beta$}
      \\
      C_w^{J,0,\nu}\vee C_w^{J,-1,\nu}\vee C_w^{J,-2,\nu}
      & \text{if $\gamma=\predecessor_I\alpha$}
      \\
      C_w^{J,0,\nu}\vee C_w^{J,+1,\nu}\vee C_w^{J,+2,\nu}
      & \text{if $\gamma=\successor_I\alpha$}
    \end{cases}
\end{equation*}
with
\begin{equation*}
  C_w^{J,\pm\ell,\nu}:=
  \begin{cases} 
    \tcp^{J,J}_I\wedge\alpha^{J,\nu}\wedge\neg\tpp^J
    & \text{if $\ell=0$ and $I\not=\{0\}$}
    \\
    \tcp^{\max\{J,J\pm1\},\min\{J,J\pm1\}}_I\wedge\alpha^{J\pm1,\nu}\wedge\tpp^J\wedge\tpp^{J\pm1}
    & \text{if $\ell=\pm1$ and $j\pm1\in\pos(w)$}
    \\
    \tcp^{\max\{J,J\pm2\},\min\{J,J\pm2\}}_I\wedge\alpha^{J\pm2,\nu}\wedge\neg\tpp^{J\pm1}
    & \text{if $\ell=\pm2$ and $j\pm2\in\pos(w)$}
    \\
    \false 
    & \text{otherwise}
  \end{cases}
\end{equation*} 
In the following, let $\theta_w$ be the substitution that maps a
proposition of the form $\alpha^{K,\nu}$ to $\osem{w,k,\nu}{\alpha}$,
provided that $\osem{w,k,\nu}{\alpha}\in\Two$, and for the other
propositions, $\theta_w$ is as in Section~\ref{sec:prop}, namely,
\begin{gather*}
  \tpp^K \mapsto \true  
  \quad\text{if $|K|=1$,}
  \qquad\quad
  \tcp^{H,K}_I \mapsto \begin{cases}
    \true  &\text{if $H-K\not=\emptyset$ and $H-K\subseteq I$,} \\
    \false &\text{if $(H-K)\cap I=\emptyset$, and} 
  \end{cases}
  \qquad\quad
  \overline{\tpp}^K \mapsto \false  
  \quad\text{if $|K|=1$.}
\end{gather*}

\begin{example}
  \label{ex:freeze}
  We illustrate the definition of the propositional formulas
  $\fPsi{w}{\gamma,J,\nu}$.  Consider the formula
  $\freeze{r}{x}\alpha$ with $\alpha=\eventually_{(0,1]} p(x)$.  
  For readability, we use $\beta$ for the
  subformula $p(x)$.
  Let $w$ be the observation
  $\big([0,\tau),(\emap,\emap)\big)\big(\{\tau\},(\emap,[r\mapsto
  d])\big)\big((\tau,\infty),(\emap,\emap)\big)$.
  We obtain the following propositional formulas, where
  $J_0=[0,\tau)$, $J_1=\{\tau\}$, $J_2=(\tau,\infty)$, and
  $\bar\beta^{J,H,\nu}$ abbreviates the conjunction
  $\tpp^{H} \wedge \tcp^{J,H}_I \wedge \beta^{H,\nu}$.
  \begin{equation*}
    \begin{array}{r@{\ }l@{\qquad}r@{\ }l@{\qquad}r@{\ }l}
      \theta_w(\fPsi{w}{\freeze{r}{x}\alpha,J_0,\emap})   & \equiv \alpha^{J_0,\emap}
      &
      \theta_w(\fPsi{w}{\freeze{r}{x}\alpha,J_1,\emap})   & \equiv \alpha^{J_1,[x\mapsto d]}  
      &
      \theta_w(\fPsi{w}{\freeze{r}{x}\alpha,J_2,\emap})   & \equiv \alpha^{J_2,\emap} 
      \\
      \theta_w(\fPsi{w}{\alpha,J_0,\emap}) & \equiv \bar\beta^{J_0,J_0,\emap} \lor \bar\beta^{J_1,J_0,\emap} \lor \bar\beta^{J_2,J_0,\emap}
      &
      \theta_w(\fPsi{w}{\alpha,J_1,[x\mapsto d]}) & \equiv \bar\beta^{J_2,J_1,[x\mapsto d]}
      &
      \theta_w(\fPsi{w}{\alpha,J_2,\emap}) & \equiv \bar\beta^{J_2,J_2,\emap}  
   \end{array}
  \end{equation*}
  First, note that as $\eventually_{(0,1]}\beta$ is syntactic sugar
  for $\true\until_{(0,1]}\beta$, we can ignore the continuation
  subformulas in the propositional formulas, since they simplify to
  $\true$.  Furthermore, note that
  $\theta_w(\fPsi{w}{\alpha,J_1,[x\mapsto d]})$ has only one anchor
  subformula (different from a propositional constant), since
  $\alpha$'s temporal constraint $(0,1]$ is unsatisfiable for $J_1$,
  that is, $\theta_w(\tcp^{J_1,J_1}_{(0,1]}) = \false$ and hence
  $\theta_w(\bar\beta^{J_1,J_1,[x\mapsto d]})\equiv\false$.
  Finally, note that for $\beta$ and $J_2$, we have the two
  propositions $\beta^{J_2,\emap}$ and $\beta^{J_2,[x\mapsto d]}$,
  where $\beta^{J_2,\emap}$ occurs in both
  $\theta_w(\fPsi{w}{\alpha,J_0,\emap})$ and
  $\theta_w(\fPsi{w}{\alpha,J_2,\emap})$, and
  $\beta^{J_2,[x\mapsto d]}$ occurs once in
  $\theta_w(\fPsi{w}{\alpha,J_1,[x\mapsto d]})$.  \exampleendmark
\end{example}

Lemma~\ref{lem:propformula} for MTL carries over to \MTLdata and its
propositional formulas $\fPsi{w}{\gamma,J,\nu}$.
\begin{lemma}
  \label{lem:propformulaMTLdata}
  The following two statements hold.
  \begin{enumerate}[(1)]
  \item If $\osem{w,j,\nu}{\gamma}\in\Two$, then 
    $\theta_w(\fPsi{w}{\gamma,J,\nu})\equiv \osem{w,j,\nu}{\gamma}$.
  \item If $\osem{w,j,\nu}{\gamma}=\unknown$, then
    $\theta_w(\fPsi{w}{\gamma,J,\nu})\not\equiv\true$ and 
    $\theta_w(\fPsi{w}{\gamma,J,\nu})\not\equiv\false$.
  \end{enumerate}
\end{lemma}
\begin{proof}
  We prove the lemma by a case split on $\gamma$.  The case
  $\gamma=\true$ is obvious and omitted.  For the case
  $\gamma=p(\overline{x})$, we have by definition that
  $\fPsi{w}{p(\overline{x}),J,\nu}=p(\overline{x})^{J,\nu}$.
  Furthermore, when $\osem{w,j,\nu}{p(\overline{x})}\in\Two$, we have
  that
  $\theta_w(p(\overline{x})^{J,\nu})=\osem{w,j,\nu}{p(\overline{x})}$;
  otherwise, $\theta_w$ is not defined for $p(\overline{x})^{J,\nu}$.
  We conclude that both implications~(1) and~(2) hold.
  The cases for $\neg\alpha$ and $\alpha\vee\beta$ are similar and
  omitted.  We also omit the details of the case for
  $\gamma=\freeze{r}{x}\alpha$ as it is also similar to the case
  $p(\overline{x})$.  Note that for $r\not\in\pdef(\rho_j)$, we have
  that
  $\fPsi{w}{\freeze{r}{x}\alpha,J,\nu}=\alpha^{J,\nu[x\mapsto\unknown]}$
  and
  $\fPsi{w}{\freeze{r}{x}\alpha,J,\nu}=\alpha^{J,\nu[x\mapsto\rho_j(r)]}$,
  for $r\in\pdef(\rho_j)$.

  Finally, we provide proof details for the case
  $\gamma=\alpha\since_I\beta$. The cases for the other three temporal
  connectives are similar and omitted.
  For $K$ with $K\leq J$, the disjunction of the propositional
  formulas
  $\Psi:=\tpp^{K} \wedge \tcp^{J,K}_I \wedge \beta^{K,\nu} \wedge
  (\overline{\tpp}^K\rightarrow\alpha^{K,\nu}) \wedge
  \bigwedge_{K<H\leq J}(\tpp^{H} \to \alpha^{H,\nu})$ of
  $\fPsi{w}{\alpha\since_I\beta,J,\nu}$ follows closely the semantic
  definition of $\osem{w,j,\nu}{\alpha\since_I\beta}$ in
  Section~\ref{subsec:semantics}, except that each $\Psi$ contains the
  additional propositional
  subformula~$\overline{\tpp}^K\rightarrow\alpha^{K,\nu}$.
  First, observe that the propositions~$\tpp^K$ and~$\tcp^{J,K}_I$
  together with the substitution $\theta_w$ take the role of
  $\istp_w(k)$ and $\tc_{w,I}(j,k)$.  Furthermore, $\theta_w$ replaces
  the propositions $\beta^{K,\nu}$ and $\alpha^{H,\nu}$ in $\Psi$ with
  the corresponding Boolean constants whenever the respective formulas
  evaluate to Boolean truth values under the three-valued semantics
  $\llbracket\cdot\rrbracket$.
  Finally, we observe that the additional propositional subformula
  $\overline{\tpp}^K\rightarrow\alpha^{K,\nu}$ in $\Psi$ simplifies to
  $\true$ when $|K|=1$, since $\theta_w(\overline{\tpp}^K)=\false$.
  For $|K|>1$, $\overline{\tpp}^K\rightarrow\alpha^{K,\nu}$ simplifies
  either to $\overline{\tpp}^K\rightarrow\alpha^{K,\nu}$,
  $\overline{\tpp}^K$, or $\true$.  Since $\theta_w$ is also not
  defined for $\tpp^K$, when $|K|>1$, $\theta_w(\Psi)\not\equiv\true$.
  Furthermore, $\theta_w(\Psi)\not\equiv\false$ if
  $\theta_w(\beta^{K,\nu})\not\equiv\false$ and
  $\theta_w(\alpha^{H,\nu})\not\equiv\false$, for all $H$ with $|H|=1$
  and $K<H\leq J$.
  With these observations, it is easy to see that the implications~(1)
  and~(2) hold.
\end{proof}

\subsubsection{Main Procedure.}

\begin{listing}
  \begin{minipage}[t]{0.48\linewidth}
\begin{lstlisting}[caption={The monitor's main loop for \MTLdata.},label={lst:monitor_data},captionpos=b]
procedure MonitorMTL$^{\downarrow}$($\phi$)
  Init($\phi$)
  loop
    $m$ := ReceiveMessage()
    $\mathit{ts}$ := UpdateKnowledge($m$)
    foreach $t$ in $\mathit{ts}$ do
      case (T1): $\ \ \tau$, $J$ := DeltaT1($t$)
                 AddTimePoint($\phi$, $J$, $\tau$)
      case (T2): $\ \ K$ := DeltaT2($t$)
                 RemoveInterval($K$)  
      case (T3.1): $\tau$, $\sigma$ := DeltaT31($t$)
                 foreach $p(\overline{x})^{\{\tau\},\nu}$ with $p\in\pdef(\sigma)$ and $\overline{x}\in\pdef(\nu)$ do
                   PropagateTruthValue($p(\overline{x})$, $\{\tau\}$, $\nu$, $\nu(\overline{x})\in\sigma(p)$)
      case (T3.2): $\tau$, $\rho$ := DeltaT32($t$)
                 foreach $\alpha^{\{\tau\},\nu}$ with $\freeze{r}{x}\alpha\in\sub(\phi)$ and $r\in\pdef(\rho)$ do
                   Rename(gate$^{\freeze{r}{x}\alpha,\{\tau\},\nu}$, $[\alpha^{\{\tau\},\nu}\mapsto\alpha^{\{\tau\},\nu[x\mapsto\rho(r)]}]$)
                   PropagateDataValue($\alpha$, $\{\tau\}$, $\nu$, $[x\mapsto\rho(r)]$)
\end{lstlisting}
  \end{minipage}
  \hfill
  \begin{minipage}[t]{0.48\linewidth}
\begin{lstlisting}[caption={Initialization procedure.},label={lst:init_data},captionpos=b]
procedure Init($\phi$)
  observation := $w_0$
  foreach $\gamma\in\sub(\phi)$ with not IsAtom($\gamma$) do 
    # Iterate top down with respect to $\color{blue}\phi$'s formula structure.
    gate$^{\gamma,[0,\infty),\emap}$ := $\fPsi{w_0}{\gamma,[0,\infty),\emap}$
    Instantiate(gate$^{\gamma,[0,\infty),\emap}$)
    if IsBool(gate$^{\gamma,[0,\infty),\emap}$) then
      PropagateTruthValue($\gamma$, $[0,\infty)$, $\emap$, ToBool(gate$^{\gamma,[0,\infty),\emap}$))
\end{lstlisting}
  \end{minipage}
\end{listing}
The monitor's main procedure for \MTLdata is shown in
Listing~\ref{lst:monitor_data} and the initialization procedure in
Listing~\ref{lst:init_data}.  Both procedures are similar to their
counterparts for MTL (see the Listings~\ref{lst:monitor}
and~\ref{lst:init}).
The main difference is that the case (T\ref{enum:observation_data})
now comprises two subcases.  The first
subcase~(T\ref{enum:observation_data}.1) handles new interpretations
for predicate symbols at a time point and is similar to the
(T\ref{enum:observation_data}) case for MTL in
Listing~\ref{lst:monitor}.  The second
subcase~(T\ref{enum:observation_data}.2) handles the freezing of
variables at a time point to data values.
Note that in the \ls{foreach} loops in both subcases, the propositions
$p(\overline{x})^{\{\tau\},\nu}$ and $\alpha^{\{\tau\},\nu}$ range
over propositions that occur in some propositional formula of the
monitor's state.  
In the following, we use (T\ref{enum:observation_data}.1)
and~(T\ref{enum:observation_data}.2) to refer to the transformation of the
corresponding subcase, respectively.

\subsubsection{State Updates.}
\label{subsubsec:stateupdates}

The central procedures for updating the monitor's state are the
procedures \ls{AddTimePoint}, \ls{RemoveInterval},
\ls{PropagateTruthValue}, and \ls{PropagateDataValue}.  Their
pseudocode is given in the
Listings~\ref{lst:tp_data}--\ref{lst:propagatedown_data}. The first
three procedures extend
their counterparts for MTL from Section~\ref{sec:prop}.  The last one
is new and propagates data values down the formula structure.
As in Section~\ref{sec:prop}, we do not fix the representation of the
propositional formulas of the monitor's state. Instead, we use the
same abstract interface for accessing and updating the state variables
\ls{gate$^{\gamma,J,\nu}$} as described in
Section~\ref{subsubsec:state}.
\begin{listing}
  \begin{minipage}[t]{0.48\linewidth}
\begin{lstlisting}[caption={Procedure for transformation~(T\ref{enum:observation_split}).},label={lst:tp_data},captionpos=b]
procedure AddTimePoint($\phi$, $J$, $\tau$)
  foreach $\textsf{gate}^{\gamma,J,\nu}$ and $K\in\{J\cap[0,\tau),\{\tau\},J\cap(\tau,\infty)\}$ do
    gate$^{\gamma,K,\nu}$ := Clone(gate$^{\gamma,J,\nu}$)
  Delete($J$) 
  if IsBool(gate$^{\phi,\{\tau\},\emap}$) then
    OutputVerdict($\{\tau\}$, ToBool(gate$^{\phi,\{\tau\},\emap}$))
  foreach gate$^{\gamma,H,\nu}$ with Contains(gate$^{\gamma,H,\nu}$, $p$), 
                            $\textsf{for}$ some proposition $p$ $\textsf{with}$ the interval $J$ do
    # Iterate top down with respect to $\color{blue}\phi$'s formula structure.
    case $\gamma = \neg\alpha$: $\hspace{.395cm}$Rename(gate$^{\gamma,H,\nu}$, $[\alpha^{J,\nu}\mapsto \alpha^{H,\nu}]$)
    case $\gamma = \alpha\lor\beta$: $\hspace{.13cm}$Rename(gate$^{\gamma,H,\nu}$, $[\alpha^{J,\nu}\mapsto \alpha^{H,\nu}, \beta^{J,\nu}\mapsto \beta^{H,\nu}]$)
    case $\gamma = \freeze{s}{y}\alpha$: $\hspace{.05cm}$Rename(gate$^{\gamma,H,\nu}$, 
                                   $\![\alpha^{J,\mu}\mapsto \alpha^{H,\mu}\mathbin{|}\mu\text{ some partial valuation}]$)
    case $\gamma = \predecessor_I\alpha$: $\hspace{.21cm}$ $\dots$ # Omitted; analogous to the next case.
    case $\gamma = \successor_I\alpha$: $\hspace{.21cm}$RefineNext($\gamma$, $H$, $\nu$, $J$, $\tau$)
                      Instantiate(gate$^{\gamma,H,\nu}$) 
    case $\gamma = \alpha\since_I\beta$: $\hspace{.05cm}$ $\dots$ # Omitted; analogous to the next case.
    case $\gamma = \alpha\until_I\beta$: $\hspace{.05cm}$RefineUntil($\gamma$, $H$, $\nu$, $J$, $\tau$)
                      Instantiate(gate$^{\gamma,H,\nu}$)
    if IsBool(gate$^{\gamma,H,\nu}$) then 
      PropagateTruthValue($\gamma$, $H$, $\nu$, ToBool(gate$^{\gamma,H,\nu}$))
\end{lstlisting}
\begin{lstlisting}[caption={Procedure for transformation~(T\ref{enum:observation_removal}).},label={lst:remove_data},captionpos=b]
procedure RemoveInterval($K$)
  foreach gate$^{\gamma,J,\nu}$ with $J\not=K$ and Contains(gate$^{\gamma,J,\nu}$, $\tpp^K$) do
    Eval(gate$^{\gamma,J,\nu}$, $[\tpp^K\mapsto\false]$)
    if IsBool(gate$^{\gamma,J,\nu}$) then 
      PropagateTruthValue($\gamma$, $J$, $\nu$, ToBool(gate$^{\gamma,J,\nu}$))
  Delete($K$) 
\end{lstlisting}
  \end{minipage}
  \hfill
  \begin{minipage}[t]{0.48\linewidth}
\begin{lstlisting}[caption={Procedure for transformation~(T\ref{enum:observation_data}.1).},label={lst:propagateup_data},captionpos=b]
procedure PropagateTruthValue($\alpha$, $J$, $\nu$, $b$)
  if IsRoot($\alpha$) and $|J|=1$ then
    OutputVerdict($J$, $b$)
  else if not IsRoot($\alpha$) then
    foreach gate$^{\gamma,K,\mu}$ with Contains(gate$^{\gamma,K,\mu}$, $\alpha^{J,\nu}$) do
      Eval(gate$^{\gamma,K,\mu}$, $[\alpha^{J,\nu}\mapsto b]$)
      if IsBool(gate$^{\gamma,K,\mu}$) then 
        PropagateTruthValue($\gamma$, $K$, $\mu$, ToBool(gate$^{\gamma,K,\mu}$))
\end{lstlisting}
\begin{lstlisting}[caption={Procedure for the transformation (T\ref{enum:observation_data}.2).},label={lst:propagatedown_data},captionpos=b]
procedure PropagateDataValue($\gamma$, $K$, $\nu$, $[x\mapsto d]$)
  if $\gamma=p(\overline{x})$ then
    $\sigma$ := PredicateInterpretations(observation, $K$)
    if $p\in\pdef(\sigma)$ and  $\overline{x}\in\pdef(\nu[x\mapsto d])$ then
      PropagateTruthValue($\gamma$, $K$, $\nu[x\mapsto d]$, $\nu[x\mapsto d](\overline{x})\in\sigma(p)$)
  else if gate$^{\gamma,K,\nu[x\mapsto d]}$ does $\textsf{not}$ exist then
    gate$^{\gamma,K,\nu[x\mapsto d]}$ := Clone(gate$^{\gamma,K,\nu}$)
    if IsBool(gate$^{\gamma,K,\nu[x\mapsto d]}$) then 
      PropagateTruthValue($\gamma$, $K$, $\nu[x\mapsto d]$, ToBool(gate$^{\gamma,K,\nu[x\mapsto d]}$))
    Rename(gate$^{\gamma,K,\nu[x\mapsto d]}$, $[\alpha^{L,\mu}\mapsto\alpha^{L,\mu[x\mapsto d]}\mathbin{|}\alpha\in\sub(\gamma),$
                               $L\text{ some interval, and }\mu\text{ some partial valuation}]$)
    foreach $\alpha^{L,\mu[x\mapsto d]}$ with Contains(gate$^{\gamma,K,\nu[x\mapsto d]}$, $\alpha^{L,\mu[x\mapsto d]}$) do
      PropagateDataValue($\alpha$, $L$, $\mu$, $[x\mapsto d]$)
  else if IsBool(gate$^{\gamma,K,\nu[x\mapsto d]}$) then 
    PropagateTruthValue($\gamma$, $K$, $\nu[x\mapsto d]$, ToBool(gate$^{\gamma,K,\nu[x\mapsto d]}$))
\end{lstlisting}
  \end{minipage}
\end{listing}

If $\gamma$ is an atomic formula of the form $p(\overline{x})$, then
\ls{PropagateDataValue} first obtains the interpretation of the
predicate symbols at the position $k$ of \ls{observation}.  It starts
the propagation of the truth value, if $p(\overline{x})$ can be
evaluated for the extended partial valuation $\nu[x\mapsto d]$.
If $\gamma$ is not an atomic formula and the propositional formula
\ls{gate}$^{\gamma,K,\nu[x\mapsto d]}$ for the extended partial
valuation $\nu[x\mapsto d]$ does not exist yet,
\ls{PropagateDataValue} creates it from \ls{gate}$^{\gamma,K,\nu}$.
When \ls{gate}$^{\gamma,K,\nu[x\mapsto d]}$ is semantically equivalent
to a Boolean constant, \ls{PropagateDataValue} starts the propagation
of the truth value. Otherwise, \ls{PropagateDataValue} continues the
propagation of the new data value down the formula structure.
Finally, if \ls{gate}$^{\gamma,K,\nu[x\mapsto d]}$ already exists and
is semantically equivalent to a Boolean constant, then---as in the
case where \ls{gate}$^{\gamma,K,\nu[x\mapsto d]}$ is newly
created---\ls{PropagateDataValue} starts the propagation of the truth
value.

\subsection{Data Structure}
\label{subsec:datastruct}

We briefly describe a graph-based data structure for representing and
updating the monitor's state variables \ls{gate$^{\gamma,J,\nu}$}.
The \emph{nodes} of the data structure are tuples of the
form~$(\gamma,J,\nu)$, with $\gamma$ a subformula of the monitored
formula $\phi$, $J$ an interval, and $\nu$ a partial valuation.  When
$\gamma$ is not atomic, the node corresponds to the state variable
\ls{gate$^{\gamma,J,\nu}$}. 
A node~$(\gamma,J,\nu)$ stores a truth value~$b\in\Three$, where the
monitor maintains the invariant
$b=\osem{\mathsf{observation},j,\nu}{\gamma}$.  If $\gamma$ is of the
form $\alpha\until_I\beta$ or $\alpha\since_I\beta$, then the node
also stores the interval~$K$ of the closest valid anchor (i.e., for
$\alpha\until_I\beta$, $k\geq j$ with
$\istp_{\mathsf{observation}}(k)= \tc_{\mathsf{observation},I}(k,j)=
\osem{\mathsf{observation},k,\nu}{\beta}=\true$ and
$\osem{\mathsf{observation},h,\nu}{\alpha}\not=\false$, for all $h$
with $j\leq h<k$ and $\istp_{\mathsf{observation}}(h)=\true$), if it
exists.
Furthermore, nodes with the same formula~$\gamma$ and partial
valuation~$\nu$ are stored in a doubly linked list, ordered by their
intervals.
The \emph{edges} of the data structure are as follows.  There is an
edge from the node~$(\alpha,K,\mu)$ to the node~$(\gamma,J,\nu)$ if a
proposition of the $(\alpha,K)$-relevant part of
$\mathsf{gate}^{\gamma,J,\nu}$ occurs in the propositional
formula~$\mathsf{gate}^{\gamma,J,\nu}$.  The edges are
bidirectional. To simplify the exposition, we use an upward directed
reading, namely, from nodes with the formula~$\alpha$ to nodes with
$\alpha$'s parent formula~$\gamma$.
For instance, both nodes~$(\alpha,J,\nu)$ and~$(\beta,J,\nu)$ have an
outgoing edge to the node~$(\alpha\lor\beta,J,\nu)$, provided that the
truth value of both nodes~$(\alpha,J,\nu)$ and~$(\beta,J,\nu)$ is
$\unknown$.

We sketch how this data structure realizes the interface specified in
Section~\ref{subsubsec:state}.
We first note that the graph-based data structure does not represent
the propositional formulas \ls{gate$^{\gamma,J,\nu}$}
explicitly. However, an explicit representation of them can be
obtained from its nodes and edges. From the incoming edges of a node
$(\gamma,J,\nu)$, we can obtain the relevant parts of
\ls{gate$^{\gamma,J,\nu}$}, in particular, the propositions occurring
in them.  Their arrangement, including the Boolean connectives between
the propositions and the relevant parts, is given through $\gamma$'s
main connective and its direct subformulas.
For example, for $\gamma=\alpha\until_I\beta$, whether the proposition
$\beta^{K,\nu}$ occurs in the $(\beta,K)$-relevant part of
\ls{gate$^{\gamma,J,\nu}$} can be determined from the node's
$(\beta,K,\nu)$ truth value and the interval of the valid anchor in
the node~$(\gamma,J,\nu)$. Note that \ls{gate$^{\gamma,J,\nu}$} does
not depend on $\beta^{K,\nu}$ when the node~$(\gamma,J,\nu)$ has a
closest valid anchor with the interval~$K'$ and~$K'<K$.
Furthermore, whether the propositions $\tpp^K$ and $\overline{\tpp}^K$
occur in the $(\beta,K)$-relevant part of \ls{gate$^{\gamma,J,\nu}$}
can be determined from the interval of the node~$(\beta,K,\nu)$.
Similarly, whether the proposition $\tcp^{K,J}_I$ occurs in the
$(\beta,K)$-relevant part of \ls{gate$^{\gamma,J,\nu}$} can be
determined from the intervals of the nodes $(\beta,K,\nu)$ and
$(\gamma,J,\nu)$.

The realization of the interface procedures is not difficult. For
instance, the procedures \ls{Add} and \ls{Remove} simply add and
remove edges.  However, some care must be taken for the procedure
\ls{Eval}.  Assume that the arguments of \ls{Eval} are
\ls{gate$^{\gamma,J,\nu}$} and the substitution
$[\alpha^{H,\nu}\mapsto\false]$, where $\gamma=\alpha\until_I\beta$
and $H>J$ with $|H|=1$.  Obviously, \ls{Eval} deletes the edge from
the node~$(\alpha,H,\nu)$ to the node~$(\gamma,J,\nu)$. This deletion
may trigger the deletion of other incoming edges to the
node~$(\gamma,J,\nu)$.  First, \ls{Eval} deletes the incoming edges
from the ``anchor'' nodes $(\beta,K,\nu)$, with $K>H$. Additionally,
\ls{Eval} deletes the interval $L$ of the node's $(\gamma,J,\nu)$
valid anchor, provided it exists and $L>H$.  Furthermore, \ls{Eval}
deletes the incoming edges from the ``continuation'' nodes
$(\alpha,K,\nu)$ that have no anchor anymore.  These ``continuation''
nodes may arise when deleting the node's valid anchor or an incoming
edge from an anchor node.  Finally, \ls{Eval} sets the
node's~$(\gamma,J,\nu)$ truth value to~$\false$, if there are no
remaining incoming edges.

\begin{example}
  \label{ex:datastruct}
  \begin{figure}[t]
    \centering
    \scalebox{.9}{
      \begin{tikzpicture}[thick, x=1pt, y=1pt]
        \tikzstyle{node} = [draw, minimum height = 14pt, minimum width = 14pt]
        \tikzstyle{trigger} = [->]
        \tikzstyle{guard} = [draw, circle, fill, minimum size = 0pt, inner sep = 0pt]

        \newcommand{\x}{20}
        \node at (\x,80) {$\freezeshort{x}\eventually_{(0,1]}p(x)$};
        \node at (\x,40) {$\phantom{\freezeshort{x}}\eventually_{(0,1]}p(x)$};
        \node at (\x,0)  {$\phantom{\freezeshort{x}\eventually_{(0,1]}}p(x)$};
        
        \renewcommand{\x}{100}
        \node at (\x,-25) {(a) initial observation $w_0$};
        
        \draw (\x-20,100) -- node[above] {$[0,\infty)$} ++ (40,0);
        
        \node[node] (Nfreeze) at (\x,80) {\scalebox{.8}{$\emap$}};
        \node[node] (Neventually) at (\x,40) {\scalebox{.8}{$\emap$}};
        \node[node] (Natomic) at (\x,0) {\scalebox{.8}{$\emap$}};
        
        \node[guard] (Geventually) at (\x,33) {};
        \node[guard] (Gfreeze) at (\x,73) {};
        
        \draw[trigger] (Natomic) to [out=90, in=-90] (Geventually);        
        \draw[trigger] (Neventually) to [out=90, in=-90] (Gfreeze);
        
        \renewcommand{\x}{180}
        \node at (\x+50,-25) {(b) observation $w_1$};
        
        \draw (\x-10,100) -- node[above] {$[0,\tau)$} ++ (40,0);
        \draw[|-|] (\x+50,100) -- node[above] {$\set{\tau}$} ++ (0.1,0);
        \draw (\x+70,100) -- node[above] {$(\tau,\infty)$} ++ (40,0);
        
        \node[node] (NfreezeL) at (\x+10, 80) {\scalebox{.8}{$\emap$}};
        \node[node] (NeventuallyL) at (\x+10, 40) {\scalebox{.8}{$\emap$}};
        \node[node] (NatomicL) at (\x+10, 0) {\scalebox{.8}{$\emap$}}; 
        
        \node[node] (NfreezeM) at (\x+50, 80) {\scalebox{.8}{$\emap$}};
        \node[node] (NeventuallyM) at (\x+50, 40) {\scalebox{.8}{$\emap$}};
        \node[node] (NatomicM) at (\x+50, 0) {\scalebox{.8}{$\emap$}};
        
        \node[node] (NfreezeR) at (\x+90, 80) {\scalebox{.8}{$\emap$}};
        \node[node] (NeventuallyR) at (\x+90, 40) {\scalebox{.8}{$\emap$}};
        \node[node, fill=white] (NatomicR) at (\x+90, 0) {\scalebox{.8}{$\emap$}};
        \node[guard] (GatomicR) at (\x+96, 7) {};
        
        \node[guard] (GeventuallyL1) at (\x+4, 33) {};
        \node[guard] (GeventuallyL2) at (\x+10, 33) {};
        \node[guard] (GeventuallyL3) at (\x+16, 33) {};
        \node[guard] (GfreezeL) at (\x+10, 73) {};
        
        \node[guard] (GeventuallyM) at (\x+61, 42) {};
        \node[guard] (GfreezeM) at (\x+50, 73) {};
        
        \node[guard] (GeventuallyR) at (\x+90, 33) {};
        \node[guard] (GfreezeR) at (\x+90, 73) {};
        
        \draw[trigger] (NatomicL) to [out=90, in=-90] (GeventuallyL1); 
        \draw[trigger] (NatomicM) .. controls (\x+50,20)  and (\x+10,0) .. (GeventuallyL2);
        \draw[trigger] (GatomicR) to [out=90, in=-90] (NeventuallyR.south);
        \draw[trigger] (NatomicR) to [out=90, in=-90] (NeventuallyM.south);
        \draw[trigger] (NatomicR) .. controls (\x+70,30) and (\x+16,5) .. (GeventuallyL3);
        
        \draw[trigger] (NeventuallyL) to [out=90, in=-90] (GfreezeL);        
        \draw[trigger] (NeventuallyM) to [out=90, in=-90] (GfreezeM);        
        \draw[trigger] (NeventuallyR) to [out=90, in=-90] (GfreezeR);        

        \draw[dashed] (NfreezeL) to [out=0, in=180] (NfreezeM);
        \draw[dashed] (NfreezeM) to [out=0, in=180] (NfreezeR);

        \draw[dashed] (NeventuallyL) to [out=0, in=180] (NeventuallyM);
        \draw[dashed] (NeventuallyM) to [out=0, in=180] (NeventuallyR);

        \draw[dashed] (NatomicL) to [out=0, in=180] (NatomicM);
        \draw[dashed] (NatomicM) to [out=0, in=180] (NatomicR);

        \renewcommand{\x}{340}
        \node at (\x+50,-25) {(c) observation $w$};
        
        \draw (\x-10,100) -- node[above] {$[0,\tau)$} ++ (40,0);
        \draw[|-|] (\x+50,100) -- node[above] {$\set{\tau}$} ++ (0.1,0);
        \draw (\x+70,100) -- node[above] {$(\tau,\infty)$} ++ (40,0);
        
        \node[node] (NfreezeL) at (\x+10, 80) {\scalebox{.8}{$\emap$}};
        \node[node] (NeventuallyL) at (\x+10, 40) {\scalebox{.8}{$\emap$}};
        \node[node] (NatomicL) at (\x+10, 0) {\scalebox{.8}{$\emap$}}; 
        
        \node[node] (NfreezeM) at (\x+50, 80) {\scalebox{.8}{$\emap$}};
        \node[node] (NeventuallyM) at (\x+60, 49) {\scalebox{.8}{$[x\!\mapsto\!d]$}}; 
        \node[node] (NatomicM) at (\x+50, 0) {\scalebox{.8}{$\emap$}};
        
        \node[node] (NfreezeR) at (\x+90, 80) {\scalebox{.8}{$\emap$}};
        \node[node] (NeventuallyR) at (\x+90, 40) {\scalebox{.8}{$\emap$}};
        \node[node] (NatomicR2) at (\x+100, 8) {\scalebox{.8}{$[x\!\mapsto\!d]$}}; 
        \node[node, fill=white] (NatomicR) at (\x+90, 0) {\scalebox{.8}{$\emap$}};

        \node[guard] (GeventuallyL1) at (\x+4, 33) {};
        \node[guard] (GeventuallyL2) at (\x+10, 33) {};
        \node[guard] (GeventuallyL3) at (\x+16, 33) {};
        \node[guard] (GfreezeL) at (\x+10, 73) {};
        
        \node[guard] (GeventuallyM) at (\x+61, 42) {};
        \node[guard] (GfreezeM) at (\x+50, 73) {};
        
        \node[guard] (GeventuallyR) at (\x+90, 33) {};
        \node[guard] (GfreezeR) at (\x+90, 73) {};
        
        \draw[trigger] (NatomicL) to [out=90, in=-90] (GeventuallyL1); 
        \draw[trigger] (NatomicM) .. controls (\x+50,20)  and (\x+10,0) .. (GeventuallyL2);
        \draw[trigger] (NatomicR) to [out=90, in=-90] (GeventuallyR);        
        \draw[trigger] (NatomicR2) to [out=90, in=-90] (GeventuallyM);
        \draw[trigger] (NatomicR) .. controls (\x+70,30) and (\x+16,5) .. (GeventuallyL3);
        
        \draw[trigger] (NeventuallyL) to [out=90, in=-90] (GfreezeL);        
        \draw[trigger] (NeventuallyM) to [out=90, in=-90] (GfreezeM);        
        \draw[trigger] (NeventuallyR) to [out=90, in=-90] (GfreezeR);        

        \draw[dashed] (NfreezeL) to [out=0, in=180] (NfreezeM);
        \draw[dashed] (NfreezeM) to [out=0, in=180] (NfreezeR);

        \draw[dashed] (NeventuallyL) to [out=0, in=180] (NeventuallyR);

        \draw[dashed] (NatomicL) to [out=0, in=180] (NatomicM);
        \draw[dashed] (NatomicM) to [out=0, in=180] (NatomicR);
      \end{tikzpicture}}
    \caption{Graph-based data structure (Example~\ref{ex:datastruct}).}
    \label{fig:datastruct}
  \end{figure}
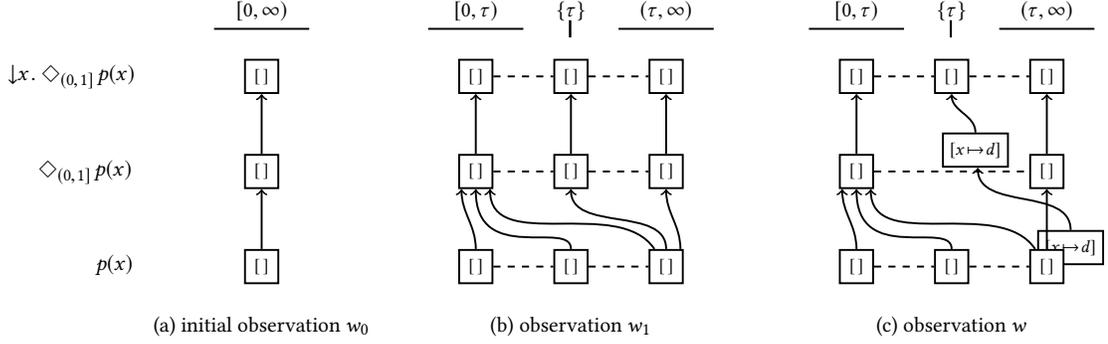
  We illustrate the data structure and its updates.
  Figure~\ref{fig:datastruct} shows the data structures associated with the
  formula~$\freeze{r}{x}\eventually_{(0,1]}p(x)$ and the observations
  (a)~$w_0\!=\!\big([0,\infty),(\emap,\emap)\big)$,
  (b)~$w_1=\big([0,\tau),(\emap,\emap)\big)\big(\{\tau\},(\emap,\emap)\big)
  \big((\tau,\infty),(\emap,\emap)\big)$,
  and
  (c)~$w=\big([0,\tau),(\emap,\emap)\big)\big(\{\tau\},(\emap,[r\mapsto
  d])\big) \big((\tau,\infty),(\emap,\emap)\big)$.
  A box in Figure~\ref{fig:datastruct} corresponds to a node of the
  graph-based data structure, where the node's formula is given by the
  row of the box, the interval by the column of the box, and the
  partial valuation is given inside the box.
  The edges are depicted as solid lines between boxes.  The dashed
  lines are the links of the ordered doubly linked lists. Note that
  the three boxes in Figure~\ref{fig:datastruct}(a) and the two boxes
  in Figure~\ref{fig:datastruct}(c) with the partial valuation~$[x\mapsto d]$
  are all stored in singleton lists.
  
  Note that $w_1$ is obtained from $w_0$ by a (T1)~transformation that
  splits the interval $[0,\infty)$ at $\tau$, and $w$ is obtained from
  $w_1$ by a (T3.2)~transformation that freezes the data value $d$ to
  the variable $x$ at $\tau$.  Observe that
  Figure~\ref{fig:datastruct}(c) does not contain the node
  $(\eventually_{(0,1]}p(x), \{\tau\}, \emap)$. This node is
  irrelevant, since it has no outgoing edges. Irrelevant nodes are
  removed from the data structure.  Furthermore, note that the data
  structure shown in Figure~\ref{fig:datastruct}(c) represents the
  propositional formulas $\theta_w(\fPsi{w}{\gamma,J,\nu})$ from
  Example~\ref{ex:freeze}.  The nonexistence of the node
  $(\eventually_{(0,1]}p(x), \{\tau\}, \emap)$ corresponds to the fact
  that the proposition $\eventually_{(0,1]}p(x)^{\{\tau\},\emap}$ does
  not occur in any of the propositional formulas.
  \exampleendmark
\end{example}

We remark that the data structure allows us to easily determine
the propositional formulas~$\mathsf{gate}^{\gamma,J,\nu}$ in which a
given proposition $\alpha^{K,\mu}$ occurs. We just need to follow the
node's $(\alpha,K,\mu)$ outgoing edges, provided that the node's truth
value is $\unknown$. Analogously, by following a node's
$(\gamma,J,\nu)$ incoming edges we can determine the propositions that
occur in~$\mathsf{gate}^{\gamma,J,\nu}$.  Hence, the \ls{foreach}
loops in the procedures \ls{RemoveInterval} and
\ls{PropagateTruthValue}, and the second one in \ls{AddTimePoint} can
be implemented efficiently.
The data structure can also be further optimized. For example, to
reduce the number of edges, a node only stores at most one outgoing
edge.  The other outgoing edges are implicit and computed on demand by
following the links of the doubly linked lists to the neighboring
nodes.  In particular, the procedure \ls{AddTimePoint} needs to update
significantly fewer outgoing edges when splitting an interval. We omit
such implementation details.

\subsection{Correctness}
\label{subsec:correct}

This section is dedicated to the monitor's correctness and we
prove the following theorem.
\begin{theorem}\label{thm:correct}
  \ls{$\mathsf{MonitorMTL}^{\downarrow}$} is observationally complete
  and sound.
\end{theorem}
\begin{proof}
  We first observe that the monitor only outputs verdicts with the
  procedure \ls{AddTimePoint($\phi$, $J$, $\tau$)} and the procedure
  \ls{PropagateTruthValue($\alpha$, $J$, $\nu$, $b$)} when
  $\alpha=\phi$ and $J$ is a singleton.  In both cases, 
  \ls{IsBool(gate$^{\phi,J,\emap}$)} returns true.  For the second
  case, observe that \ls{PropagateTruthValue} is only called when
  \ls{IsBool(gate$^{\phi,J,\emap}$)} returns true.
  Moreover, $\nu$ is $\emap$ in these calls, since state variables
  $\mathsf{gate}^{\alpha,J,\nu}$ with new partial valuations $\nu\neq\emap$
  are only created by the procedure 
  \ls{PropagateDataValue($\gamma$, $K$, $\nu$, $[x\mapsto d]$)}, 
  which is never called with the argument $\gamma=\phi$.
  Thus whenever the monitor outputs a verdict
  $(J,b)$, then $J=\{\tau\}$ and
  $\mathsf{gate}^{\phi,J,\emap}\equiv b$, for some $\tau\in\Qpos$ and
  $b\in\Two$.

  Let $\bar{w}$ be a valid observation sequence that represents the
  monitor's input. Without loss of generality, we assume that a single
  transformation is applied in each iteration, that is, for each
  $i\in\Nat$, $w_{i+1}$ is obtained from $w_i$ by exactly one of the
  transformations~(T1), (T2), (T3.1), or~(T3.2).
  For an observation $w$ of $\bar{w}$, we denote by
  $\mathsf{gate}_{w}^{\gamma,J,\nu}$ the value of the state variable
  $\mathsf{gate}^{\gamma,J,\nu}$ at the end of the iteration that
  processes the observation~$w$, that is, $w$ is the value of the
  monitor's state variable \ls{observation}.

  The equivalence below follows from
  Lemma~\ref{lem:key_equivalence}, which is stated and proved later.
  For an observation~$w$ of~$\bar{w}$, a time point in $w$ with
  timestamp~$\tau$, and $b\in\Two$, it holds that
  \begin{equation}
    \label{eq:top}
    \theta_w(\fPsi{w}{\phi,\{\tau\},\emap}) \equiv b 
    \quad \text{iff} \quad
    \mathsf{gate}_w^{\phi,\{\tau\},\emap} \equiv b
    \,.
  \end{equation}
  Furthermore, we note that in the iteration $w$, the monitor's state
  contains the state variable $\mathsf{gate}^{\phi,J,\emap}$ for any
  interval~$J$ that occurs in a letter of $w$.

  Observational soundness follows from the above observation on when
  the monitor output verdicts, the equivalence~(\ref{eq:top}), and
  Lemma~\ref{lem:propformulaMTLdata}.  
  To show observational completeness, suppose that
  $\eosem{w,\tau,\emap}{\phi}=b\in\Two$.  We must show that the
  verdict~$(\{\tau\},b)$ is output in this iteration~$w$ or has
  already been output in a previous iteration of $\bar{w}$.  From
  $\eosem{w,\tau,\emap}{\phi}\in\Two$, it follows that there is a time
  point~$j\in\pos(w)$ with the timestamp~$\tau$.  Furthermore, 
  $\osem{w,j,\emap}{\phi}=b$.  It follows from
  Lemma~\ref{lem:propformulaMTLdata} that
  $\theta_w(\fPsi{w}{\phi,\{\tau\},\emap})\equiv b$, and
  by~(\ref{eq:top}), we obtain that
  $\mathsf{gate}_w^{\phi,\{\tau\},\emap}\equiv b$.  We are done when
  the procedure \ls{PropagateTruthValue} outputs the
  verdict~$(\{\tau\},b)$. Otherwise, let $w'\sqsubseteq w$ be the
  first observation in $\bar{w}$ for which
  \ls{IsBool($\mathsf{gate}_{w'}^{\phi,J,\emap}$)} returns true, for
  some interval~$J$ with $\tau\in J$.
  Furthermore, let $w''$ be the observation of $\bar{w}$ when
  $J'\subseteq J$ is split into $J'\cap[0,\tau)$, $\{\tau\}$, and
  $J'\cap(\tau,\infty)$.  Clearly, $w'\sqsubseteq w''\sqsubsetneq
  w$. In this iteration, $\mathsf{gate}_{w''}^{\phi,\{\tau\},\emap}$
  is set to $\mathsf{gate}_{w'}^{\phi,J',\emap}$ by the call to
  \ls{Clone} in the \ls{AddTimePoint} procedure.  Note that
  $\mathsf{gate}_{w''}^{\phi,\{\tau\},\emap}\equiv
  \mathsf{gate}_{w'}^{\phi,J,\emap}\equiv b$.  After the creation of
  $\mathsf{gate}_{w''}^{\phi,\{\tau\},\emap}$, the monitor outputs the
  verdict $(\{\tau\},b)$ by calling the procedure \ls{OutputVerdict}.
\end{proof}

In the remainder of this section, we establish the monitor's key
invariants (Lemma~\ref{lem:key_existence} and
Lemma~\ref{lem:key_equivalence}). The equivalence~(\ref{eq:top}),
used to prove Theorem~\ref{thm:correct}, is a straightforward
consequence of Lemma~\ref{lem:key_equivalence}, and
Lemma~\ref{lem:key_existence} is used to establish
Lemma~\ref{lem:key_equivalence}.
To state the invariants, we introduce further notation.  As in the
proof of Theorem~\ref{thm:correct}, let $\bar{w}$ be a valid
observation sequence that represents the monitor's input.  Again, we
assume without loss of generality that $w_{i+1}$ is obtained from
$w_i$ by exactly one of the transformations~(T1), (T2), (T3.1),
or~(T3.2), for each $i\in\Nat$.  Furthermore,
$\mathsf{gate}_{w}^{\gamma,J,\nu}$ denotes the value of the state
variable $\mathsf{gate}^{\gamma,J,\nu}$ at the end of the iteration
that processes the observation~$w$ of $\bar{w}$.
To simplify matters, we also assume that state variables are not
garbage collected even when they are irrelevant. This
assumption does not affect the monitor's correctness because for an
irrelevant state variable $\mathsf{gate}_{w}^{\psi,J,\nu}$, the
corresponding proposition~$\psi^{J,\nu}$ does not occur in any
relevant \ls{gate} state variable. The monitor only does more work
than necessary.

The following definition allows us to state which state variables the
monitor maintains.
For an observation~$w$, we define inductively the set $\Val_w(\psi,J)$
of the \emph{relevant valuations} for $\psi\in\sub(\phi)$ at
interval~$J$, where $J$ ranges over the intervals that occur in the
letters of~$w$, as
\begin{equation*}
  \Val_w(\phi,J) := \{\emap\}
  \qquad\text{and}\qquad                   
  \Val_w(\psi,J) :=
  \big\{ 
    \nu 
    \mathbin{\big\vert}
    \text{$\theta_w(\fPsi{w}{\gamma,K,\mu})$ depends on $\psi^{J,\nu}$},
    \text{ for some $K$ and $\mu\in\Val_w(\gamma,K)$} 
  \big\} 
  \,,
\end{equation*}
for $\psi\not=\phi$, with the parent formula $\gamma$.
Recall that a propositional formula $\Psi$ \emph{depends on} the
proposition~$p$ if
$[p\mapsto\true](\Psi)\not\equiv[p\mapsto\false](\Psi)$.

\begin{example}
  We revisit Example~\ref{ex:freeze} with the formula
  $\freeze{r}{x}\alpha$ and the observation
  $w=\big(J_0,(\emap,\emap)\big)\big(J_1,(\emap,[r\mapsto
  d])\big)\big(J_2,(\emap,\emap)\big)$.
  Recall that $\alpha=\eventually_{(0,1]}\beta$, with $\beta=p(x)$,
  $J_0=[0,\tau)$, $J_1=\{\tau\}$, and $J_2=(\tau,\infty)$.
  We have the following relevant valuations.
  \begin{equation*}
    \begin{array}{r@{\ }l@{\qquad}r@{\ }l@{\qquad}r@{\ }l}
      \Val_w(\freeze{r}{x}\alpha,J_0) & = \{\emap\} &
      \Val_w(\freeze{r}{x}\alpha,J_1) & = \{\emap\} &
      \Val_w(\freeze{r}{x}\alpha,J_2) & = \{\emap\}
      \\
      \Val_w(\alpha,J_0) & = \{\emap\} &
      \Val_w(\alpha,J_1) & = \{[x\mapsto d]\} &
      \Val_w(\alpha,J_2) & = \{\emap\}
      \\
      \Val_w(\beta,J_0) & = \{\emap\} &
      \Val_w(\beta,J_1) & = \{\emap\} &
      \Val_w(\beta,J_2) & = \{\emap,[x\mapsto d]\}
    \end{array}
  \end{equation*}
  For instance, $[x\mapsto d]\in\Val_w(\beta,J_2)$ because
  $\theta_w(\fPsi{w}{\alpha,J_1,[x\mapsto d]})$ depends on
  $\beta^{J_2,[x\mapsto d]}$ and $[x\mapsto d]\in
  \Val_w(\alpha,J_1)$. The latter membership in turn holds because
  $\theta_w(\fPsi{w}{\freeze{r}{x}\alpha,J_1,\emap})$ depends on
  $\alpha^{J_1,[x\mapsto d]}$ and
  $\emap\in \Val_w(\freeze{r}{x}\alpha,J_1)$, by the definition of the
  base case.
  We also point out the correspondence between the nodes in the
  graph-based data structure and the relevant valuations.  Compare,
  for instance, Figure~\ref{fig:datastruct}(c) and the relevant
  valuations from this example.~\exampleendmark
\end{example}

Finally, we make the simplifying assumption that only a single
variable is frozen to a data value by (T3.2) transformations.  That
is, we assume that a register occurs at most once in the formula
$\phi$. Note that for a register~$r$ that occurs twice in $\phi$, we
can replace one occurrence with a fresh register~$r'$ and assume that
$r'$ carries the same data value at a time point as $r$.  Furthermore,
we can split a (T3.2) transformation into multiple ones such that the
register assignment of any of these transformations only maps a single
register to a data value.  Under this assumption, the following
technical lemma holds, which states that when this transformation is
used, only subformulas of the freeze subformula containing the
involved register can have new relevant valuations.

\begin{lemma}
  \label{lem:valT3_register}
  Let $w$ and $w'$ be observations such that $w'$ is obtained from $w$
  by the transformation~(T3.2), with $\tau$ and $\rho$ the
  corresponding timestamp and register assignment, respectively.
  For any $\psi\in\sub(\phi)$, interval $J$ in $w$, and partial
  valuation~$\nu$, it holds that if
  $\nu\in\Val_{w'}(\psi,J)\setminus\Val_w(\psi,J)$, then $\psi$ is a
  proper subformula of some $\freeze{r}{x}{\alpha}\in\sub(\phi)$ with
  $r\in\pdef(\rho)$ and $\nu(x)=\rho(r)$.
  Additionally, the following conditions hold for any partial
  valuation~$\mu$, if also $\mu\in\Val_{w'}(\gamma,K)$ and
  $\theta_{w'}(\fPsi{w'}{\gamma,K,\mu})$ depends on $\psi^{J,\nu}$,
  where $\gamma$ is $\psi$'s parent formula and $K$ an interval in
  $w$.
  \begin{enumerate}
  \item If $\mu\notin\Val_{w}(\gamma,K)$ then $\mu(x)=\rho(r)$ and 
    $\theta_{w}(\fPsi{w}{\gamma,K,\mu[x\mapsto\unknown]})$ depends
    on $\psi^{J,\nu[x\mapsto\unknown]}$.
  \item If $\mu\in\Val_{w}(\gamma,K)$ then
    $\gamma=\freeze{r}{x}{\alpha}$, 
    $\set{\tau}=J=K$, and
    $\nu=\mu[x\mapsto\rho(r)]$.
  \end{enumerate}
\end{lemma}
\begin{proof}
  We prove the lemma's first part by contraposition. Namely, we show
  that if $\psi$ is not a proper subformula of some
  $\freeze{r}{x}{\alpha}\in\sub(\phi)$ with $r\in\pdef(\rho)$ and
  $\nu(x)=\rho(r)$, then
  $\Val_{w'}(\psi,J)\subseteq \Val_{w}(\psi,J)$.
  If $\psi=\phi$ then, by definition,
  $\Val_{w'}(\psi,J)=\Val_{w}(\psi,J)=\{\emap\}$.  
  Let $\gamma$ be $\psi$'s parent formula.  By assumption, $\gamma$ is
  not a subformula of some $\freeze{r}{x}{\alpha}$ with
  $r\in\pdef(\rho)$ and $\nu(x)=\rho(r)$.  We have that if
  $\theta_{w'}(\fPsi{w'}{\gamma,K',\mu'})$ depends on $\psi^{J',\nu'}$
  then $\theta_{w}(\fPsi{w}{\gamma,K',\mu'})$ also depends on
  $\psi^{J',\nu'}$, for any intervals $K'$ and $J'$ of $w$ and partial
  valuations $\mu'$ and $\nu'$.  Note that
  $\theta_w\sqsubseteq\theta_{w'}$ and
  $\fPsi{w'}{\gamma,K',\mu'}=\fPsi{w}{\gamma,K',\mu'}$.  It follows
  that $\Val_{w'}(\psi,J)\subseteq \Val_{w}(\psi,J)$.

  We make a case split to prove the lemma's second part.

  \noindent
  \emph{Case I:} $\psi=\alpha$. That is,
  $\gamma=\freeze{r}{x}{\psi}$. We first show that
  $\mu\in\Val_{w}(\gamma,K)$. If, for the sake of a contradiction,
  $\mu\notin\Val_{w}(\gamma,K)$, then it follows from the lemma's
  first part for $\gamma$, $K$, and $\mu$ that $\gamma$ is a proper
  subformula of some $\freeze{r'}{x'}{\alpha'}$, with
  $r'\in\pdef(\rho)$ and $\mu(x')=\rho(r')$.  This contradicts the
  assumption that only one variable is frozen to a data value by the
  transformation.
  Hence, $\mu\in\Val_{w}(\gamma,K)$, and (1)~trivially holds. We
  prove~(2). Note that, since $\gamma=\freeze{r}{x}{\alpha}$, we have
  that $\fPsi{w}{\gamma,K,\mu}=\psi^{K,\nu'}$ and
  $\fPsi{w'}{\gamma,K,\mu}=\psi^{K,\nu''}$, for some partial
  valuations $\nu'$ and $\nu''$.  From $\nu\in\Val_{w'}(\psi,J)$, it
  follows that
  $\theta_{w'}(\fPsi{w'}{\gamma,K,\mu})\equiv \psi^{J,\nu}$ and thus
  $J=K$.  From $\nu\in\Val_{w'}(\psi,J)\setminus\Val_w(\psi,J)$, it
  follows that $\{\tau\}=J=K$. From the definition of
  $\fPsi{w'}{\gamma,\mu,K}$, it follows that
  $\nu=\mu[x\mapsto\varrho(r)]$.
    
  \noindent
  \emph{Case II:} $\psi$ is a proper subformula of $\alpha$.  As
  $\theta_{w'}(\fPsi{w'}{\gamma,K,\mu})$ depends on $\psi^{J,\nu}$, we
  obtain that $\theta_{w}(\fPsi{w}{\gamma,K,\mu})$ depends on
  $\psi^{J,\nu}$.  If $\mu\in\Val_{w}(\gamma,K)$, then
  $\nu\in\Val_{w}(\psi,J)$, which contradicts the assumption
  $\nu\in\Val_{w'}(\psi,J)\setminus\Val_{w}(\psi,J)$.  
  Hence, $\mu\notin\Val_{w}(\gamma,K)$, and (2)~trivially holds.  We
  prove~(1).
  From the lemma's first part applied to $\gamma$, $K$, and $\mu$, we
  obtain that $\mu(x)=\rho(r)$, and therefore $\mu(x)=\nu(x)$.
  Furthermore, as $\psi^{J,\nu}\notin\pdef(\theta_{w'})$, we have that
  $\psi^{J,\nu[x\mapsto\bot]}\notin\pdef(\theta_{w})$, and thus
  $\theta_{w}(\fPsi{w}{\gamma,K,\mu[x\mapsto\bot]})$ depends on
  $\psi^{J,\nu[x\mapsto\bot]}$.
\end{proof}

The next lemma establishes the key invariant about the existence of
the monitor's \ls{gate} state variables.
\begin{lemma}\label{lem:key_existence}
  Let $w$ be an observation of $\bar{w}$, $J$ an interval of~$w$,
  $\psi\in\sub(\phi)$ a nonatomic formula, and $\nu\in\Val_w(\psi,J)$
  a partial valuation.
  The monitor's state at the iteration that processes $w$ contains the
  state variable $\mathsf{gate}^{\psi,J,\nu}$.
\end{lemma}
\begin{proof}
  We reason by induction on the position of $w$ in the
  sequence~$\bar{w}$.  Recall that we assume, without loss of
  generality, that a single transformation is applied to an
  observation in $\bar{w}$.
  In the base case, the observation~$w$ is $w_0$.  The interval
  $[0,\infty)$ is the only interval of a letter in $w_0$ and
  $\Val_{w_0}(\psi,[0,\infty))=\{\emap\}$, for any
  $\psi\in\sub(\phi)$.  Since the monitor has not received any
  messages, only the procedure \ls{Init} has been executed so
  far. \ls{Init} creates in its \ls{foreach} loop for every
  $\psi\in\sub(\phi)$ the state variable
  \ls{gate$^{\psi,[0,\infty),\emap}$}.
  This concludes the base case.

  For the step case, we assume that the statement holds for~$w$ and
  prove it for $w'$, the observation after~$w$ in~$\bar{w}$.  Let
  $\psi\in\sub(\phi)$, $J$ an interval of $w'$, and
  $\nu\in\Val_{w'}(\psi,J)$.  We must prove the existence of the
  state variable~$\mathsf{gate}^{\psi,J,\nu}_{w'}$.  We make a case
  distinction on the type of the transformation~$t$ that transforms
  $w$ into $w'$.
  The cases~(T1),~(T2), and~(T3.1) are similar and straightforward. We
  only sketch the (T1) case. Let $J'$ be the interval that is returned
  by \ls{DeltaT1($t$)}, that is, the interval that is split.  If
  $J\not\subseteq J'$, then it follows that $\nu\in\Val_w(\psi,J)$.
  By the induction hypothesis, we have that
  $\mathsf{gate}_{w}^{\psi,J,\nu}$ exists. Since this state variable
  is not deleted, we have that $\mathsf{gate}_{w'}^{\psi,J,\nu}$
  exists.  If $J\subseteq J'$, that is, $J$ originates from the
  interval $J'$, then we have that $\nu\in\Val_w(\psi,J')$ and obtain
  by the induction hypothesis that $\mathsf{gate}_{w}^{\psi,J',\nu}$
  exists.  The procedure \ls{AddTimePoint} creates in its first
  \ls{foreach} loop the state variable $\mathsf{gate}^{\psi,J,\nu}$ by
  cloning $\mathsf{gate}^{\psi,J',\nu}$.

  It remains to prove the (T3.2) case.
  Let $\tau$ be the timestamp and $\rho$ the partial register
  assignment returned by \ls{DeltaT32($t$)}.
  If $\mathsf{gate}^{\psi,J,\nu}_w$ exists, then the existence of
  $\mathsf{gate}^{\psi,J,\nu}_{w'}$ directly follows from the
  observation that no state variable is deleted in the (T3.2) case.
  For the remainder of the proof, suppose that $\mathsf{gate}^{\psi,J,\nu}_w$ does
  not exist, where $\psi$ is a proper subformula of~$\phi$ with the
  parent formula $\gamma$.
  Note that if $\psi=\phi$ then $\nu=\emap$, since $\phi$ is closed.
  It is easy see that $\mathsf{gate}^{\phi,J,\emap}_w$ exists and
  hence also $\mathsf{gate}^{\phi,J,\emap}_{w'}$.
  From the induction hypothesis, it follows that $\nu\notin\Val_w(\psi,J)$.
  From $\nu\in\Val_{w'}(\psi,J)$, it follows that 
  $\theta_{w'}(\fPsi{w'}{\gamma,K,\mu})$ depends on $\psi^{J,\nu}$,
  for some interval~$K$ in~$w'$ and $\mu\in\Val_{w'}(\gamma,K)$.
  From Lemma~\ref{lem:valT3_register}, we obtain that $\gamma$ is a
  subformula of some $\freeze{r}{x}{\alpha}\in\sub(\phi)$,
  $r\in\pdef(\rho)$, and $\nu(x)=\rho(r)$.
  We prove the existence of $\mathsf{gate}^{\psi,J,\nu}_{w'}$ by
  induction on the distance between $\alpha$ and $\psi$, that is, the
  formula length of~$\psi$ minus the formula length of~$\alpha$.

  For the base case, we have that $\psi=\alpha$ and $\gamma=\freeze{r}{x}{\alpha}$.
  For the sake of contradiction, suppose that
  $\mu\notin\Val_{w}(\gamma,K)$. 
  From
  Lemma~\ref{lem:valT3_register}(1), it follows that $\mu(x)=\rho(r)$.
  However, from the definitions of $\Val_{w'}(\gamma,K)$ and
  $\fPsi{w'}{\gamma,K,\mu}$, we have that $x\notin\pdef(\mu)$, which
  contradicts $\mu(x)=\rho(r)$.
  Hence $\mu\in\Val_{w}(\gamma,K)$. From the outer induction
  hypothesis, it follows that $\mathsf{gate}^{\gamma,K,\mu}_w$ exists.
  By Lemma~\ref{lem:valT3_register}(2), we have that
  $\nu=\mu[x\mapsto\rho(r)]$ and $J=K=\set{\tau}$.
  Therefore, \ls{PropagateDataValue($\psi$, $J$, $\mu$, $[x\mapsto\rho(r)]$)}
  is called from
  \ls{MonitorMTL$^{\downarrow}$}.  The first \ls{else if} branch of
  \ls{PropagateDataValue} is executed, which creates the state
  variable $\mathsf{gate}^{\psi,J,\nu}$.

  For the step case, we have that $\psi$ is a proper subformula of $\alpha$. 
  By the inner induction hypothesis, $\mathsf{gate}^{\gamma,K,\mu}_{w'}$
  exists. 

  \noindent
  \emph{Case I:} $\mathsf{gate}^{\gamma,K,\mu}_w$ does not exist.
  Therefore,
  $\mathsf{gate}^{\gamma,K,\mu}$ is created at~$w'$ within
  \ls{PropagateDataValue($\gamma$, $K$, $\mu'$, $[x\mapsto\rho(r)]$)},
  for some partial valuation~$\mu'$. Note that $\mu=\mu'[x\mapsto \rho(r)]$
  and $x\notin\pdef(\mu')$.
  It also follows from the outer induction hypothesis that
  $\mu\notin\Val_{w}(\gamma,K)$.  From
  Lemma~\ref{lem:valT3_register}(1), it follows that
  $\theta_{w}(\fPsi{w}{\gamma,K,\mu'})$ depends on $\psi^{J,\nu'}$,
  where $\nu'=\nu[x \mapsto\unknown]$.  This means that
  \ls{Contains($\mathsf{gate}_w^{\gamma,K,\mu'}$, $\psi^{J,\nu'}$)}
  returns true.
  As $\mathsf{gate}_{w'}^{\gamma,K,\mu}$ is obtained from
  $\mathsf{gate}_w^{\gamma,K,\mu'}$ by cloning and renaming its
  propositions, we obtain that also
  \ls{Contains($\mathsf{gate}_{w'}^{\gamma,K,\mu}$, $\psi^{J,\nu}$)}
  returns true. Therefore, \ls{PropagateDataValue} is called with the
  parameters $\psi$, $J$, $\nu'$, and $[x\mapsto \rho(r)]$. The
  state variable $\mathsf{gate}_{w'}^{J,\psi,\nu}$ is created
  within this call.
  
  \noindent
  \emph{Case II:} $\mathsf{gate}^{\gamma,K,\mu}_w$ exists. 
  It must be
  the case that $\theta_w(\Psi^{\gamma,K,\mu}_w)$ depends on
  $\psi^{J,\nu}$, since $\theta_{w'}(\Psi^{\gamma,K,\mu}_{w'})$
  depends on $\psi^{J,\nu}$.  It follows that $\nu\in\Val_w(\psi,J)$,
  which is a contradiction, and hence this second case cannot occur.
\end{proof}

The final lemma establishes the key invariant about the semantic
equivalence of the monitor's \ls{gate} state variables for which we
have shown the existence in Lemma~\ref{lem:key_existence}.
\begin{lemma}\label{lem:key_equivalence}
  Let $w$ be an observation of $\bar{w}$, $J$ an interval of~$w$,
  $\psi\in\sub(\phi)$ a nonatomic formula, and $\nu\in\Val_w(\psi,J)$
  a partial valuation.  It holds that
  $\mathsf{gate}_{w}^{\psi,J,\nu} \equiv
  \theta_w(\fPsi{w}{\psi,J,\nu})$.
\end{lemma}
\begin{proof}
  As in Lemma~\ref{lem:key_existence}, we reason by induction on the
  position of $w$ in the sequence~$\bar{w}$.
  In the base case, the observation~$w$ is $w_0$.  We have that
  $J=[0,\infty)$ and $\nu=\emap$.
  Only the procedure \ls{Init} is executed, which initializes
  \ls{gate$^{\psi,[0,\infty),\emap}$} with
  $\fPsi{w_0}{\psi,[0,\infty),\emap}$.  The execution of the
  procedures \ls{Instantiate} and \ls{PropagateTruthValue}, which are
  called by \ls{Init}, results in applying the
  substitution~$\theta_{w_0}$ to \ls{gate$^{\psi,[0,\infty),\emap}$}.
  This concludes the base case.

  For the step case, we assume that the statement holds for~$w$ and
  prove it for $w'$, the observation after~$w$ in~$\bar{w}$.  Let
  $\psi\in\sub(\phi)$, $J$ an interval of $w'$, and
  $\nu\in\Val_{w'}(\psi,J)$.  We make a case distinction on the type
  of the transformation~$t$ that transforms $w$ into $w'$.
  We start with the~(T3.2) case.

  \emph{Transformation~(T3.2).}
  We first note that a state variable is modified only by \ls{Rename}
  (from \ls{PropagateDataValue}) and by \ls{Eval} (from
  \ls{PropagateTruthValue}).
  Furthermore, a state variable is modified at most once by
  \ls{Rename}. Indeed, the first modification happens just after
  creating the state variable, using \ls{Clone}. A second modification
  cannot happen, because the \ls{else if} branch in which the second
  call would hypothetically occur is executed only when the state
  variable does not exist already.
  Also, a call to \ls{Rename} cannot be preceded by a call to
  \ls{PropagateTruthValue} (for the same \ls{gate} state variable).
  We conclude that the possible modification by \ls{Rename} precedes
  the modifications by \ls{Eval} in the sequence of modifications of a
  state variable during the processing of the current transformation.
  We denote by $\mathsf{gate}_{w\to w'}^{\psi,J,\nu}$ the value of the
  $\mathsf{gate}^{\psi,J,\nu}$ after the possible modification by
  \ls{Rename}, and before the modifications by~\ls{Eval}.
  Note also that if $\mathsf{gate}_{w}^{\psi,J,\nu}$ exists, then
  $\mathsf{gate}_{w\to w'}^{\psi,J,\nu} =
  \mathsf{gate}_{w}^{\psi,J,\nu}$.

  We have that
  $\mathsf{gate}_{w\to w'}^{\psi,J,\nu} \equiv
  \theta_w(\fPsi{w'}{\psi,J,\nu})$.
  Note that the right-hand side of the semantic equivalence uses the
  substitution for~$w$ and the propositional formula for~$w'$.
  The proof is by a straightforward induction on the length of~$\phi$
  minus the length of~$\psi$. We omit it.

  We now prove that $\mathsf{gate}_{w'}^{\psi,J,\nu}\equiv\theta_{w'}(\fPsi{w'}{\psi,J,\nu})$.
  We reason by an inner induction on the size of $\gamma$ (i.e., on
  the number of its connectives). The base case (when the size
  of~$\gamma$ is~$1$) is a special case of the step case, and is
  therefore omitted.
  For the step case, consider an arbitrary call to \ls{Eval} with
  parameters $\mathsf{gate}_{w'}^{\psi,J,\nu}$ and
  $[\alpha^{K,\mu}\mapsto b]$. Clearly, $\alpha$ is a direct subformula
  of~$\psi$.
  If $\alpha$ is atomic, then $\alpha=p(\overline{x})$ for some $p\in
  P$, and \ls{PropagateTruthValue($\alpha$, $K$, $\mu$, $b$)} was called
  from the \ls{PropagateDataValue} procedure. Therefore,
  $b=\osem{w',k,\mu}{\alpha}$.
  If $\alpha$ is not atomic, %
  then \ls{PropagateTruthValue($\alpha$, $K$, $\mu$, $b$)} was
  called either from \ls{PropagateDataValue} or 
  from \ls{PropagateTruthValue} (recursively).
  From the conditions under which the call was made (namely, that
  \ls{IsBool($\mathsf{gate}_{w'}^{\alpha,K,\mu}$) returns true}), we
  deduce in all cases that $b\equiv
  \mathsf{gate}_{w'}^{\alpha,K,\mu}$. From the induction hypothesis
  and Lemma~\ref{lem:propformulaMTLdata}, it follows that $b =
  \osem{w',k,\mu}{\alpha}$.
  Thus, in both cases, $b = \osem{w',k,\mu}{\alpha}$.
  This also tells us that, for different calls to \ls{Eval}, a
  proposition $\alpha^{K,\mu}$ cannot be replaced with different
  Boolean values. That is, we have shown that
  $\mathsf{gate}_{w'}^{\psi,J,\nu} \equiv \theta(\mathsf{gate}_{w\to
    w'}^{\psi,J,\nu})$, for some substitution~$\theta$ that replaces
  propositions $\alpha^{K,\mu}$ with $\osem{w',k,\mu}{\alpha}\in\Two$.
  
  To conclude the (T3.2) case, it suffices to show that for any
  proposition $\alpha^{K,\mu}$ of $\mathsf{gate}_{w\to
    w'}^{\psi,J,\nu}$ such that $\mathsf{gate}_{w\to
    w'}^{\psi,J,\nu}$ depends on $\alpha^{K,\mu}$ and 
  $\alpha^{K,\mu}\in\pdef(\theta_{w'})\setminus\pdef(\theta_{w})$,
  we have $\alpha^{K,\mu}\in\pdef(\theta)$.
  That is, we have that \ls{PropagateTruthValue($\alpha$, $K$, $\mu$, $b$)} is called,
  where $b = \osem{w',k,\mu}{\alpha}$.
  As $\alpha^{K,\mu}\in\pdef(\theta_{w'})\setminus\pdef(\theta_{w})$,
  we have that either $\osem{w,k,\mu}{\alpha}\notin\Two$ or
  $\mu\notin\Val_{w}(\alpha,K)$.
  Note first that as $\alpha^{K,\mu}\in\pdef(\theta_{w'})$, we have
  that $\mu\in\Val_{w'}(\alpha,K)$.  We now make a case distinction.

  \noindent\emph{Case I}:
  $\mu\in\Val_{w}(\alpha,K)$. Then $\osem{w,k,\mu}{\alpha}\notin\Two$.
  Therefore, $\mathsf{gate}_{w}^{\alpha,K,\mu}$ exists (by
  Lemma~\ref{lem:key_existence}); however,
  $\mathsf{gate}_{w}^{\alpha,K,\mu}$ is not semantically equivalent to a Boolean
  constant. As $\mu\in\Val_{w'}(\alpha,K)$,
  $\mathsf{gate}_{w'}^{\alpha,K,\mu}$ exists, by
  Lemma~\ref{lem:key_existence}. Also, from the inner induction
  hypothesis, $\mathsf{gate}_{w'}^{\alpha,K,\mu}\equiv b$.
  Therefore, \ls{Eval} was called on
  $\mathsf{gate}^{\alpha,K,\mu}$ while executing
  \ls{PropagateTruthValue}. Thus,
  \ls{PropagateTruthValue($\alpha$, $K$, $\mu$, $b$)} is called.

  \noindent\emph{Case II}: $\mu\notin\Val_{w}(\alpha,K)$. Since
  $\mu\in\Val_{w'}(\alpha,K)$, then, as in the proof of
  Lemma~\ref{lem:key_existence}, we obtain that
  \ls{PropagateDataValue} is called with parameters $\alpha$, $K$, $\mu[x\mapsto\bot]$, $[x\mapsto d]$,
  where $x$ and $d$ are the variable frozen by the current
  transformation and the corresponding value, respectively.  Again, since the
  $\mathsf{gate}_{w'}^{\alpha,K,\mu}\equiv b$ by the
  inner induction hypothesis, we have that
  \ls{PropagateTruthValue($\alpha$, $K$, $\mu$, $b$)} is called from
  \ls{PropagateDataValue}.
  This concludes the (T3.2)~case.

  \emph{Transformation~(T1).}
  Let $\tau$ and $K$ be the timestamp and the interval returned by
  \ls{DeltaT1($t$)}, respectively.  Note that $\tau\in K$ and we
  assume that $\tau>0$.  
  For an interval $H$ of $w'$, we define $\hat{H}:=H$ if
  $H\not\subseteq K$ and $\hat{H}:=K$ if $H\subseteq K$.
  As $\nu\in\Val_{w'}(\psi,J)$, we have that
  $\nu\in\Val_{w}(\psi,\hat{J})$.  From Lemma~\ref{lem:key_existence},
  we obtain the existence of
  $\mathsf{gate}_{w}^{\psi,\hat{J},\nu}$.

  We first remark that the procedure \ls{AddTimePoint} creates
  $\mathsf{gate}^{\psi,J,\nu}$ in its first \ls{foreach} loop from
  $\mathsf{gate}^{\psi,\hat{J},\nu}$ by \ls{Clone}, if $\hat{J}=K$.
  In \ls{AddTimePoint}'s second \ls{foreach} loop, the procedures
  \ls{Rename}, \ls{RefineNext}, \ls{RefineUntil}, \ls{Instantiate}, or
  \ls{Eval} may modify $\mathsf{gate}^{\psi,J,\nu}$.  Note that
  \ls{Rename}, \ls{RefineNext}, \ls{RefineUntil}, or \ls{Instantiate}
  are directly called from \ls{AddTimePoint} and at most once.  In
  contrast, \ls{Eval} is called from \ls{PropagateTruthValue}, and
  \ls{Eval} may modify $\mathsf{gate}^{\psi,J,\nu}$ multiple times.
  Furthermore, \ls{Eval}'s modifications happen after modifications by
  \ls{Instantiate}, which in turn happen after modifications by
  \ls{Rename}, \ls{RefineNext}, or \ls{RefineUntil}.
  The reason is that the loop iterates top-down over $\phi$'s formula
  structure. This means, if \ls{Eval} modifies
  $\mathsf{gate}^{\psi,J,\nu}$ in the iteration for some state
  variable $\mathsf{gate}^{\gamma,H,\mu}$, then $\gamma$ is a
  subformula of $\psi$. In particular, modifications by \ls{Rename},
  \ls{RefineNext}, or \ls{RefineUntil} on $\mathsf{gate}^{\psi,J,\nu}$
  have been carried out in an earlier iteration, namely, the one for
  $\mathsf{gate}^{\psi,J,\nu}$.
  We denote by $\mathsf{gate}_{w\to w'}^{\psi,J,\nu}$ the value of the
  state variable $\mathsf{gate}^{\psi,J,\nu}$ after modifications by
  \ls{Rename}, \ls{RefineNext}, or \ls{RefineUntil}, and before
  modifications by \ls{Instantiate} or \ls{Eval}.

  The proof of
  $\mathsf{gate}_{w'}^{\psi,J,\nu}\equiv\theta_{w'}(\fPsi{w'}{\psi,J,\nu})$
  comprises two parts. The first part shows that
  $\mathsf{gate}_{w\to w'}^{\psi,J,\nu} \equiv
  \delta(\fPsi{w'}{\psi,J,\nu})$, where $\delta$ behaves like
  $\theta_w$, except that it carries over the truth value assignment
  for propositions with the interval $K$ to the propositions
  originating from splitting $K$.  That is, we define
  \begin{equation*}
    \delta(p):=\begin{cases}
      \theta_w(\gamma^{\hat{H},\mu}) & 
      \text{if $p=\gamma^{H,\mu}$ and $\gamma^{\hat{H},\mu}\in\pdef(\theta_w)$,}
      \\
      \theta_w(\tcp_I^{\hat{H},\hat{L}}) &
      \text{if $p=\tcp_I^{H,L}$ and $\tcp_I^{\hat{H},\hat{L}}\in\pdef(\theta_w)$,}
      \\
      \theta_w(p) &
      \text{if $p\in\pdef(\theta_w)$ and $p$ is of the form
        $\tpp^H$ or $\overline{\tpp}^H$.}
    \end{cases}
  \end{equation*}
  Note that if $p\in\pdef(\theta_w)$, then $\theta_w(p)\in\Two$.  Also
  note that $\delta$ is undefined for propositions of the form
  $\tpp^H$ and $\overline{\tpp}^H$ with $H\subseteq K$.
  The second part, which we omit,
  since it is analogous to the second part of the previous~(T3.2)
  case, uses the first part to show that
  $\mathsf{gate}_{w'}^{\psi,J,\nu} \equiv
  \theta_{w'}(\fPsi{w'}{\psi,J,\nu})$.

  For the first part, it suffices to show that the relevant parts of
  $\mathsf{gate}_{w\to w'}^{\psi,J,\nu}$ are semantically equivalent
  to their relevant counterparts in $\delta(\fPsi{w'}{\psi,J,\nu})$.
  Indeed, note that $\theta_w(\fPsi{w}{\gamma,\hat{J},\nu})$ is
  determined by its relevant parts.  As
  $\mathsf{gate}_{w}^{\psi,\hat{J},\nu} \equiv
  \theta_w(\fPsi{w}{\psi,\hat{J},\nu})$ by the induction hypothesis,
  $\mathsf{gate}_{w}^{\psi,\hat{J},\nu}$ and therefore also
  $\mathsf{gate}^{\psi,J,\nu}$ when newly created are determined by
  their relevant parts.  Finally,
  $\mathsf{gate}_{w\to w'}^{\psi,J,\nu}$ is determined by its relevant
  parts, as $\mathsf{gate}^{\psi,J,\nu}$ is only altered through the
  procedures of the interface presented in
  Section~\ref{subsubsec:state} (page~\pageref{def:API}).
  In the following, let $\chi$ be a direct subformula of $\psi$ and
  $H$ an interval of~$w'$.  We assume that the
  $(\chi,\hat{H})$-relevant part in $\fPsi{w}{\psi,\hat{J},\nu}$
  exists. Otherwise, there is nothing to prove.  Furthermore,
  $\Theta(\Psi)$ denotes the $(\chi,H)$-relevant part of the
  propositional formula~$\Psi$, and $\hat{\Theta}(\Psi)$ denotes its
  $(\chi,\hat{H})$-relevant part.
  
  There is a substitution~$\zeta$ such that
  $\Theta(\fPsi{w'}{\psi,J,\nu}) =
  \zeta\big(\hat{\Theta}(\fPsi{w}{\psi,\hat{J},\nu})\big)$. For
  instance, if $\psi=\neg\alpha$, then
  $\zeta=[\alpha^{\hat{H},\nu}\mapsto\alpha^{H,\nu}]$, and if
  $\psi=\alpha\until_I\beta$ and $\chi=\beta$, then
  $\zeta=\zeta_2\circ\zeta_1$, where $\circ$ denotes function
  composition, and $\zeta_1$ and $\zeta_2$ are the substitutions
  $[\alpha^{\hat{H},\nu}\mapsto\alpha^{H,\nu},\beta^{\hat{H},\nu}\mapsto\beta^{H,\nu},
  \tpp^{\hat{H}}\mapsto\tpp^H,
  \overline{\tpp}^{\hat{H}}\mapsto\overline{\tpp}^H]\cup
  [\tcp^{\hat{H},L}_I\mapsto\tcp^{H,L}_I\mid\text{$L$ is an interval
    in $w$}]$ and
  $[\tcp^{L,\hat{H}}_I\mapsto\tcp^{L,H}_I \mid\text{$L$ is an interval
    in $w'$}]$, respectively.
  We have that
  \begin{equation}
    \label{eq:relevantparts}
    \Theta\big(\delta(\fPsi{w'}{\psi,J,\nu})\big) =
    \delta\big(\Theta(\fPsi{w'}{\psi,J,\nu})\big) =
    \delta\Big(\zeta\big(\hat{\Theta}(\fPsi{w}{\psi,\hat{J},\nu})\big)\Big) =
    \zeta\Big(\theta_w\big(\hat{\Theta}(\fPsi{w}{\psi,\hat{J},\nu})\big)\Big) =
    \zeta\Big(\hat{\Theta}\big(\theta_{w}(\fPsi{w}{\psi,\hat{J},\nu})\big)\Big)
    \,.
  \end{equation}
  The first and the last equalities hold because a relevant part is
  determined even after some propositions have been replaced by
  Boolean constants. The other two equalities follow from the
  definitions.
  
  We remark that $\zeta$ is the substitution applied by
  \ls{AddTimePoint} to the relevant parts of
  $\mathsf{gate}^{\psi,J,\nu}$ when this state variable depends on
  some proposition $p$ with the interval~$K$.
  For instance, if $\psi=\alpha\until_I\beta$ and $\chi=\beta$, then
  $\zeta_1$ is the substitution applied to the \ls{anchor} variable
  and $\zeta_2$ is the substitution applied to the state variable,
  when $H\subseteq K$, in \ls{RefineUntil}
  (cf. Listing~\ref{lst:auxUNTIL}).
  Therefore, we obtain the semantic equivalence
  \begin{equation}
    \label{eq:refineuntil}
    \Theta(\mathsf{gate}_{w\to w'}^{\psi,J,\nu}) \equiv
    \zeta\big(\hat{\Theta}(\mathsf{gate}_{w}^{\psi,\hat{J},\nu})\big)
    \,.
  \end{equation}
  This equivalence holds even when
  $\mathsf{gate}_{w}^{\psi,\hat{J},\nu}$ does not depend on a
  proposition $p$ with the interval~$K$.  In this case,
  $\mathsf{gate}_{w\to w'}^{\psi,J,\nu} =
  \mathsf{gate}_{w}^{\psi,\hat{J},\nu}$ and there is no proposition in
  $\mathsf{gate}_{w}^{\psi,\hat{J},\nu}$ for $\zeta$ to substitute.

  The right-hand sides of the semantic equivalence
  in~(\ref{eq:refineuntil}) and of the right-most equality
  in~(\ref{eq:relevantparts}) are semantically equivalent by the
  induction hypothesis.  We conclude that the left-hand sides
  in~(\ref{eq:refineuntil}) and of the left-most equality
  in~(\ref{eq:relevantparts}) are also semantically equivalent.

  \emph{Transformation~(T2).}
  Let $K$ be the interval returned by \ls{DeltaT2($t$)}.
  The proof is similar to the (T3.2) case.  We only remark
  that we use
  $\fPsi{w'}{\gamma,J,\nu}\equiv[\tpp^K\mapsto\false](\fPsi{w}{\gamma,J,\nu})$
  in the base case of the corresponding induction.

  \emph{Transformation~(T3.1).}
  The proof is similar to the (T3.2) case and is therefore
  omitted.
\end{proof}

\section{Experimental Evaluation}
\label{sec:eval}

We have implemented the online algorithms for monitoring from 
Sections~\ref{sec:prop} and~\ref{sec:data} in a prototype tool,
written in the programming language Go~(\url{golang.org}).  In this
section, we experimentally evaluate the performance of our prototype
tool, focusing on the impact of different message orderings.  

\subsubsection*{Setup}

\begin{figure}[t]
  \centering
  \footnotesize
  \begin{gather}
    \tag{P1}
    \always 
    \freeze{cid}{c} \freeze{tid}{t} \freeze{sum}{a} 
    \mathit{trans}(c,t,a)\wedge a>2000 
    \rightarrow \eventually_{[0,3]} \mathit{report}(t)
    \\
    \tag{P2}
    \always 
    \freeze{cid}{c} \freeze{tid}{t} \freeze{sum}{a}{}
    \mathit{trans}(c,t,a) \land a>2000\to
    \always_{(0,5]}\freeze{tid}{t'} \freeze{sum}{a'} 
    \mathit{trans}(c,t',a') \rightarrow a'\leq 2000
    \\
    \tag{P3}
    \always 
    \freeze{cid}{c}{} \freeze{tid}{t} \freeze{sum}{a}{}
    \mathit{trans}(c,t,a) \land a>2000\to
    \big((\freeze{tid}{t'} \freeze{sum}{a'} \mathit{trans}(c,t',a')\rightarrow t=t')
    \weakuntil \mathit{report}(t)\big)
    \\
    \tag{P4}
    \always 
    \freeze{cid}{c}{} \freeze{tid}{t} \freeze{sum}{a}{}
    \mathit{trans}(c,t,a) \land a>2000\to
    \always_{[0,6]}
    \freeze{tid}{t'} \freeze{sum}{a'} \mathit{trans}(c,t',a')
    \to \eventually_{[0,3]} \mathit{report}(t') 
  \end{gather}
  \noindent\rule{\textwidth}{0.1pt}
  \begin{gather}
    \tag{P1$'$}
    \always 
    \mathit{transaction} \wedge \mathit{suspicious} \to 
    \eventually_{[0,3]} \mathit{report}
    \\
    \tag{P2$'$}
    \always 
    \mathit{transaction} \wedge \mathit{suspicious} \to
    \always_{(0,5]} \mathit{transaction} \to \neg\mathit{suspicious}
    \\
    \tag{P3$'$}
    \always 
    \mathit{transaction} \wedge \mathit{suspicious} \to
    \big(
      (\mathit{transaction}\rightarrow\eventually_{[0,3]}\mathit{report})
      \weakuntil 
      \mathit{unflag}
    \big)
    \\
    \tag{P4$'$}
    \always 
    \mathit{transaction} \wedge \mathit{suspicious} \to
    \always_{[0,6]}
    \mathit{transaction} \to \eventually_{[0,3]} \mathit{report}
  \end{gather}
  \caption{Formulas used in the experimental evaluation.}
  \label{fig:specifications}
\end{figure}
For our experimental evaluation, we use a standard desktop computer
with a 3.3\,GHz CPU (Intel Xeon E3-1230V2), 16\,GB of RAM, and the
Linux operating system (Ubuntu~16.04).  The prototype was compiled
with the Go compiler~1.10 and executed single-threaded.  Furthermore,
we use the formulas in Figure~\ref{fig:specifications}, which vary in
their temporal requirements and the data involved. (P1) to (P4)
express compliance policies from the banking domain and are variants
of policies that have been used in previous case
studies~\cite{Basin_etal:rv_mfotl}.  (P1$'$) to (P4$'$) are
propositional versions of (P1) to (P4), except (P3$'$), which has an
additional temporal connective and accounts for the additional
event~$\mathit{unflag}$.

In the following, we provide some intuition on (P1) to (P4).
We start by explaining the predicate symbols that model the events
that the banking system is assumed to log or transmit to the monitor.
The predicate $\mathit{trans}(c,t,a)$ represents the execution of the
transaction~$t$ of the customer~$c$ transferring the amount~$a$ of
money.  The predicate $\mathit{report}(t)$ represents the reporting of
the transaction~$t$, that is, $t$ is marked as suspicious.
Note that a message sent to the monitor describes an event and the
register values. For instance, when executing a transaction, the
registers $\mathit{tid}$ and $\mathit{cid}$ store the identifiers of
the transaction and the customer; the amount of the transaction is
stored in the register $\mathit{sum}$. For a $\mathit{report}$ event,
the register $\mathit{tid}$ stores the identifier of the transaction
whereas the other registers for the customer and the amount store the
default value~$0$.

The formula~(P1) requires that a transaction~$t$ of a customer~$c$
must be reported within at most three seconds if the transferred
amount~$a$ exceeds the threshold of \$2,000.
(P2) to~(P4) are variants of~(P1).  (P2) requires that whenever a
customer~$c$ makes a transaction that exceeds the threshold, then any
of $c$'s future transactions within the next five seconds must not
exceed the threshold.
(P3) requires that whenever a customer~$c$ makes a transaction~$t$
that exceeds the threshold, then $c$ is not allowed to make further
transactions until the transaction $t$ is reported. Note that the
syntactic sugar $\weakuntil$ (``weak until'') is used here instead of
the primitive temporal connective~$\until$. We do not require that the
transaction must eventually be reported.
(P4) requires that whenever a customer~$c$ makes a transaction that
exceeds the threshold, then any of $c$'s transactions within the next
six seconds must be reported within three seconds.

Finally, we synthetically generate log files.  Each log spans over 60
seconds and contains one event per time point, for instance,
corresponding to a single
transaction.  The number of events in a log is determined by the
\emph{event rate}, which is the approximate number of events per
second.  For each time point~$i$, with $0\leq i<60$, the number of
events with a timestamp in the time interval $[i,i+1)$ is randomly
chosen within $\pm$10\% of the event rate. For instance, a log with
event rate 100 comprises approximately 6,000 events.  The events and
their parameters are randomly chosen such that the number of
violations is in a provided range.
Note that when the monitor processes an event it performs several
state updates, which correspond to the
transformations~(T\ref{enum:observation_split}),~(T\ref{enum:observation_removal}),
and~(T\ref{enum:observation_data}): (1)~The monitor adds a new time
point with the event's timestamp, (2)~it may remove one or more
nonsingleton intervals for which the monitor will not receive any
events in the future, and (3)~it propagates data and truth values.
Since the messages can be received in any order by the monitor, it
must determine whether all events within a time period have
been received.  To this end, we attach to each event a sequence
number.  The monitor removes the nonsingleton interval $J$ if the
event's sequence number for the time point before $J$ is the
predecessor of the event's sequence for the time point after $J$. In
Section~\ref{sec:app}, we consider the general setting where the
monitor receives events from different sources.

\begin{figure}[t]
  \centering
  \subfigure[in-order (\MTLdata)]{\includegraphics[scale=.18]{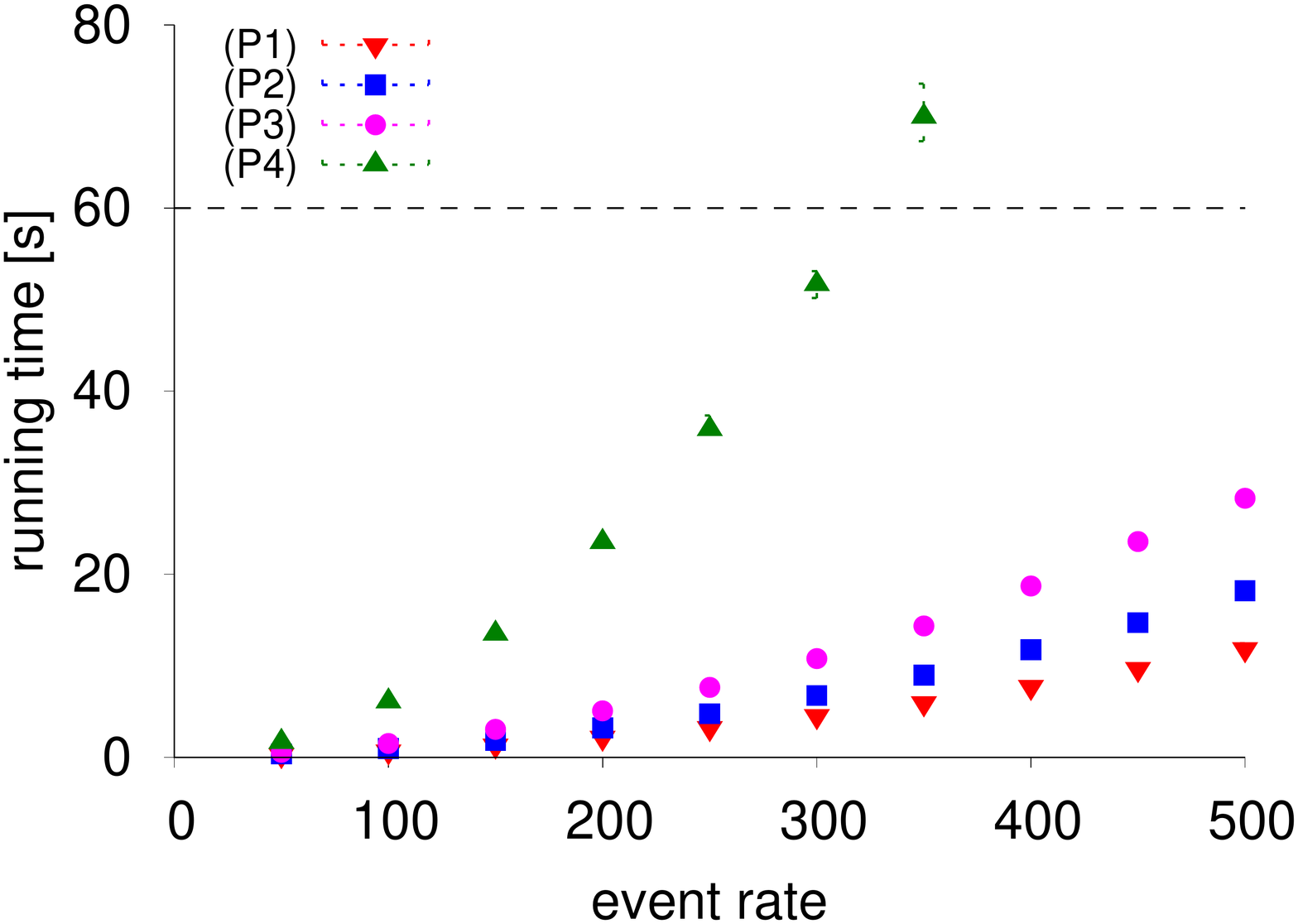}}
  \qquad\qquad
  \subfigure[out-of-order (\MTLdata)]{\includegraphics[scale=.18]{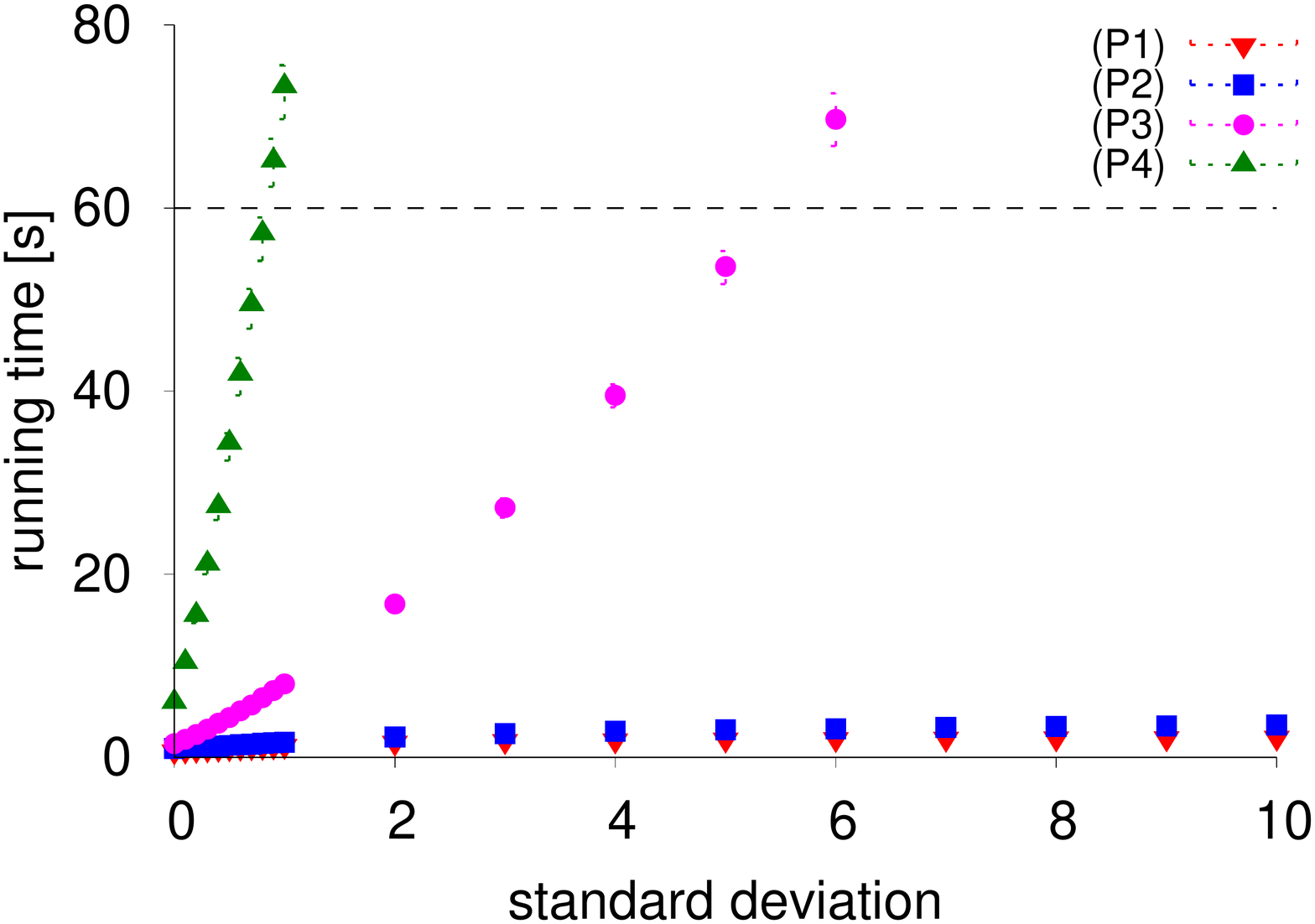}}

  \subfigure[in-order (MTL)]{\includegraphics[scale=.18]{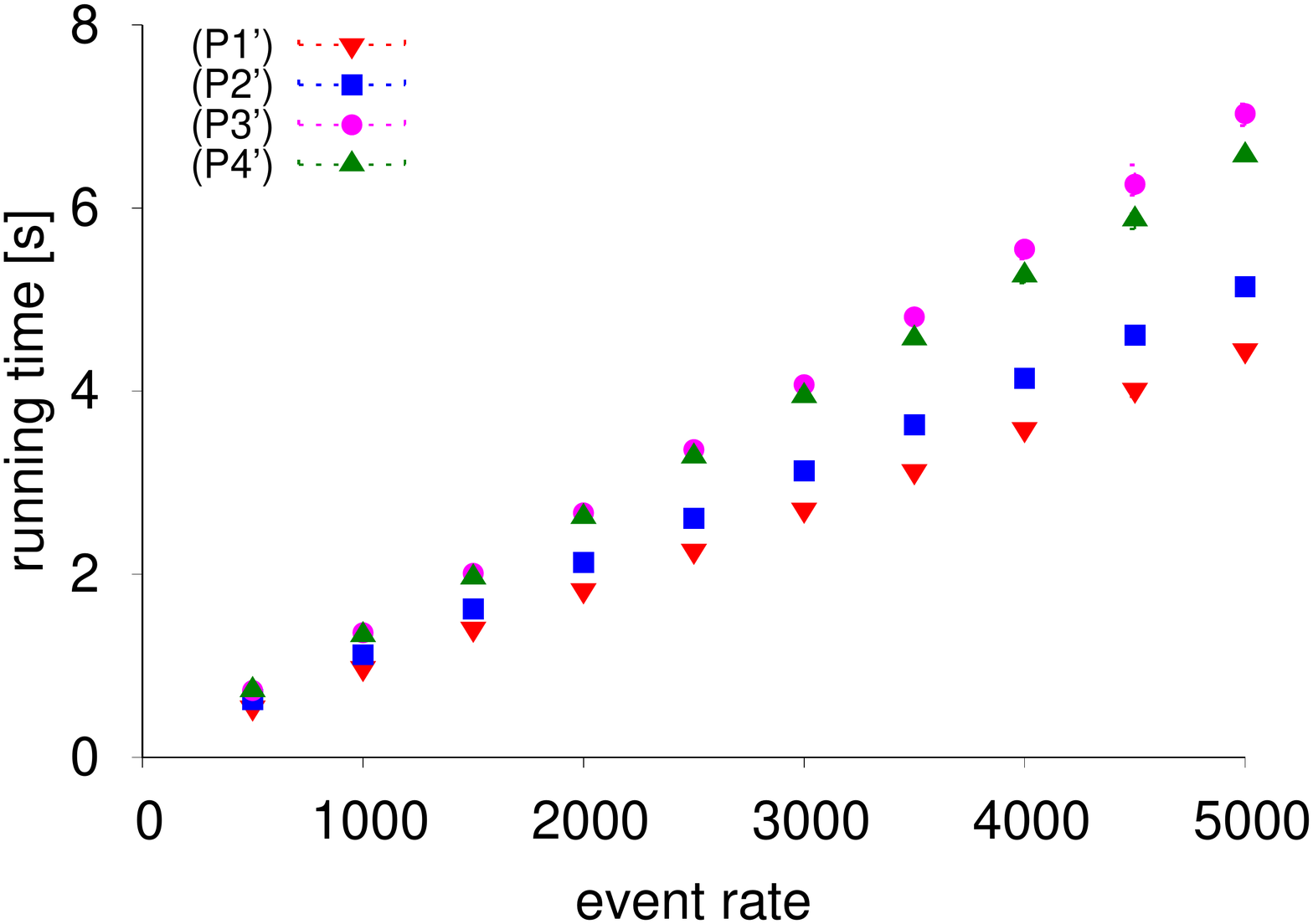}}
  \qquad\qquad
  \subfigure[out-of-order (MTL)]{\includegraphics[scale=.18]{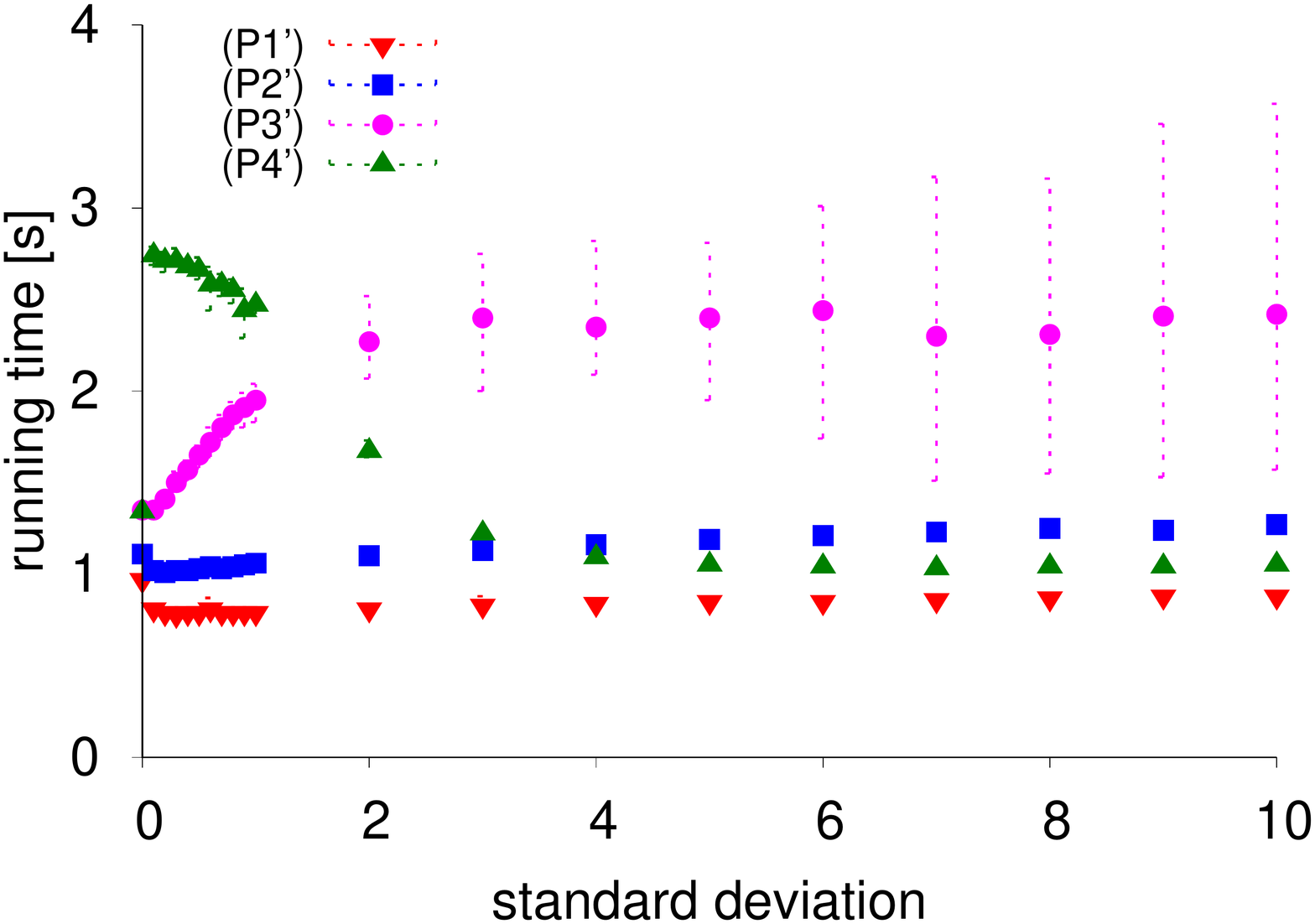}}

  \caption{Running times. Here each data point shows the average over five
    logs together with the minimum and maximum, which are very close
    to the average, except for (P3$'$) in (d).}
  \label{fig:performance}
\end{figure}

\subsubsection*{In-order Delivery}

In our first setting, messages are received ordered by their
timestamps and are never lost. Namely, all events of the log are
processed in the order of their timestamps. 
Figure~\ref{fig:performance}(a) shows the prototype's running times
for different event rates.  Note that each log spans 60~seconds and a
running time below 60 seconds essentially means that the events in the
log could have been processed online.  The dashed horizontal lines
mark this border. Memory usage does not exceed~50\,MB, except for (P4)
where it increases to around 300\,MB.

\subsubsection*{Out-of-order Delivery}

In our second setting, messages can arrive out of order, but they are
not lost. We control the degree of message arrival disruption as
follows.
For the events in a generated log file, we choose their arrival times,
which determine the order in which the monitor processes them.  An
event's arrival time is derived from the event's timestamp by
offsetting it by a random delay with respect to the normal
distribution with a mean of $\mu$ time units and a standard
deviation~$\sigma$.
Intuitively, the degree of ``out-of-orderness'' increases for larger
standard deviations.  For the degenerate case $\sigma=0$, the random
delay is $0$ and the reordered log is identical to the original log.
For $\sigma>0$, the random delay is, for example, between $\mu-\sigma$
and $\mu+\sigma$ with probability~$0.68$ and with probability~$0.95$
between $\mu-2\sigma$ and $\mu+2\sigma$. This means that for different
standard deviations $\sigma,\sigma'>0$, the random delays for~$\sigma$
are more likely spread over a larger range than for~$\sigma'$ when
$\sigma>\sigma'$, which in the end results in reordered logs where
the events are less ordered.  Finally, we remark that the choice of
$\mu$ does not impact the event reordering; with a large enough mean
$\mu$, the random delays are (most likely) positive.

Figure~\ref{fig:performance}(b) shows the prototype's running times on
logs with the fixed event rate 100 for different deviations, where
$\mu$ is fixed to $10$ and $\sigma$ ranges over different values
between $0$ and $10$.
For instance, for (P1), the logs are processed in under a second for
$\sigma=0$ and around two seconds for $\sigma=10$.  Memory usage stays
moderate for small deviations (below 100\,MB for $\sigma<1$), but can
increase significantly for larger deviations (almost 1\,GB for (P3)
with $\sigma=5$ and (P4) with $\sigma=1$).  Reasons for this are the
larger data structure and also the queued messages, since messages
arrive faster than they can be processed by the monitor.

\subsubsection*{Interpretation}

For~(P1) to~(P4), the running times are nonlinear in the event rate.
This is expected from Theorem~\ref{thm:decidable}.  The growth is
mainly caused by the data values occurring in the events.  A log with
a higher event rate contains more different data values and the
monitor's state must account for these. In particular, the graph-based
data structure contains multiple nodes for a subformula~$\gamma$ and
an interval~$J$, but different partial valuations~$\nu$.
As expected, (P1) is the easiest to monitor. In addition to the
outermost temporal connective $\always$, it only has a single temporal
connective with a three second bound and a single block of freeze
quantifiers.  (P4) is hardest to monitor, since it has two blocks of
freeze quantifiers and two bounded temporal connectives, which are
nested, resulting in a time window of nine seconds.
The running times increase when messages are received out of order,
which is also expected. For (P1) and (P2), however, the increase is
almost insignificant. In contrast, for (P3) and (P4), the running
times increase rapidly. This can be traced back to the formulas'
larger time window and the two blocks of freeze quantifiers.

In the propositional setting, the running times only increase linearly
with respect to the event rate and logs are processed significantly
faster.  Furthermore, the out-of-order delivery of events has only a minor
impact of the running times.  See Figures~\ref{fig:performance}(c)
and (d), where the event rate is one order of magnitude higher. Our
prototype processes most events in a fraction of a millisecond, and a
noticeable amount of the computation time is actually spent in parsing
the events.
However, some care must be taken when comparing the figures of the
propositional setting with the setting with data values.  First, the
formulas express different policies. For instance, in (P1$'$) and
(P4$'$) a report might discharge multiple transactions.  Second, the
logs for the propositional settings differ from the logs for the
formulas~(P1) to~(P4). In particular, the events in the log files
generated for the propositional settings do not account for different
customers.
Overall, one pays a price at runtime for the expressivity gain given
by the freeze quantifier.  This price can be traced back to the number
of nodes in the graph-based data structure that the monitor
maintains. For MTL, the number of nodes in the data structure for an
interval is bounded by the number of subformulas, whereas for
\MTLdata, the number of nodes for an interval is dominated by the
different data values that occur in the messages.

To put the experimental results in perspective, we also compare our
prototype with the MONPOLY tool~\cite{monpoly}.
MONPOLY's specification language is, like \MTLdata, a point-based
real-time logic.  It is richer than \MTLdata in that it admits
existential and universal quantification over domain
elements. However, MONPOLY specifications are syntactically restricted
in that temporal future connectives must be bounded (except for the
outermost connective $\always$).
Thus, (P3) does not have a counterpart in MONPOLY's specification
language.
MONPOLY handles the counterparts of (P1), (P2), and (P4) significantly
faster, up to three orders of magnitude.  In the propositional
setting, the running times only differ by a factor less than five.
Comparing the performance of both tools should, however, be taken with
a grain of salt.  First, while MONPOLY has undergone several rounds of
optimizations, our prototype is fairly unoptimized.  More significant,
MONPOLY only handles the restrictive setting where messages must be
received in order, and MONPOLY outputs violations for specifications
with (bounded) future only after all events in the relevant time
window are available, whereas our prototype outputs verdicts
promptly. For instance, for the formula $\always_{[0,3]} p$, if $p$
does not hold at the time point $i$ with timestamp~$\tau$, then our
prototype outputs the corresponding verdict directly after processing
the time point~$i$, whereas MONPOLY reports this violation at the
first time point with a timestamp larger than $\tau+3$.

In summary, our experimental evaluation shows that one pays a high
price to handle an expressive specification language together with
message delays. Nevertheless, our prototype's performance is
sufficient to monitor systems that generate hundreds of events per
second; in a propositional setting, the prototype already handles
several thousand events per second.  Furthermore, the prototype can be
used as a starting point for more efficient implementations.

\section{Monitoring Application}
\label{sec:app}

In this section, we describe a deployment of the online algorithms
presented for verifying distributed systems at runtime.  We first
describe the system design and the underlying system assumptions.  We
also discuss some practical aspects and consequences of our
deployment.

\subsection{Deployment}

We target distributed systems with multiple interacting components.
The objective is to determine at runtime whether the system's
behavior, as observed and reported by the components, satisfies or
violates a given specification~$\phi$ at some or all time points.

We sketch our system design, which extends the original system with an
additional monitoring component for~$\phi$, where $\phi$ is a closed
\MTLdata formula.  The original system components are instrumented
such that they report their performed actions to the monitoring
component by sending dedicated messages over a unidirectional channel.
Each such message also names the performing component and the
time. Furthermore, the message contains a sequence number.  That is,
each component maintains a counter, which counts the actions it has
performed so far, and includes the counter's value with every message
sent to the monitor.  With these numbers, the monitor can determine if
no action has been performed in a given interval (see
Section~\ref{subsubsec:sequencenumber} for details).
In addition to the messages that describe the performed actions, a
component can send ``alive'' messages.  They inform the monitoring
component that the respective component has not performed any action
for a while.
In summary, there are two types of messages:
$\mathit{action}(C,\tau,s,d)$ and $\mathit{alive}(C,\tau,s)$, where
$C$ is the component name, $\tau\in\Qpos$ the timestamp, $s$ the
component's sequence number, and $d$ a description of the performed
action.

Before providing further details and discussing the consequences of
this deployment, we list and comment on the assumptions on the
underlying system model.

\paragraph{{\sf A1}: The system is static.}  

This means that no system components are created or removed at
runtime.  Furthermore, the monitor is aware of the existence of all
the system components.  Note that this assumption can easily be
eliminated by building into our approach a mechanism to register
components before they become active and unsubscribing them when they
become inactive.  To register components we can, for example, use a
simple protocol where a component sends a registration request and
waits until it receives a message that confirms the registration.

\paragraph{{\sf A2}: Communication between components is asynchronous and
  unreliable.  However, messages are neither tampered with nor
  delivered to wrong components.}

Asynchronous, unreliable communication means that messages may be
received in an order different from which they were sent, and some
messages may be lost and therefore never received.  Note that message
loss covers the case where a system component crashes without
recovery.  A component that stops executing is indistinguishable to
other processes from one that stops sending messages or none of its
messages are received.  We explain in Section~\ref{subsubsec:recovery}
that it is also straightforward to handle the case where crashed
components can recover.  The assumption ruling out tampering and
improper delivery can be discharged in practice by adding information
to each message, such as a recipient identifier and a cryptographic
hash value, which are checked when receiving the message.

\paragraph{{\sf A3}: System components, including the monitor, are
  trustworthy.}

This means, in particular, that the components correctly report their
observations and do not send bogus messages.

\paragraph{{\sf A4}: Reported actions are consistent.}

This means that messages from components to the monitor do not
contradict each other.  For instance, there are never two messages to
the monitor such that one is saying that a proposition $p$ is true at a
time~$\tau\in\Qpos$ and the other one is saying that $p$ is false at
$\tau$.

\paragraph{{\sf A5}: The system components perform infinitely many
  actions in the limit.}

This guarantees that the observable system behavior is in the limit a
timed word. Note that \MTLdata specifies properties about infinite
system behavior. In particular, \MTLdata's three-valued semantics over
observations approximates infinite behavior as the interval of an
observation's last letter is unbounded and can always be refined.  We
would need to use another specification language if we want to express
properties about finite system behavior.  However, note that a monitor
is always aware of only a finite part of the observed system behavior.
Furthermore, since channels are unreliable and messages can be lost, a
monitor might even, in the limit, be aware only of a finite part of
the infinite system behavior.

\subsection{Discussion}
\label{subsec:discussion}

\subsubsection{State Updates}
\label{subsubsec:sequencenumber}

Each message may result in multiple updates of the monitor's state. A
message $\mathit{action}(C,\tau,s,d)$ results in adding a time point
with the timestamp $\tau$, the propagation of data and truth values,
and also the removal of nonsingleton intervals.
A message $\mathit{alive}(C,\tau,s)$ may result in the removal of
nonsingleton intervals, which in turn may trigger the propagation of
truth values.

With the messages' sequence numbers, the monitor can infer which
intervals can be removed.  When monitoring a single system component,
this inference is obvious.
We sketch the general case when monitoring a system with the
components $C_0,\dots,C_m$.  Let $J_0,\dots,J_n$ be the nonsingleton
intervals of the letters in an observation. The monitor labels each of
these intervals with a set $S_{J_j}$ of the components from which it
may receive an $\mathit{action}$ message with a timestamp in $J_j$ in
the future.
Additionally, the monitor maintains for each component $C_i$ triples
of the form $(s,I,s')$, where $I$ is an interval and $s,s'\in\Nat$
with $s\leq s'$. The intuition is that all $\mathit{action}$ messages
from $C_i$ with a timestamp in $I$ have been received by the monitor,
and $s$ and $s'$ are the smallest and largest sequence number of these
messages, respectively.  The monitor adds a triple $(t,\{\tau\},t)$
when receiving from $C_i$ a message with the timestamp $\tau$ and the
sequence number~$t$.  The monitor also merges triples when
possible. For example, the triples $(s,I,s')$ and $(t,\{\tau\},t)$
with $t=s-1$ or $t=s'+1$ are merged into the triple
$(\min\{t,s\},I\Cup \{\tau\},\max\{t,s'\})$, where $I\Cup\{\tau\}$ is
the smallest interval that contains $I$ and $\{\tau\}$.
Whenever one of the intervals~$J_j$ is a subset of the interval of
such a triple, the monitor removes $C_i$ from the set~$S_{J_j}$.  When
$S_{J_j}$ becomes empty, the monitor removes the letter with the
interval~$J_j$ from the observation.
Note that the intervals $J_0,\dots,J_n$ can be ordered and stored in a
balanced search tree. Analogously, the triples can be ordered and also
stored in balanced search trees with pointers to their predecessors
and successors.

\subsubsection{Accuracy of Timestamps}
\label{subsubsec:accuracy}

The monitor's verdicts are computed with respect to the information in
messages that the monitor receives from the system components.  Even
though we assume trustworthy system components (A3), their
observations might not match with the actual system behavior.  In
particular, the timestamp~$\tau$ in a message may be inaccurate
because $\tau$ comes from the clock of a system component that has
drifted from the actual time.  One may wonder in what sense are the
verdicts meaningful.

Consider first the guarantees we have under the additional system
assumption that timestamps are precise and from the domain~$\Qpos$.
Under this assumption, $w_i\sqsubseteq w$, for all $i\in\Nat$, where
the $w_i$'s are observations describing the reported system behavior
and $w$ is a timed word that represents the real system behavior.  It
follows from Theorem~\ref{thm:conservative} that the verdicts computed
from the reported system behavior~$w_i$ are also valid for the system
behavior~$w$.

Assuming precise timestamps is, however, a strong assumption, which does
not hold in practice, since real clocks are imprecise.  Moreover, each
system component uses its local clock to timestamp its messages, and
these clocks might differ due to clock drifts.  In fact, assuming
synchronized clocks boils down to having a synchronized system at
hand.
Nevertheless, we argue that for many kinds of specifications and
systems, relying on timestamps from existing clocks in monitoring is
good enough in practice. First, under stable conditions (like
temperature), state-of-the-art hardware clocks already achieve a high
accuracy and their drifts are, even over a longer time period, rather
small~\cite{Cristian_Fetzer:timed_model}.  Moreover, there are
protocols like the Network Time Protocol (NTP)~(see \url{www.ntp.org})
for synchronizing clocks in distributed systems that work well in
practice. For local area networks, NTP can maintain synchronization of
clocks within one millisecond~\cite{Mills:NTP_improved}.  Overall,
with state-of-the-art techniques, we can obtain timestamps that are
``accurate enough'' for many monitoring applications, for instance,
for checking whether deadlines are met when the deadlines are in the
order of seconds or even milliseconds.  Furthermore, if the monitored
system guarantees an upper bound on the imprecision of timestamps, we
can often account for this imprecision in the specification.  For
example, for checking at runtime that requests are acknowledged within
100 milliseconds, when the imprecision between two clocks is always
less than a millisecond, we can use the formula
$\always \mathit{req}\rightarrow
\once_{[0,1)}\eventually_{[0,101)}\mathit{ack}$ to avoid false alarms.

\subsubsection{Component Crashes}
\label{subsubsec:recovery}

When a system component crashes, its state is lost. For recovery, we
must bring the component into a state that is safe for the system.
To safely restart a system component that is not the monitor, we must
restore its sequence number.  We can use any persistent storage
available to store this number. In case the component crashes while
storing this number, we can increment the restored number by one. This
might result in knowledge gaps for the monitor, since some intervals
will never be identified as complete. However, the computed verdicts
are still sound.
For the recovery of a crashed monitor, we just need to initialize it.
A recovered monitor corresponds to a monitor that has not yet received
any message.  This is safe in the sense that the recovered monitor
will only output sound verdicts.  When the monitor also logs received
messages in a persistent storage, it can replay them to close some of
its knowledge gaps. Note that the order in which these messages are
replayed is irrelevant, and they can even be replayed whenever the
recovered monitor is idle.

\section{Related Work}
\label{sec:related}

In this section, we examine related work. Our focus is on system
verification, in particular, runtime verification, a well-established
area for checking at runtime whether a system's execution fulfills a
given specification.  We structure our discussion along the aspects of
multiple truth values, data values, and distributed systems.

\subsubsection*{Multi-valued Semantics}

Multi-valued semantics for temporal logics are widely used in runtime
verification, see for
example,~\cite{Bauer_etal:rv_tltl,Bauer_Falcone:decentralised_monitor,Scheffel_Schmitz:asynchronous_distributed_rv,Mostafa_Bonakdarbour:decentralized_rv}.
Their semantics extend the classical LTL semantics by also assigning
non-Boolean truth values to finite prefixes of infinite
words~\cite{Bauer_etal:ltl_rv}.  The additional truth values
differentiate when neither some nor all extensions of a finite word
satisfy a formula.  However, in contrast to the presented three-valued
semantics of \MTLdata used in this paper, the Boolean and temporal
connectives are not extended over the additional truth values.
Furthermore, the partial order $\prec$ on the truth values, which
orders them in knowledge, is not considered.  Note that having the
third truth value~$\bot$ at the logic's object level and the partial
order~$\prec$ is at the core of our three-valued semantics for
\MTLdata and our monitoring approach; namely, it is used to account
for a monitor's knowledge gaps.  Another difference is that a
formula's truth value is not defined by the possible extensions of a
finite word. As pointed out in Remark~\ref{rem:correctness}, including
the possible extensions can render monitoring infeasible.

The monitoring approaches by \citeetal{Garg_etal:policy_incomplete}
and \citeetal{BasinKMZ-RV12}, both targeting the auditing of policies
on system logs, also account for knowledge gaps, that is, logs that
may not contain all the actions performed by a system.  Both
approaches handle rich policy specification languages with first-order
quantification and a three-valued semantics.
Garg et al.'s approach~\citeyearpar{Garg_etal:policy_incomplete},
which is based on formula rewriting, is, however, not suited for
online use, since it does not process logs incrementally.  It also
only accounts for knowledge gaps in a limited way, namely, the
interpretation of a predicate symbol cannot be partially unknown, for
example, for certain time periods. Furthermore, their approach is not
complete.
Basin et al.'s approach~\citeyearpar{BasinKMZ-RV12}, which is based on
their prior work~\cite{Basin_etal:rv_mfotl}, can be used online.
However, the problem of how to output verdicts incrementally as prior
knowledge gaps are resolved is not addressed, and thus it does not
deal with out-of-order events.  Moreover, the semantics of the
specification language handled does not reflect a monitor's partial
view about the system behavior. Instead, it is given for infinite data
streams that represent system behavior in the limit.
The runtime-verification approach by
\citeetal{Stoller_etal:RV_stateest} also accounts for gaps in
traces. These gaps are, however, caused by sampling the state of the
monitored system to reduce the runtime-verification overhead and trace
elements are processed ordered.  Furthermore, their approach is not
based on a multi-valued semantics for a temporal logic. Instead, an
{a priori} trained model (namely, a hidden Markov model) for
estimating the likelihood of missing trace elements is used to compute
the probability of the specification's satisfaction.

Multi-valued semantics for temporal logics have also been considered
in other areas of system verification.  For instance,
\citeetal{Chechik_etal:mv_mc} describe a model-checking approach for a
multi-valued extension for the branching-time temporal logic CTL.
Their CTL extension is similar to our  \MTLdata extension in 
that it allows one to reason about uncertainty at the logic's
object level. However, the task they consider is different from
ours. Namely, in model checking, the system model is given---usually
finite-state---and correctness is checked offline with respect to the
model's described executions; in contrast, in runtime verification,
one checks online the correctness of the observed system behavior.
The three-valued semantics for LTL provided by
\citeetal{Godefroid_Piterman:gmc} is also related to our three-valued
semantics for \MTLdata.  It is, however, based on infinite words, not
observations (Definition~\ref{def:observation}).  Similar to
(T\ref{enum:observation_data}) of Definition~\ref{def:observation}, a
proposition with the truth value~$\unknown$ at a position can be
refined by $\true$ or $\false$.  In contrast, their semantics does not
support refinements that add and delete letters
(cf.~(T\ref{enum:observation_split})
and~(T\ref{enum:observation_removal}) of
Definition~\ref{def:observation}).

\subsubsection*{Data Values}

\citeetal{Havelund_etal:rvbook_data} overview and compare different
runtime-verification approaches that allow one to reason online about
data values in event streams. Among them are parametric
runtime-verification
approaches~\cite{Rosu_Chen:parametric_rv,Barringer_etal:qea} and
approaches that handle first-order extensions of temporal
logics~\cite{Basin_etal:rv_mfotl,Halle_Villemaire:fo_ltl}. Those
approaches share some similarities to our approach, in particular, how
the freeze quantifier is used to reason about data values.
As explained in Example~\ref{ex:mtl}, the freeze quantifier can be seen
as a weak form of the standard first-order quantifiers.
Although the first-order extensions are more expressive than \MTLdata,
the expressiveness of \MTLdata seems sufficient for many
runtime-verification applications because the data values often appear
uniquely in the events.  Handling specification languages with
first-order quantification like MFOTL~\cite{Basin_etal:rv_mfotl} in
settings with only partial knowledge and out-of-order event streams is
nontrivial and various restrictions seem to be
necessary~\cite{BasinKMZ-RV12}.
In a nutshell, in parametric runtime verification, one slices a single
event stream according to the events' data values in multiple streams,
which are then monitored separately and checked against propositional
specifications~\cite{Rosu_Chen:parametric_rv} or nonpropositional
specifications, as for instance, quantified event
automata~\cite{Barringer_etal:qea}.  The bindings of the data values
within the sliced event streams are implicit in most of those
approaches, and the slicing criteria is hard-coded in the monitoring
algorithm. In contrast, the freeze quantifier explicitly binds the
data values to logical variables.  Furthermore, our monitoring
algorithm for \MTLdata processes a single event stream.

\citeetal{Feng_etal:MTL_data} define a similar extension of MTL with
the freeze quantifier as in \MTLdata.  Their analysis focuses on the
computational complexity of the path-checking problem. However, they
use a finite trace semantics, which is less suitable for runtime
verification.
\citeetal{Brim_etal:stlstar},
and~\citeetal{Ryckbosch_Diwan:performancetraces} also provide
extensions of LTL with the freeze quantifier together with monitoring
algorithms.
Note that the rule-based runtime-verification approach
EAGLE~\cite{Barringer_etal:eagle} already allowed one, similar to the
freeze quantifier, to freeze data values in events to variables.
Neither \citeetal{Feng_etal:MTL_data}, \citeetal{Brim_etal:stlstar},
\citeetal{Ryckbosch_Diwan:performancetraces}, nor
\citeetal{Barringer_etal:eagle} consider out-of-order messages and
knowledge gaps in event streams.
Finally, \citeetal{Demri:2009} analyze the complexity of the
satisfiability problem of LTL extended with the freeze quantifier. In
particular, they provide translations of restricted fragments to
register automata. Applications to runtime verification are not
explored.

\subsubsection*{Distributed Systems}

Several runtime-verification approaches have been developed for
distributed systems.  \citeetal{Francalanza_etal:rvbook_dist} provide
an overview and we limit ourselves here to those approaches that are
closely related to ours.  Overall, all of them make different
assumptions on the system model and thus target different kinds of
distributed systems.  Furthermore, they handle different specification
languages.  We are not aware of any approach in the literature that
handles specifications with real-time constraints or accounts for
network failures.

\citeetal{ColomboFalcone:global_clock} propose a runtime-verification
approach, based on formula rewriting, that also allows the monitor to
receive messages out of order.  Their approach only handles the
propositional temporal logic LTL with the three-valued semantics
proposed by \citeetal{Bauer_etal:ltl_rv}.  In a nutshell, their
approach unfolds temporal connectives as time progresses and special
propositions act as placeholders for subformulas.  The subsequent
assignment of these placeholders to Boolean truth values triggers the
reevaluation and simplification of the formula.  Their approach only
guarantees soundness but not completeness, since the simplification
rules used for formula rewriting are incomplete.  Finally, its
performance with respect to out-of-order messages is not evaluated.

\citeetal{Sen_etal:decentralized_distributed_monitoring} use an LTL
variant with epistemic operators to express distributed knowledge.
The verdicts output by the monitors are correct with respect to the
local knowledge the monitors obtained about the systems' behavior.
Since their LTL variant only has temporal connectives that refer to
the past, only safety properties are expressible.
\citeetal{Scheffel_Schmitz:asynchronous_distributed_rv} extend this
work to handle also some liveness properties by working with a richer
fragment of LTL that includes temporal connectives that refer to the
future.  The algorithm by
\citeetal{Bauer_Falcone:decentralised_monitor} assumes a lock-step
semantics and thus only applies to synchronous systems.
\citeetal{Falcone_etal:decentralized_monitor_regular} weaken this
assumption. However, each component must still output its observations
at each time point, which is determined by a global clock.  The
observations are then received by the monitors at possibly later time
points.  The algorithm by
\citeetal{Mostafa_Bonakdarbour:decentralized_rv} assumes lossless FIFO
channels for asynchronous communication.  Logical clocks are used to
partially order messages.

\subsubsection*{Miscellaneous}

The problem of processing streams in which events may appear
out of order has also been considered in contexts other than runtime
verification, namely, in stream processing.  For example,
\citeetal{Srivastava_Widom:datastreams} use buffering and heartbeats
so that continuous queries are evaluated correctly under the
assumption that the heartbeats are sufficiently large.  Various
parameters are considered to generate the heartbeats.  However,
queries are not processed promptly, but always with a delay.
\citeetal{Li_etal:out-of-order} propose a stream-processing
architecture with a global mechanism that reports progress and allows
one to finalize a partial evaluation of a query on a time window.
Events are processed promptly.  The messages' sequence numbers, which
we use to determine whether the monitor may be missing a message from
a system component in some time period, can been seen as such a global
mechanism.

\section{Conclusion}
\label{sec:concl}

We have presented a runtime-verification approach based on three truth
values to checking real-time specifications given as \MTLdata
formulas.  Our approach targets distributed systems and handles the
practically relevant setting where the messages sent to monitors can
be delayed, reordered, or lost, and it provides soundness and
completeness guarantees.
Although our experimental evaluation is promising, our approach does
not yet scale to monitoring systems that generate thousands or even
millions of events per second.  This requires additional research,
including algorithmic optimizations.  We plan to investigate this in
future work, as well as to deploy and evaluate our approach in
realistic, large-scale case studies.

\begin{acks}
  This work received funding from the European Union's Horizon~2020
  research and innovation programme under the grant agreement No~779852.
\end{acks}



\end{document}